\documentclass[10pt,reqno,a4paper]{amsart}
\usepackage[british]{babel}
\usepackage{amssymb,amstext,amsthm,eucal,bbm,mathrsfs,amscd,bbold,pifont}
\usepackage[margin=1in]{geometry}
\usepackage{array}
\usepackage[svgnames,table]{xcolor}
\usepackage[T1]{fontenc}
\usepackage[utf8]{inputenc}
\usepackage{enumitem}
\usepackage[small]{eulervm}
\usepackage{tgpagella}
\usepackage{microtype}
\usepackage{booktabs} %for \toprule,\midrule,\bottomrule
\usepackage{caption}
\usepackage{tikz,pgfplots}
\usetikzlibrary{calc,arrows,shapes,cd}
\usepackage[unicode]{hyperref}
\hypersetup{%
  pdftitle   = {Spatially isotropic homogeneous spacetimes},
  pdfkeywords = {Lie algebras, kinematical, spacetime},
  pdfauthor  = {José Figueroa-O'Farrill, Stefan Prohazka},
  pdfcreator = {\LaTeX\ with package \flqq hyperref\frqq},
  linkcolor=NavyBlue,
  citecolor=ForestGreen,
  urlcolor=OrangeRed,
  anchorcolor=OrangeRed,
  colorlinks=true
}
\theoremstyle{definition}
\newtheorem{definition}{Definition}
\theoremstyle{plain}
\newtheorem{lemma}{Lemma}

\newtheorem{theorem}[lemma]{Theorem}

\newtheorem{criterion}{Criterion}

\theoremstyle{definition}

\newcommand{\eC}{\mathscr{C}}
\newcommand{\eL}{\mathscr{L}}
\newcommand{\eN}{\mathscr{N}}
\newcommand{\eQ}{\mathscr{Q}}
\newcommand{\eX}{\mathscr{X}}
\newcommand{\eJ}{\mathscr{J}}
\newcommand{\Ggr}{\mathcal{G}}
\newcommand{\Hgr}{\mathcal{H}}
\newcommand{\Kgr}{\mathcal{K}}
\newcommand{\Mgr}{\mathcal{M}}
\newcommand{\Ngr}{\mathscr{N}}
\newcommand{\g}{\mathfrak{g}}
\newcommand{\h}{\mathfrak{h}}
\renewcommand{\a}{\mathfrak{a}}
\renewcommand{\b}{\mathfrak{b}}
\renewcommand{\d}{\partial}

\newcommand{\gl}{\mathfrak{gl}}
\newcommand{\p}{\mathfrak{p}}

\renewcommand{\c}{\mathfrak{c}}
\newcommand{\e}{\mathfrak{e}}
\renewcommand{\r}{\mathfrak{r}}
\newcommand{\n}{\mathfrak{n}}

\newcommand{\so}{\mathfrak{so}}
\newcommand{\su}{\mathfrak{su}}

\renewcommand{\k}{\mathfrak{k}}

\newcommand{\s}{\mathfrak{s}}
\newcommand{\m}{\mathfrak{m}}
\renewcommand{\u}{\mathfrak{u}}
\renewcommand{\H}{H}

\ifdefined\C
\renewcommand{\C}{\boldsymbol{C}}
\else
\newcommand{\C}{\boldsymbol{C}}
\fi
\newcommand{\J}{\boldsymbol{J}}
\newcommand{\B}{\boldsymbol{B}}
\newcommand{\Pt}{\widetilde{\boldsymbol{P}}}
\newcommand{\Bt}{\widetilde{\boldsymbol{B}}}

\newcommand{\be}{\boldsymbol{e}}

\renewcommand{\P}{\boldsymbol{P}}
\newcommand{\V}{\boldsymbol{V}}

\renewcommand{\Re}{\operatorname{Re}}

\newcommand{\Bbar}{\overline{\B}}
\newcommand{\Pbar}{\overline{\P}}
\newcommand{\barH}{\overline{H}}
\newcommand{\barP}{\overline{P}}
\newcommand{\Xbar}{\overline{X}}
\newcommand{\Ybar}{\overline{Y}}
\newcommand{\Zbar}{\overline{Z}}

\renewcommand{\AA}{\mathbb{A}}
\newcommand{\EE}{\mathbb{E}}
\newcommand{\MM}{\mathbb{M}}
\newcommand{\PP}{\mathbb{P}}
\newcommand{\RR}{\mathbb{R}}
\newcommand{\ZZ}{\mathbb{Z}}
\renewcommand{\SS}{\mathbb{S}}
\newcommand{\HH}{\mathbb{H}}

\newcommand{\CC}{\mathbb{C}}

\newcommand{\Aut}{\operatorname{Aut}}
\newcommand{\GL}{\operatorname{GL}}
\newcommand{\Ad}{\operatorname{Ad}}
\newcommand{\SO}{\operatorname{SO}}
\newcommand{\SU}{\operatorname{SU}}
\newcommand{\Hom}{\operatorname{Hom}}

\newcommand{\RRnZ}{\RR^{\neq 0}}
\newcommand{\mink}{S1}
\newcommand{\ds}{S2}
\newcommand{\ads}{S3}
\newcommand{\euc}{S4}
\newcommand{\sph}{S5}
\newcommand{\hyp}{S6}
\newcommand{\gal}{S7}
\newcommand{\dsg}{S8}
\newcommand{\tdsg}{S9}
\newcommand{\adsg}{S10}
\newcommand{\tadsg}{S11}
\newcommand{\twodgal}{S12}
\newcommand{\car}{S13}
\newcommand{\dsc}{S14}
\newcommand{\adsc}{S15}
\newcommand{\flc}{S16}
\newcommand{\xone}{S17}
\newcommand{\xtwo}{S18}
\newcommand{\xthree}{S19}
\newcommand{\xfour}{S20}
\newcommand{\st}{A21}
\newcommand{\tst}{A22}
\newcommand{\athree}{A23}
\newcommand{\twoda}{A24}
\newcommand{\cm}{\checkmark}
\newcommand{\tm}{ }
\newcommand{\mn}{ }
\newcommand{\zro}{ }
\newcommand{\zLC}{\mathsf{LC}}
\newcommand{\zAdSC}{\mathsf{AdSC}}
\newcommand{\zdSC}{\mathsf{dSC}}
\newcommand{\zdS}{\mathsf{dS}}
\newcommand{\zAdS}{\mathsf{AdS}}
\newcommand{\zC}{\mathsf{C}}
\newcommand{\zG}{\mathsf{G}}
\newcommand{\zS}{\mathsf{S}}
\newcommand{\zTS}{\mathsf{TS}}
\newcommand{\zAdSG}{\mathsf{AdSG}}
\newcommand{\zdSG}{\mathsf{dSG}}
\newcommand{\ztAdSG}{\mathsf{AdSG}}
\newcommand{\ztdSG}{\mathsf{dSG}}
\definecolor{gris}{rgb}{0.5,0.5,0.5}
\newcommand{\zero}{{\color{gris}0}}
\usepackage{mathtools}
\newcommand{\choice}[2]{\substack{\mathrlap{#1}\\ \mathrlap{#2}}}
\allowdisplaybreaks[0]
\usepackage{graphicx}

\begin{document}

\title{Spatially isotropic homogeneous spacetimes}
\author[Figueroa-O'Farrill]{José Figueroa-O'Farrill}
\author[Prohazka]{Stefan Prohazka}
\address[JMF]{Maxwell Institute and School of Mathematics, The University
  of Edinburgh, James Clerk Maxwell Building, Peter Guthrie Tait Road,
  Edinburgh EH9 3FD, Scotland, United Kingdom}
\email{\href{mailto:j.m.figueroa@ed.ac.uk}{j.m.figueroa@ed.ac.uk}}
\address[SP]{Université Libre de Bruxelles and International Solvay Institutes,
  Physique Mathématique des Interactions Fondamentales, Campus Plaine
  - CP~231, B-1050 Bruxelles, Belgium, Europe}
\email{\href{mailto:stefan.prohazka@ulb.ac.be}{stefan.prohazka@ulb.ac.be}}
\begin{abstract}
  We classify simply-connected homogeneous ($D+1$)-dimensional
  spacetimes for kinematical and aristotelian Lie groups with
  $D$-dimensional space isotropy for all $D\geq 0$.  Besides
  well-known spacetimes like Minkowski and (anti) de Sitter we find
  several new classes of geometries, some of which exist only for
  $D=1,2$.  These geometries share the same amount of symmetry
  (spatial rotations, boosts and spatio-temporal translations) as the
  maximally symmetric spacetimes, but unlike them they do not
  necessarily admit an invariant metric.  We determine the possible
  limits between the spacetimes and interpret them in terms of
  contractions of the corresponding transitive Lie algebras. We
  investigate geometrical properties of the spacetimes such as whether
  they are reductive or symmetric as well as the existence of
  invariant structures (riemannian, lorentzian, galilean, carrollian,
  aristotelian) and, when appropriate, discuss the torsion and
  curvature of the canonical invariant connection as a means of
  characterising the different spacetimes.
\end{abstract}
\dedicatory{In memoriam Andrew Ranicki}
\thanks{EMPG-18-01}
\maketitle
\tableofcontents

\section{Introduction}
\label{sec:introduction}

\subsection{Motivation and contextualisation}
\label{sec:motiv-cont}

The laws of physics are to a good approximation invariant under
spatial rotations, spatio-temporal translations and inertial
transformations (\emph{boosts}).  This leads to space and
time homogeneity and space isotropy, but does not determine their
precise realisation completely.  This freedom is for example evident
from the way boosts act on space and time (e.g., compare the galilean
and Poincaré boosts) and through the existence of curvature (e.g.,
compare Minkowski and de Sitter spacetimes).  This leads to the
fundamental question~\cite{MR0238545}:
\begin{center}
  \emph{What are the possible kinematical symmetries of space and
    time?}
\end{center}
A partial mathematical answer to this question is the classification
of kinematical Lie algebras up to isomorphism, which started with the
seminal work of Bacry and Lévy-Leblond \cite{MR0238545} and of Bacry
and Nuyts \cite{MR857383}, who classified kinematical algebras (with
space isotropy) in the classical case of $3+1$ dimensions, and
culminated recently with a classification in arbitrary dimension using
techniques in deformation theory \cite{Figueroa-OFarrill:2017ycu,
  Figueroa-OFarrill:2017tcy, Andrzejewski:2018gmz}. The reason this
classification is only a partial answer is that the isomorphism type
of the Lie algebra is too coarse an invariant: it does not determine
uniquely the geometric realisation of the Lie algebra. The ur-example
is the Lorentz Lie algebra $\so(D+1,1)$, which acts transitively and
isometrically both on de~Sitter spacetime and on hyperbolic space in
$D+1$ dimensions, and, in what is possibly a new twist on an old tale,
we will see that it also acts transitively on a carrollian spacetime
of the same dimension.

The first step towards a complete answer to the fundamental question
was taken already in the original paper \cite{MR0238545} of Bacry and
Lévy-Leblond.  Although restricted to $3+1$ dimensions and to
spacetimes admitting parity and time-reversal transformations, they
already distinguish between the abstract Lie algebras and their
geometric realisations on homogeneous spacetimes, arriving at a list
of eleven possible kinematics.  Our more refined analysis in this
paper reduces that list to ten, since the para-galilean and static
kinematical Lie algebras lead to isomorphic homogeneous aristotelian
spacetimes.  In addition, we drop the requirement of parity or
time-reversal symmetries and we work in arbitrary (positive) dimension
$D+1$.

More precisely, in this paper we give a more complete answer to the
fundamental question by classifying the geometric realisations of
kinematical Lie algebras on simply-connected homogeneous spacetimes.
The classification we present in this paper, while encompassing the
classical geometries like (anti) de Sitter, Minkowski, galilean and
carrollian spacetimes and providing a way to systematically understand
their relations, will also uncover new spacetimes and their
connections to the ones just mentioned.

By comparison to the seminal works~\cite{MR0238545,MR857383}, the
novelty of our approach is predicated on the following features:
\begin{description}
\item[Geometry] Although our study departs from the classification of
  kinematical (and aristotelian) Lie algebras, our work focuses on
  classifying homogeneous spaces.  This is an important distinction
  because, as we will see (see, e.g.,
  Table~\ref{tab:LAs-to-spacetimes}), the same Lie algebra may act
  transitively on different spacetimes, while different Lie algebras
  may act transitively on the same spacetime.  For example, it follows
  from a careful analysis that despite there being a richer set of
  isomorphism classes of kinematical Lie algebras for $D=3$ than for
  $D>3$, there are no uniquely four-dimensional kinematical
  homogeneous spacetimes.
\item[Parity and time reversal] We relax the ``by no means compelling''
  restriction of parity and time-reversal invariance of the
  homogeneous spacetimes and, in so doing, we uncover novel kinematical
  spacetimes.  The possibility of dropping this restriction was
  already noted in \cite{MR0238545}
  and was dropped at the Lie algebraic level in \cite{MR857383}, where
  it was observed that every kinematical Lie algebra (with $D=3$) acts
  transitively on some four-dimensional homogeneous spacetime, but they
  stopped short of investigating the precise relationship between the
  Lie algebras and the homogeneous spacetimes.  As we will see, the
  relation is rather intricate, as illustrated, for example, in
  Table~\ref{tab:LAs-to-spacetimes}.
\item[Dimension] We go beyond (both above and below) $3+1$
  dimensions.  Our analysis is valid for any (positive) spacetime
  dimension $D+1$.  Whereas the case of $3+1$ dimensions turns out
  (after some detailed analysis) to be already generic, in low
  dimensions ($D \leq 2$) the situation is more involved and the
  classification of two- and three-dimensional homogeneous spacetimes
  differs markedly from that in generic dimension.
\end{description}

Homogeneous kinematical spacetimes are known to play an important rôle
in physics.  For example Minkowski and (anti) de Sitter spacetimes
are crucial in high energy physics, general relativity and
cosmology and many other spacetimes arise from them via
limits.  These limits often induce contractions of their
symmetry algebras.  It is therefore not surprising that they
too arise in various areas of physics.  Lie groups and their
homogeneous spaces have a plethora of applications
(representation theory, coadjoint orbits,...) and are ubiquitous in
physics (classical mechanics, hydrodynamics, cosmology, ...) and one
might, therefore, hope that the same is true for our (novel)
spacetimes. These might be hard to foresee at present, after all in
\cite{MR0238545} it was deemed that ``the physical interest of [the
Carroll groups] is very much reduced'', whereas there is no lack of
interest in the Carroll group at present.  Nevertheless we wish to
highlight interesting applications and areas where kinematical
spacetimes and their Lie algebras do arise:
\begin{description}
\item[Gauge/gravity duality] Anti de~Sitter spacetime (AdS) has
  been the focus of substantial interest due to the conjectured duality to
  conformal field theory (CFT)~\cite{Maldacena:1997re}.  Since AdS is a
  kinematical spacetime it might be tempting to see if this relation
  generalises to possible non-AdS/non-CFT dualities, especially since
  many of the kinematical spacetimes naturally arise as limits of
  AdS.  Indeed, one of the main motivations for the study of
  kinematical spacetimes is to explore possible new holographies
  beyond AdS/CFT.  Kinematical spacetimes might arise either as bulk
  spacetimes, similarly to Schrödinger~\cite{Son:2008ye, Balasubramanian:2008dm}
  and Lifshitz~\cite{Kachru:2008yh} spacetimes, or as geometries to
  which bulk theories couple~\cite{Christensen:2013lma,
    Christensen:2013rfa}.  For reviews see,
  e.g.,~\cite{Taylor:2015glc, Hartnoll:2016apf}.

  Furthermore, homogeneous spaces have already shown their usefulness
  in holographic setups beyond AdS/CFT, see e.g.,
  \cite{SchaferNameki:2009xr, Jottar:2010vp, Bagchi:2010xw,
    Duval:2012qr, Grosvenor:2017dfs} and therefore one might
  anticipate further interesting results based on analyses of our
  novel spacetimes.

\item[Condensed matter] Besides holographic applications,
  non-relativistic spaces and especially Newton--Cartan geometry has
  been a useful tool in the construction of effective field theories
  for quantum Hall states~\cite{Son:2013rqa, Geracie:2014nka,
    Jensen:2014aia} where coset constructions provide a systematic tool
  to implement symmetries (see, e.g.,~\cite{Brauner:2014jaa}).  It is often the
  case that the underlying symmetries are given by the centrally
  extended Galilei algebra: that is, the Bargmann algebra and in a
  follow-up paper we will present a classification of homogeneous
  spacetimes of Lie algebras generalising the Bargmann algebra.
 
\item[Cosmology] The classification presented in this paper can be
  understood as an extension of the classification of maximally
  symmetric lorentzian spacetimes. When imposing the restriction of
  the existence of an invariant lorentzian metric on the spacetimes we
  indeed recover the well known result that they consist of the (anti)
  de~Sitter and Minkowski spacetimes.\footnote{In general, the
    Friedmann–Lemaître–Robertson–Walker cosmologies do not correspond
    to homogeneous spaces since they are only homogeneous in space and
    not in time.} Dropping the assumption of the existence of a
  lorentzian metric, but keeping the same amount of symmetry and
  especially space isotropy and space and time homogeneity, basically
  leads to this more general classification. Like the lorentzian
  geometries, the other spacetimes represent empty universes and might
  be relevant for approximations of the de~Sitter universe, see e.g.,
  \cite{Aldrovandi:1998im}.
\item[Ultra-relativistic structures] The ``absolute space'' limit of
  lorentzian spacetimes leads to carrollian structures, where the
  metric is degenerate and space is absolute. It is closely related to
  the strong coupling limit of general
  relativity~\cite{Henneaux:1979vn}, arises as a limit of duality
  invariant theories~\cite{Bunster:2012hm} and has recently been
  connected to asymptotically flat spacetimes~\cite{Duval:2014uva}.
  The non-flat carrollian spacetimes have attracted less attention but
  might lead to interesting generalisations.
\item[New theories] Lie algebras and their associated spacetimes are a
  natural starting point for the construction of novel theories.
  Gauging the symmetries of various kinematical algebras (and their
  central extensions) has been investigated in recent years, e.g., in
  \cite{Andringa:2010it,Hartong:2015xda,Bergshoeff:2017btm}.
  Furthermore, a thorough analysis of connections and dynamical
  trajectories of some kinematical structures has been
  undertaken~\cite{Bekaert:2014bwa,Bekaert:2015xua}.
 
  A distinguished class of theories are those which are governed by an
  action principle.  Here $2+1$ dimensions seem especially fruitful,
  where theories based on Chern--Simons actions and (generalisations of) kinematical algebras
  have been constructed~\cite{Papageorgiou:2009zc,
    Papageorgiou:2010ud, Bergshoeff:2016lwr, Hartong:2016yrf,
    Bergshoeff:2016soe, Bergshoeff:2017btm,Joung:2018frr}.  Following the seminal
  work of \cite{LeBellac1973}, recently
  galilean~\cite{Duval:2014uoa,Bagchi:2014ysa,Festuccia:2016caf} and
  carrollian~\cite{Duval:2014uoa,Basu:2018dub} electrodynamics and
  gravity~\cite{Bergshoeff:2017btm,Hansen:2019pkl} and their possible
  action principles have been investigated. Given the new results of
  this work there remains much room for further explorations.
\end{description}

Having motivated our interest on kinematical Lie algebras and their
spacetimes, we now give a somewhat detailed overview of the contents
of the paper.

\subsection{Overview of results}
\label{sec:overview-results}

One of the main results in this paper is the classification of
(simply-connected) spacetimes which extend the class of maximally
symmetric lorentzian manifolds familiar from general relativity.  In
this section we will review what is already known about this
classification and will summarise how the results obtained in this
paper complete that picture.  Although in the paper we also consider
spacetimes which are unique to two and three dimensions, the bulk of
the discussion in this overview section will focus on those spacetimes which
exist in all dimensions; although we will mention at the end how these
results are modified in low dimension.

Our starting points are the de~Sitter spacetimes.  These are lorentzian
spacetimes which are locally isometric to quadric hypersurfaces in
pseudo-euclidean spaces.  Concretely, \emph{de~Sitter} and
\emph{anti~de~Sitter spacetimes} in $D+1$ dimensions with radius of
curvature $R$ are locally isometric, respectively, to the quadrics
\begin{equation}
  x_1^2 + x_2^2 + \cdots + x_D^2 + x_{D+1}^2 - x_{D+2}^2 = R^2
  \qquad\text{in $\RR^{D+1,1}$,}
\end{equation}
 and
\begin{equation}
  x_1^2 + x_2^2 + \cdots + x_D^2 - x_{D+1}^2 - x_{D+2}^2 = - R^2
  \qquad\text{in $\RR^{D,2}$.}
\end{equation}
More precisely, the de~Sitter spacetimes are the simply-connected
universal covers of these quadrics.  Taking the limit $R \to \infty$
is equivalent to the zero curvature limit in which we recover
Minkowski spacetime: the real affine space $\AA^{D+1}$ with a metric
which, when expressed relative to affine coordinates, is given by
\begin{equation}
  dx_1^2 + dx_2^2 + \cdots + dx_D^2 - c^2 dx_{D+1}^2,
\end{equation}
where we have introduced the speed of light $c$. We may now take the
\emph{non-relativistic limit} (on the co-metric) in which
$c\to \infty$ or the \emph{ultra-relativistic limit} in which
$c \to 0$. In the former case we arrive at the \emph{galilean
  spacetime}, whereas in the latter, we arrive at the \emph{carrollian
  spacetime} \cite{MR0192900,Duval:2014uoa}. These spacetimes are no
longer lorentzian: the (co\nobreakdash-)metric becomes degenerate in
the limit, leading to a galilean and a carrollian structure,
respectively. This is not to say that on the underlying manifold of
such spacetimes one could not define a lorentzian metric, but simply
that any such metric would not be invariant under the kinematical
symmetries of the spacetime, in the way that the Minkowski metric is
Poincaré invariant.

Geometrically, the ultra- and non-relativistic limits can be
understood in terms of what they do to the light cone present in the
tangent space at any point in Minkowski spacetime, as depicted in
Figure~\ref{fig:lightcones}, where we see that the light cone
collapses to a timelike line or a spacelike hyperplane in the ultra-
and non-relativistic limits, respectively. Since the tangent spaces in
a lorentzian manifold are lorentzian vector spaces containing their
own light cones, we can consider these limits not just for Minkowski
spacetime, but for any lorentzian manifold. In particular, we can do
this with the de~Sitter spacetimes. The non-relativistic limits of the
de~Sitter spacetimes are the \emph{galilean (anti) de~Sitter}
spacetimes (also known as the \emph{Newton--Hooke} or
\emph{non-relativistic
  cosmological} spacetimes)~\cite{MR0238545, Derome1972,
  Carinena:1981nq,Gibbons:2003rv}, whereas the ultra-relativistic
limits are the \emph{carrollian (anti) de~Sitter} spacetimes (also
known as \emph{para-euclidean} and \emph{para-Minkowski} spacetimes)
\cite{MR0238545}. As in the case of Minkowski spacetimes, the limiting
spacetimes are no longer lorentzian, but have galilean and
carrollian structures, respectively.

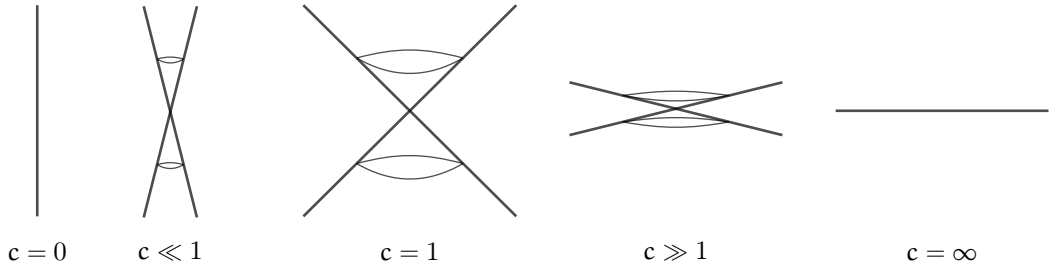
\begin{figure}[h!]
  \centering
  \begin{tikzpicture}[x=0.7cm,y=0.7cm]
    %
    % grid
    %
%    \draw [color=gray,step=1] (0,-2.5) grid (20,2.5);
    % 
    % coordinates
    %
    \coordinate (sw) at (6.,-1.);
    \coordinate (se) at (8.,-1.);
    \coordinate (ne) at (8.,1.);
    \coordinate (nw) at (6.,1.);
    \coordinate (usw) at (2.2503423410911236,-1.0203486648922608);
    \coordinate (use) at (2.7503423410911236,-1.0203486648922608); 
    \coordinate (une) at (2.7503423410911236,0.9796513351077392); 
    \coordinate (unw) at (2.2503423410911236,0.9796513351077392); 
    \coordinate (nsw) at (10.999892248241657,-0.21078008816633959);
    \coordinate (nse) at (12.999892248241657,-0.21078008816633959); 
    \coordinate (nne) at (12.999892248241657,0.28921991183366036);
    \coordinate (nnw) at (10.999892248241657,0.28921991183366036);
    %
    % light cones
    %
    \begin{scope}[line width=0.5pt,opacity=0.7,color=black]
      \draw (nw) to [out=330,in=210] (ne);
      \draw (sw) to [out=330,in=210] (se); 
      \draw (nw) to [out=15,in=165] (ne);      
      \draw (sw) to [out=15,in=165] (se);      
      \draw (unw) to [out=330,in=210] (une);
      \draw (usw) to [out=330,in=210] (use); 
      \draw (unw) to [out=15,in=165] (une);      
      \draw (usw) to [out=15,in=165] (use);      
      \draw (nnw) to [out=350,in=190] (nne);
      \draw (nsw) to [out=350,in=190] (nse);
      \draw (nnw) to [out=8,in=172] (nne);      
      \draw (nsw) to [out=8,in=172] (nse);      
    \end{scope}
    \begin{scope}[line width=1pt,opacity=0.7,color=black]
      \draw (5.,2.)-- (9.,-2.);
      \draw (9.,2.)-- (5.,-2.);
      \draw (9.999892248241657,0.5392199118336602)-- (13.999892248241657,-0.46078008816633964);
      \draw (13.999892248241657,0.5392199118336602) -- (9.999892248241657,-0.46078008816633964);
      \draw (3.0003423410911236,1.979651335107739)-- (2.0003423410911236,-2.020348664892261);
      \draw (2.0003423410911236,1.979651335107739)-- (3.0003423410911236,-2.020348664892261);
      \draw (15.,0.)-- (19.,0.);
      \draw (0.,2.)-- (0.,-2.);
    \end{scope}
    %
    % labels
    %
    \node at (7,-3) [above] {$c=1$};
    \node at (2.5,-3) [above] {$c\ll 1$};
    \node at (0,-3) [above] {$c=0$};
    \node at (12,-3) [above] {$c\gg 1$};
    \node at (17,-3) [above] {$c=\infty$};
  \end{tikzpicture}
  \caption{Effect on light cone of ultra- (left) and non-relativistic (right) limits.}
  \label{fig:lightcones}
\end{figure}

Just like Minkowski spacetime is the zero-curvature limit of the
de~Sitter spacetimes, the galilean (resp.~carrollian) spacetime can be
obtained as a zero-curvature limit of the galilean (resp.~carrollian)
(anti) de~Sitter spacetimes.  These spacetimes are not lorentzian and
thus, in contrast to the Minkowski case, the curvature being taken to
zero is not the Riemann curvature of a (non-existing) Levi-Civita
connection.  Indeed, as we will see, these spacetimes are homogeneous
spaces of kinematical Lie groups and as homogeneous spaces they are
reductive and symmetric and hence in possession of a canonical
torsion-free invariant connection.  It is that connection whose
curvature is being sent to zero.

% added in response to referee question
The galilean and carrollian symmetric spaces can also arise as limits
of the riemannian symmetric spaces.  The physical interpretation of
the riemannian analogues of the non- and ultra-relativistic limits of
euclidean space to galilean and carrollian spacetimes, respectively,
is not so clear.  There is no longer a light cone and hence no longer
a privileged timelike direction.  Nevertheless we may choose any
direction (all are equivalent, since the riemannian symmetric spaces
are isotropic) and rescale the metric along that direction or along
the perpendicular plane, and in this way arrive at the galilean and
carrollian spacetimes.  Neither the clock one-form in the galilean
spacetime nor the invariant vector field in the carrollian spacetime
are actually induced in the limit, so they have neither a minkowskian
nor a euclidean preferred interpretation.

The resulting picture (incomplete at this stage) is summarised in
Figure~\ref{fig:state-of-prior-art}, where
$\hyperlink{S1l}{\mathbb{M}}$, $\hyperlink{S13l}{\mathsf{C}}$ and
$\hyperlink{S7l}{\mathsf{G}}$ stand for Minkowski, carrollian and
galilean spacetimes, respectively; $\hyperlink{S4l}{\mathbb{E}}$,
$\hyperlink{S6l}{\mathbb{H}}$ and $\hyperlink{S5l}{\mathbb{S}}$ for
euclidean space, hyperbolic space and the round sphere, respectively;
\hyperlink{S3l}{$\mathsf{(A)dS}$} for (anti) de~Sitter spacetimes; and
$\hyperlink{S10l}{\mathsf{(A)dSG}}$ and \hyperlink{S15l}{$\mathsf{(A)dSC}$} for the galilean and
carrollian (anti) de~Sitter spacetimes, respectively. Diagonal arrows
are flat limits, whereas horizontal and vertical arrows are,
respectively, ultra- and non-relativistic limits.  Notice that whereas
the analogue of the ultra-relativistic limit of hyperbolic space is
carrollian AdS, the analogue of the non-relativistic limit is galilean
dS, and vice versa for the round sphere.

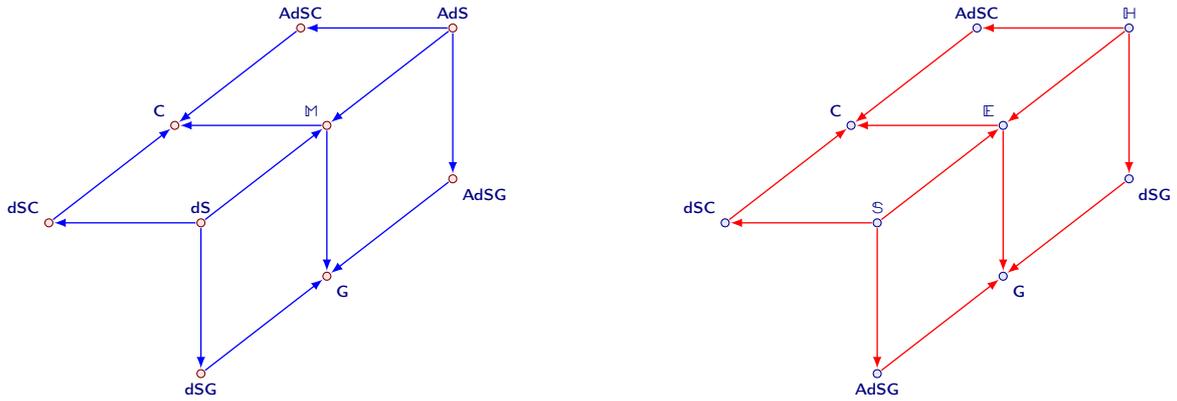
\begin{figure}[h!]
  \centering
  \begin{tikzpicture}[>=latex, shorten >=2pt, shorten <=2pt, x=1.0cm,y=1.0cm]
    %
    % grid
    %
%    \draw [color=gray, step] (-2,-3) grid (10,3);
    % 
    % vertices
    % 
    \coordinate [label=above:{\tiny $\hyperlink{S2l}{\mathsf{dS}}$}] (ds) at (5.688048519056286,-0.5838170592960186);
    \coordinate [label=below:{\tiny $\hyperlink{S8l}{\mathsf{dSG}}$}] (dsg) at (5.688048519056286, -2.5838170592960186);
    \coordinate [label=above left:{\tiny $\hyperlink{S14l}{\mathsf{dSC}}$}] (dsc) at  (3.688048519056286,-0.5838170592960186);
    \coordinate [label=above left:{\tiny $\hyperlink{S13l}{\mathsf{C}}$}] (c) at (5.344024259528143,0.7080914703519907);
    \coordinate [label=above left:{\tiny $\hyperlink{S1l}{\mathbb{M}}$}] (m) at (7.344024259528143,0.7080914703519907);
    \coordinate [label=below right:{\tiny $\hyperlink{S7l}{\mathsf{G}}$}] (g) at  (7.344024259528143, -1.2919085296480093);
    \coordinate [label=above:{\tiny $\hyperlink{S3l}{\mathsf{AdS}}$}] (ads) at (9,2);
    \coordinate [label=below right:{\tiny $\hyperlink{S10l}{\mathsf{AdSG}}$}] (adsg) at (9,0);
    \coordinate [label=above:{\tiny $\hyperlink{S15l}{\mathsf{AdSC}}$}]  (adsc) at (7,2);
    % 
    % edges
    % 
    \draw [->,line width=0.5pt,color=blue] (adsc) -- (c);
    \draw [->,line width=0.5pt,color=blue] (dsc) -- (c);
    \draw [->,line width=0.5pt,color=blue] (ads) -- (m);
    \draw [->,line width=0.5pt,color=blue] (adsg) -- (g);
    \draw [->,line width=0.5pt,color=blue] (dsg) -- (g);
    \draw [->,line width=0.5pt,color=blue] (ds) -- (m);
    \draw [->,line width=0.5pt,color=blue] (ds) -- (dsc);
    \draw [->,line width=0.5pt,color=blue] (ds) -- (dsg);
    \draw [->,line width=0.5pt,color=blue] (m) -- (c);
    \draw [->,line width=0.5pt,color=blue] (m) -- (g);
    \draw [->,line width=0.5pt,color=blue] (ads) -- (adsc);
    \draw [->,line width=0.5pt,color=blue] (ads) -- (adsg);
    % 
    % points
    % 
    \foreach \point in {ads,ds,adsc,dsc,adsg,dsg,m,g,c}
    \filldraw [color=red!50!black,fill=red!10!white] (\point) circle (1.5pt);
  \end{tikzpicture}\hspace{2cm}
  \begin{tikzpicture}[>=latex, shorten >=2pt, shorten <=2pt, x=1.0cm,y=1.0cm]
    %
    % grid
    %
%    \draw [color=gray, step] (-2,-3) grid (10,3);
    % 
    % vertices
    % 
    \coordinate [label=above:{\tiny $\hyperlink{S5l}{\mathbb{S}}$}] (sph) at (5.688048519056286,-0.5838170592960186);
    \coordinate [label=below:{\tiny $\hyperlink{S10l}{\mathsf{AdSG}}$}] (adsg) at (5.688048519056286, -2.5838170592960186);
    \coordinate [label=above left:{\tiny $\hyperlink{S14l}{\mathsf{dSC}}$}] (dsc) at  (3.688048519056286,-0.5838170592960186);
    \coordinate [label=above left:{\tiny $\hyperlink{S13l}{\mathsf{C}}$}] (c) at (5.344024259528143,0.7080914703519907);
    \coordinate [label=above left:{\tiny $\hyperlink{S4l}{\mathbb{E}}$}] (e) at (7.344024259528143,0.7080914703519907);
    \coordinate [label=below right:{\tiny $\hyperlink{S7l}{\mathsf{G}}$}] (g) at  (7.344024259528143, -1.2919085296480093);
    \coordinate [label=above:{\tiny $\hyperlink{S6l}{\mathbb{H}}$}] (h) at (9,2);
    \coordinate [label=below right:{\tiny $\hyperlink{S8l}{\mathsf{dSG}}$}] (dsg) at (9,0);
    \coordinate [label=above:{\tiny $\hyperlink{S15l}{\mathsf{AdSC}}$}]  (adsc) at (7,2);
    % 
    % edges
    % 
    \draw [->,line width=0.5pt,color=red] (adsc) -- (c);
    \draw [->,line width=0.5pt,color=red] (dsc) -- (c);
    \draw [->,line width=0.5pt,color=red] (h) -- (e);
    \draw [->,line width=0.5pt,color=red] (adsg) -- (g);
    \draw [->,line width=0.5pt,color=red] (dsg) -- (g);
    \draw [->,line width=0.5pt,color=red] (sph) -- (e);
    \draw [->,line width=0.5pt,color=red] (sph) -- (dsc);
    \draw [->,line width=0.5pt,color=red] (sph) -- (adsg);
    \draw [->,line width=0.5pt,color=red] (e) -- (c);
    \draw [->,line width=0.5pt,color=red] (e) -- (g);
    \draw [->,line width=0.5pt,color=red] (h) -- (adsc);
    \draw [->,line width=0.5pt,color=red] (h) -- (dsg);
    % 
    % points
    % 
    \foreach \point in {h,sph,adsc,dsc,adsg,dsg,e,g,c}
    \filldraw [color=blue!50!black,fill=blue!10!white] (\point) circle (1.5pt);
    %
    % legend
    %
  \end{tikzpicture}
  \caption{Maximally symmetric spaces and their limits:
    non-relativistic (vertical), ultra-relativistic (horizontal) and
    flat (diagonal).}
  \label{fig:state-of-prior-art}
\end{figure}

One of the main results in this paper is to complete this picture to
the one illustrated by Figure~\ref{fig:generic-d-graph} (and also
Figures~\ref{fig:d=3-graph} and \ref{fig:d=2-graph} for lower
dimension) which includes the (simply-connected) homogeneous
spacetimes of all the kinematical Lie groups, with the exception of
the riemannian maximally symmetric spaces, whose inclusion might
obscure more than enlighten.  Much of this paper is devoted to
explaining that picture and describing how to arrive at it, but for
now let us describe briefly its salient features:
\begin{itemize}
\item The galilean de~Sitter spacetime \hyperlink{S8l}{$\zdSG$} is actually the
  unique symmetric point in a one-parameter family
  \hyperlink{S9l}{$\zdSG_\gamma$}, with $\gamma \in [-1,1]$, of reductive
  homogeneous spaces, distinguished by the torsion of the canonical
  connection, which vanishes at the symmetric point $\gamma = -1$.
\item Similarly, the galilean anti de~Sitter spacetime \hyperlink{S10l}{$\zAdSG$} is the
  unique symmetric point in a one-parameter family
  \hyperlink{S11l}{$\zAdSG_\chi$}, with $\chi \geq 0$, of reductive
  homogeneous spaces, distinguished by the torsion of the canonical
  connection, which vanishes at the symmetric point $\chi = 0$.
  Moreover, $\hyperlink{S11l}{\zAdSG_\infty} := \lim_{\chi\to \infty}
  \hyperlink{S11l}{\zAdSG_\chi} = \hyperlink{S9l}{\zdSG_1}$.
\item There is a non-reductive homogeneous spacetime $\hyperlink{S17l}{\zLC}$ of
  $\SO(D+1,1)$ with an invariant carrollian structure admitting a
  limit to the carrollian spacetime.  This will be shown to be
  isomorphic to the (future) light cone in Minkowski spacetime in one
  dimension higher, hence the notation.
\item There are several aristotelian homogeneous spacetimes, which are
  spacetimes without boosts:
  \begin{itemize}
  \item the static affine spacetime $\hyperlink{A21}{\mathsf{S}}$ to which all other
    spacetimes have limits,
  \item a torsional aristotelian spacetime $\hyperlink{A22}{\mathsf{TS}}$ corresponding to the
    group manifold of a non-abelian solvable Lie group, and
  \item the Einstein static universe
    \hyperlink{A23m}{$\RR \times \SS^D$} and its hyperbolic version
    \hyperlink{A23p}{$\RR \times \HH^D$}, which do not arise from
    kinematical groups for $D\neq 3$.
  \end{itemize}
\end{itemize}
In addition, although not depicted in
Figure~\ref{fig:generic-d-graph}, there are also the riemannian
maximally symmetric spaces (sphere \hyperlink{S5l}{$\SS^{D+1}$},
hyperbolic \hyperlink{S6l}{$\HH^{D+1}$} and euclidean
\hyperlink{S4l}{$\EE^{D+1}$}), whose rôle as spacetimes, due to
their compact ``boosts'', is questionable.

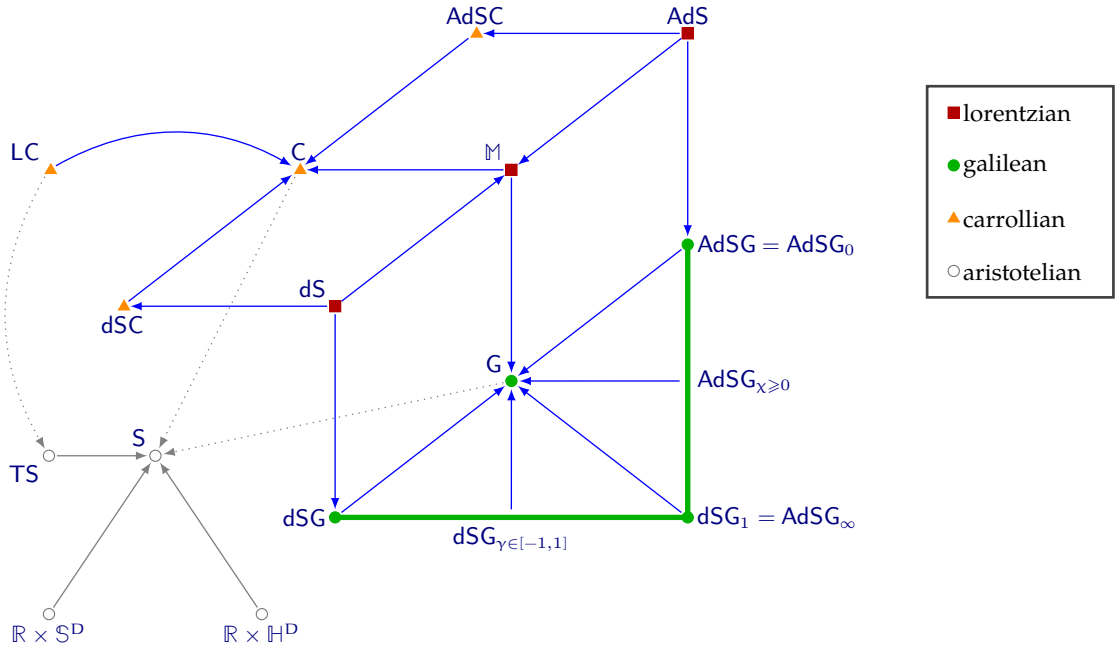
\begin{figure}[h!]
  \centering
  \begin{tikzpicture}[scale=1.4,>=latex, shorten >=3pt, shorten <=3pt, x=1.0cm,y=1.0cm]
    %
    % grid
    %
%   \draw [color=gray,step=.5] (3,-5) grid (14,3);
    % 
    % vertices
    % 
    \coordinate [label=above left:{\hyperlink{S2l}{\small $\zdS$}}] (ds) at (5.688048519056286,-0.5838170592960186);
    \coordinate [label=left:{\hyperlink{S8l}{\small $\zdSG$}}] (dsg) at (5.688048519056286, -2.5838170592960186);
    \coordinate [label=right:{\hyperlink{S9l}{\small $\ztdSG_1 = \ztAdSG_\infty$}}] (dsgone) at (9, -2.5838170592960186);
    \coordinate [label=below:{\hyperlink{S14l}{\small $\zdSC$}}] (dsc) at  (3.688048519056286,-0.5838170592960186);
    \coordinate [label=above:{\hyperlink{S13l}{\small $\zC$}}] (c) at (5.344024259528143, 0.7080914703519907);
    \coordinate [label=above left:{\hyperlink{S1l}{\small $\MM$}}] (m) at (7.344024259528143,0.7080914703519907);
    \coordinate [label=above left:{\hyperlink{S7l}{\small $\zG$}}] (g) at  (7.344024259528143, -1.2919085296480093);
    \coordinate [label=above:{\hyperlink{S3l}{\small $\zAdS$}}] (ads) at (9,2);
    \coordinate [label=right:{\hyperlink{S10l}{\small $\zAdSG=\ztAdSG_0$}}] (adsg) at (9,0);
    \coordinate [label=above:{\hyperlink{S15l}{\small $\zAdSC$}}]  (adsc) at (7,2);
    \coordinate [label=above left:{\hyperlink{A21}{\small $\zS$}}] (s) at (4, -2);
    \coordinate [label=above left:{\hyperlink{S16l}{\small $\zLC$}}]  (flc) at (3, 0.7080914703519907);
    \coordinate [label=below left:{\hyperlink{A22}{\small $\zTS$}}]  (ts) at (3, -2); 
    \coordinate [label=below:{\hyperlink{A23p}{\small $\RR\times\SS^D$}}]  (esu) at (3, -3.5); 
    \coordinate [label=below:{\hyperlink{A23m}{\small $\RR\times\HH^D$}}]  (hesu) at (5, -3.5); 
    % 
    % labels
    %
    \coordinate [label=below:{\hyperlink{S9l}{\small $\ztdSG_{\gamma\in[-1,1]}$}}] (tdsg) at (7.344024259528143, -2.5838170592960186);
    \coordinate [label=right:{\hyperlink{S11l}{\small $\ztAdSG_{\chi\geq0}$}}] (tadsg) at (9, -1.2919085296480093);
    % 
    % additional edges
    % 
    \draw [->,line width=0.5pt,dotted,color=gray] (c) -- (s);
    \draw [->,line width=0.5pt,dotted,color=gray] (g) -- (s);
    \draw [->,line width=0.5pt,color=blue] (dsgone) -- (g);
    \draw [->,line width=0.5pt,color=blue] (7.344024259528143,-2.5838170592960186) -- (g);
    \draw [->,line width=0.5pt,color=blue] (9, -1.2919085296480093) -- (g);
    % 
    % edges
    % 
    \draw [->,line width=0.5pt,color=blue] (adsc) -- (c);
    \draw [->,line width=0.5pt,color=blue] (dsc) -- (c);
    \draw [->,line width=0.5pt,color=blue] (ads) -- (m);
    \draw [->,line width=0.5pt,color=blue] (adsg) -- (g);
    \draw [->,line width=0.5pt,color=blue] (dsg) -- (g);
    \draw [->,line width=0.5pt,color=blue] (ds) -- (m);
    \draw [->,line width=0.5pt,color=blue] (ds) -- (dsc);
    \draw [->,line width=0.5pt,color=blue] (ds) -- (dsg);
    \draw [->,line width=0.5pt,color=blue] (m) -- (c);
    \draw [->,line width=0.5pt,color=blue] (m) -- (g);
    \draw [->,line width=0.5pt,color=blue] (ads) -- (adsc);
    \draw [->,line width=0.5pt,color=blue] (ads) -- (adsg);
    %
    % new edges
    %
    \draw [->,line width=0.5pt,color=blue] (flc) to [out=30,in=150] (c);
    \draw [->,line width=0.5pt,dotted,color=gray] (flc) to [out=240,in=120] (ts);
    \draw [->,line width=0.5pt,color=gray] (ts) to (s); 
    \draw [->,line width=0.5pt,color=gray] (esu) to (s); 
    \draw [->,line width=0.5pt,color=gray] (hesu) to (s); 
    %
    % points
    %
    % galilean
    % continua
    \begin{scope}[>=latex, shorten >=0pt, shorten <=0pt, line width=2pt, color=green!70!black]
      % \draw [->,shorten >=5pt] (adsg) --(dsgone);
      \draw (adsg) --(dsgone);
      \draw (dsg) -- (dsgone);% node {$\boldsymbol{]}$};
    \end{scope}
    \foreach \point in {g,adsg,dsg,dsgone}
    \filldraw [color=green!70!black,fill=green!70!black] (\point) circle (1.5pt);
    % lorentzian
    \foreach \point in {ads,ds,m}
    \filldraw [color=red!70!black,fill=red!70!black] (\point) ++(-1.5pt,-1.5pt) rectangle ++(3pt,3pt);
    % carrollian
    \foreach \point in {adsc,dsc,flc,c}
    \filldraw [color=DarkOrange,fill=DarkOrange] (\point) ++(-1pt,-1pt) -- ++(3pt,0pt) -- ++(-1.5pt,2.6pt) -- cycle;
    % aristotelian
    \foreach \point in {s,ts,esu,hesu}
    \draw [color=gray!90!black] (\point) circle (1.5pt);
    %
    % legend
    %
    \begin{scope}[xshift=0.5cm]
    \draw [line width=1pt,color=gray!50!black] (10.75,-0.5) rectangle (12.5,1.5);
    \filldraw [color=red!70!black,fill=red!70!black] (11,1.25) ++(-1.5pt,-1.5pt) rectangle ++(3pt,3pt) ; 
    \draw (11,1.25) node[color=black,anchor=west] {\small lorentzian}; 
    \filldraw [color=green!70!black,fill=green!70!black] (11,0.75) circle (1.5pt) node[color=black,anchor=west] {\small galilean};
    \filldraw [color=DarkOrange,fill=DarkOrange] (11,0.25) ++(-1.5pt,-1pt) -- ++(3pt,0pt) -- ++(-1.5pt,2.6pt) -- cycle;
    \draw (11,0.25) node[color=black,anchor=west] {\small carrollian};
    \draw [color=gray!90!black] (11,-0.25) circle (1.5pt) node[color=black,anchor=west] {\small aristotelian};
    \end{scope}
  \end{tikzpicture}
  \caption{Homogeneous spacetimes in dimension $D+1 \geq 4$ and their limits.}
  \label{fig:generic-d-graph}
\end{figure}

When $D=2$, Figure~\ref{fig:generic-d-graph} is slightly modified to
Figure~\ref{fig:d=3-graph}. From that figure we see that all that
happens now is that there is a new aristotelian spacetime (\hyperlink{A24}{\twoda}) and
a new two-parameter family (\hyperlink{S12l}{\twodgal$_{\gamma, \chi}$}) of galilean
spacetimes interpolating between the torsional galilean (anti)
de~Sitter spacetimes. There are limits from every spacetime in that
family to the galilean spacetime. This figure also omits the
riemannian maximally symmetric spaces.

\begin{figure}[h!]
  \centering
  \begin{tikzpicture}[scale=1.4,>=latex, shorten >=3pt, shorten <=3pt,
    x=1.0cm,y=1.0cm,line join=bevel]
    %
    % grid
    %
%    \draw [color=gray,step=.5] (0,-4) grid (12,3);
    % 
    % vertices
    % 
    \coordinate [label=above left:{\hyperlink{S2l}{\small $\zdS$}}] (ds) at (5.688048519056286,-0.5838170592960186);
    \coordinate [label=left:{\hyperlink{S8l}{\small $\zdSG$}}] (dsg) at (5.688048519056286, -2.5838170592960186);
    \coordinate [label=right:{\hyperlink{S9l}{\small $\ztdSG_1 = \ztAdSG_\infty$}}] (dsgone) at (9, -2.5838170592960186);
    \coordinate [label=below:{\hyperlink{S14l}{\small $\zdSC$}}] (dsc) at  (3.688048519056286,-0.5838170592960186);
    \coordinate [label=above:{\hyperlink{S13l}{\small $\zC$}}] (c) at (5.344024259528143, 0.7080914703519907);
    \coordinate [label=above left:{\hyperlink{S1l}{\small $\MM$}}] (m) at (7.344024259528143,0.7080914703519907);
    \coordinate [label=above left:{\hyperlink{S7l}{\small $\zG$}}] (g) at  (7.344024259528143, -1.2919085296480093);
    \coordinate [label=above:{\hyperlink{S3l}{\small $\zAdS$}}] (ads) at (9,2);
    \coordinate [label=right:{\hyperlink{S10l}{\small $\zAdSG =  \ztAdSG_0$}}] (adsg) at (9,0);
    \coordinate [label=above:{\hyperlink{S15l}{\small $\zAdSC$}}]  (adsc) at (7,2);
    \coordinate [label=above left:{\hyperlink{A21}{\small $\zS$}}] (s) at (4, -2);
    \coordinate [label=above left:{\hyperlink{S16l}{\small $\zLC$}}]  (flc) at (3, 0.7080914703519907);
    \coordinate [label=below left:{\hyperlink{A22}{\small $\zTS$}}]  (ts) at (3, -2);
    \coordinate [label=below:{\hyperlink{A23p}{\small $\RR\times\SS^2$}}]  (esu) at (3, -3.5); 
    \coordinate [label=below:{\hyperlink{A23m}{\small $\RR\times\HH^2$}}]  (hesu) at (5, -3.5); 
    \coordinate [label=below:{\hyperlink{A24}{\small \twoda}}]  (twoda) at (4, -3.5); 
    % 
    % labels
    %
    \coordinate [label=below:{\hyperlink{S9l}{\small $\ztdSG_{\gamma\in[-1,1]}$}}] (tdsg) at (7.344024259528143, -2.5838170592960186);
    \coordinate [label=right: {\hyperlink{S11l}{\small $\ztAdSG_{\frac{2}{\chi}}$}}] (tadsg) at (9, -1.2919085296480093);
    % 
    % additional edges
    % 
    \draw [->,line width=0.5pt,dotted,color=gray] (c) -- (s);
    \draw [->,line width=0.5pt,dotted,color=gray] (g) -- (s);
    % 
    % edges
    % 
    \draw [->,line width=0.5pt,color=blue] (adsc) -- (c);
    \draw [->,line width=0.5pt,color=blue] (dsc) -- (c);
    \draw [->,line width=0.5pt,color=blue] (ads) -- (m);
    \draw [->,line width=0.5pt,color=blue] (adsg) -- (g);
    \draw [->,line width=0.5pt,color=blue] (dsg) -- (g);
    \draw [->,line width=0.5pt,color=blue] (ds) -- (m);
    \draw [->,line width=0.5pt,color=blue] (ds) -- (dsc);
    \draw [->,line width=0.5pt,color=blue] (ds) -- (dsg);
    \draw [->,line width=0.5pt,color=blue] (m) -- (c);
    \draw [->,line width=0.5pt,color=blue] (m) -- (g);
    \draw [->,line width=0.5pt,color=blue] (ads) -- (adsc);
    \draw [->,line width=0.5pt,color=blue] (ads) -- (adsg);
    %
    % new edges
    %
    \draw [->,line width=0.5pt,color=blue] (flc) to [out=30,in=150] (c);
    \draw [->,line width=0.5pt,dotted,color=gray] (flc) to [out=240,in=120] (ts);
    \draw [->,line width=0.5pt,color=gray] (ts) to (s);
    \draw [->,line width=0.5pt,color=gray] (esu) to (s); 
    \draw [->,line width=0.5pt,color=gray] (hesu) to (s); 
    \draw [->,line width=0.5pt,color=gray] (twoda) to (s); 
    %
    % points
    %
    % galilean
    % continua
    \begin{scope}[line width=2pt, color=green!70!black]
      \filldraw [color=green!70!black, fill=green!15!white] (dsgone)
      -- (dsg) to [out=20,in=250] (adsg) -- (dsgone);
    \end{scope}
    \coordinate [label=right:{\hyperlink{S12}{\small \twodgal$_{\gamma, \chi}$}}] (tdg) at (7.5,-2.25); 
    \draw [->, line width=0.5pt,color=blue] (dsgone) -- (g);
    \draw [->, line width=0.5pt,color=blue] (tdg) -- (g);
    \draw [->, line width=0.5pt,color=blue] (8.61128, -1.4101) -- (g);
    \foreach \point in {g,adsg,dsg,dsgone}
    \filldraw [color=green!70!black,fill=green!70!black] (\point) circle (1.5pt);
    % lorentzian
    \foreach \point in {ads,ds,m}
    \filldraw [color=red!70!black,fill=red!70!black] (\point) ++(-1.5pt,-1.5pt) rectangle ++(3pt,3pt);
    % carrollian
    \foreach \point in {adsc,dsc,flc,c}
    \filldraw [color=DarkOrange,fill=DarkOrange] (\point) ++(-1pt,-1pt) -- ++(3pt,0pt) -- ++(-1.5pt,2.6pt) -- cycle;
    % aristotelian
    \foreach \point in {s,ts,esu,hesu,twoda}
    \filldraw [color=gray!90!black,fill=gray!10!white] (\point) circle (1.5pt);
    %
    % legend
    %
    \begin{scope}[xshift=0.5cm]
    \draw [line width=1pt,color=gray!50!black] (10.75,-0.5) rectangle (12.5,1.5);
    \filldraw [color=red!70!black,fill=red!70!black] (11,1.25) ++(-1.5pt,-1.5pt) rectangle ++(3pt,3pt) ; 
    \draw (11,1.25) node[color=black,anchor=west] {\small lorentzian}; 
    \filldraw [color=green!70!black,fill=green!70!black] (11,0.75) circle (1.5pt) node[color=black,anchor=west] {\small galilean};
    \filldraw [color=DarkOrange,fill=DarkOrange] (11,0.25) ++(-1.5pt,-1pt) -- ++(3pt,0pt) -- ++(-1.5pt,2.6pt) -- cycle;
    \draw (11,0.25) node[color=black,anchor=west] {\small carrollian};
    \draw [color=gray!90!black] (11,-0.25) circle (1.5pt) node[color=black,anchor=west] {\small aristotelian};
    \end{scope}
  \end{tikzpicture}
  \caption{Three-dimensional homogeneous spacetimes and their limits.}
  \label{fig:d=3-graph}
\end{figure}
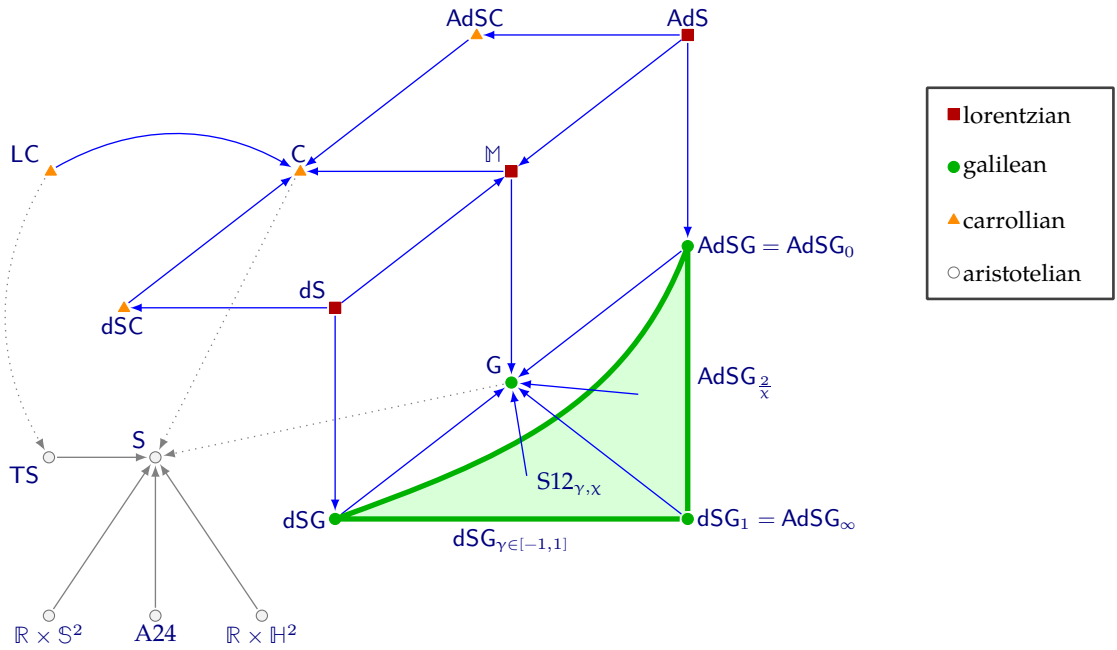

Finally, when $D=1$, Figure~\ref{fig:generic-d-graph} is also
modified.  Now there are accidental pairwise isomorphisms between some
of the symmetric spacetimes due to the possibility of redefining what
we mean by space and time.  In addition there are new two-dimensional
spacetimes with no discernible structure: two spacetimes (\hyperlink{S17l}{\xone} and
\hyperlink{S18l}{\xtwo}) and two continua
(\hyperlink{S19l}{\xthree$_\chi$} and
\hyperlink{S20l}{\xfour$_\chi$}).  The resulting picture is depicted
in Figure~\ref{fig:d=2-graph}.

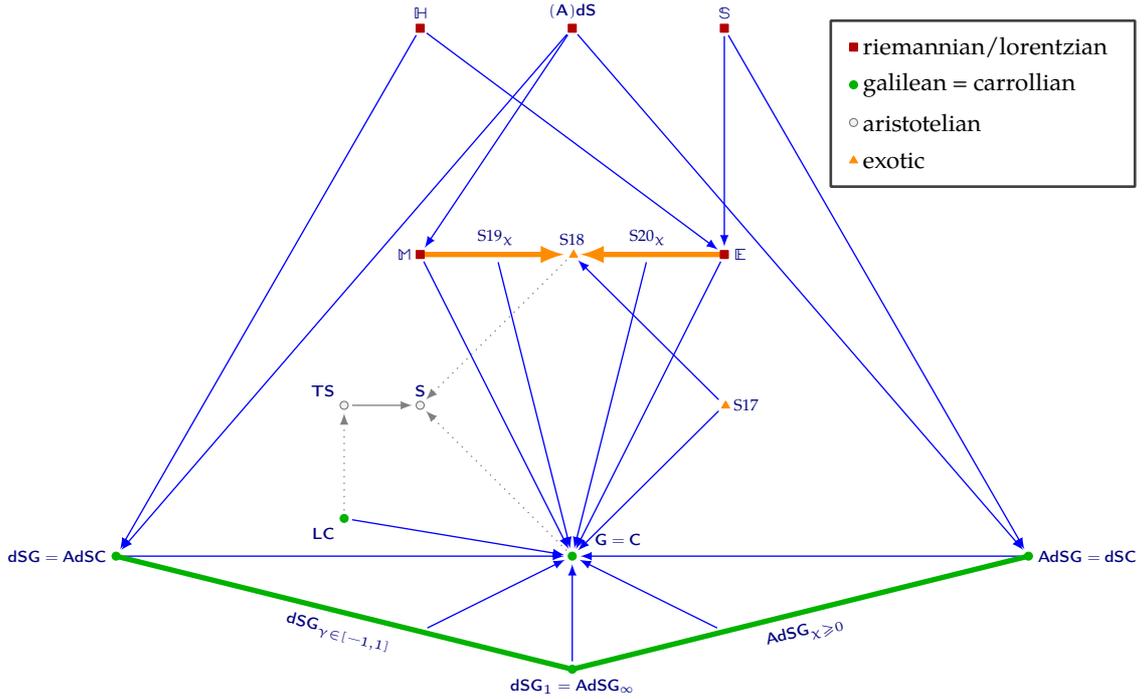
\begin{figure}[h!]
  \centering
  \begin{tikzpicture}[scale=1,>=latex, shorten >=3pt, shorten <=3pt,
    x=1.0cm,y=1.0cm,line join=bevel]
    %
    % grid
    %
    % \draw [color=gray,step=.5] (-8,-1) grid (8,11);
    % 
    % vertices
    % 
    \coordinate [label=below:{\hyperlink{S9l}{\tiny $\ztdSG_1 = \ztAdSG_\infty$}}] (dsgone) at (0,1.5);
    \coordinate [label=left:{\hyperlink{S8l}{\tiny $\zdSG = \zAdSC$}}]  (dsg) at (-6,3); 
    \coordinate [label=right:{\hyperlink{S10l}{\tiny $\zAdSG = \zdSC$}}] (adsg) at (6,3); 
    \coordinate [label={[shift={(0.6,0.02)}]{\hyperlink{S7l}{\tiny $\zG=\zC$}}}] (g) at  (0,3); 
    \coordinate [label=above:{\hyperlink{S18l}{\tiny \xtwo}}] (xtwo) at (0,7);    
    \coordinate [label=above:{\hyperlink{S3l}{\tiny $\mathsf{(A)dS}$}}] (ads) at (0,10); 
    \coordinate [label=above:{\hyperlink{A21}{\tiny $\zS$}}] (s) at (-2, 5); 
    \coordinate [label=below left:{\hyperlink{S16l}{\tiny $\zLC$}}]  (flc) at (-3,3.5); 
    \coordinate [label=above left:{\hyperlink{A22}{\tiny $\zTS$}}]  (ts) at (-3, 5); 
    \coordinate [label=right:{\hyperlink{S17l}{\tiny \xone}}] (xone) at (2,5); 
    \coordinate [label=left:{\hyperlink{S1l}{\tiny $\MM$}}] (m) at (-2,7); 
    \coordinate [label=right:{\hyperlink{S4l}{\tiny $\EE$}}] (e) at (2,7);     
    \coordinate [label=above:{\hyperlink{S5l}{\tiny $\SS$}}] (sph) at (2,10);
    \coordinate [label=above:{\hyperlink{S6l}{\tiny $\HH$}}] (hyp) at (-2,10);
    % % 
    % % labels
    % %
    \coordinate [label=above:{\hyperlink{S19l}{\tiny  \xthree$_\chi$}}] (xthree) at (-1,7); 
    \coordinate [label=above:{\hyperlink{S20l}{\tiny  \xfour$_\chi$}}] (xfour) at (1,7); 
    % % 
    % % additional edges
    % % 
    \draw [->,line width=0.5pt,dotted,color=gray] (g) -- (s);
    % % 
    % % edges
    % % 
    \draw [->,line width=0.5pt,color=blue] (xone) -- (g); 
    \draw [->,line width=0.5pt,color=blue] (xone) -- (xtwo); 
    \draw [->,line width=0.5pt,color=blue] (m) -- (g);
    \draw [->,line width=0.5pt,color=blue] (e) -- (g);
    \draw [->,line width=0.5pt,color=blue] (adsg) -- (g);
    \draw [->,line width=0.5pt,color=blue] (dsg) -- (g);
    \draw [->,line width=0.5pt,color=blue] (ads) -- (m);
    \draw [->,line width=0.5pt,color=blue] (ads) -- (dsg);
    \draw [->,line width=0.5pt,color=blue] (ads) -- (adsg); 
    \draw [->,line width=0.5pt,color=blue] (sph) -- (e); 
    \draw [->,line width=0.5pt,color=blue] (hyp) -- (e); 
    \draw [->,line width=0.5pt,color=blue] (sph) -- (adsg); 
    \draw [->,line width=0.5pt,color=blue] (hyp) -- (dsg); 
    % %
    % % new edges
    % %
    \draw [->,line width=0.5pt,color=blue] (flc) to (g);
    \draw [->, line width=0.5pt,dotted,color=gray] (flc) to (ts);
    \draw [->,line width=0.5pt,color=gray] (ts) to (s);
    \draw [->, line width=0.5pt,dotted,color=gray] (xtwo) to (s);
    % %
    % % points
    % %
    % % exotica
    % % continua
    \begin{scope}[->,>=latex, shorten >=3pt, shorten <=0pt, line width=2pt, color=DarkOrange]
      \draw (e) --(xtwo); 
      \draw (m) --(xtwo); 
    \end{scope}
    \draw [->,line width=0.5pt,color=blue] (-1,7) -- (g);
    \draw [->,line width=0.5pt,color=blue] (1,7) -- (g);
    \foreach \point in {xone,xtwo}
    \filldraw [color=DarkOrange,fill=DarkOrange] (\point) ++(-1pt,-1pt) -- ++(3pt,0pt) -- ++(-1.5pt,2.6pt) -- cycle;
    % galilean/carrollian
    \begin{scope}[shorten >=0pt, shorten <=0pt, line width=2pt, color=green!70!black]
      \draw (dsg) -- (dsgone) node [midway, below, sloped] (tdsg) {\tiny \hyperlink{S9l}{$\ztdSG_{\gamma\in[-1,1]}$}};
      \draw (dsgone) -- (adsg) node [midway, below, sloped] (tadsg) {\tiny \hyperlink{S11l}{$\ztAdSG_{\chi\geq0}$}};
    \end{scope}
    \draw [->,line width=0.5pt,color=blue] (dsgone) -- (g); 
    \draw [->,line width=0.5pt,color=blue] (-2,2) -- (g); 
    \draw [->,line width=0.5pt,color=blue] (2,2) -- (g); 
    \foreach \point in {g,adsg,dsg,flc,dsgone}
    \filldraw [color=green!70!black,fill=green!70!black] (\point) circle (1.5pt);
    % % pseudo-riemannian
    \foreach \point in {e,m,sph,hyp,ads}
    \filldraw [color=red!70!black,fill=red!70!black] (\point) ++(-1.5pt,-1.5pt) rectangle ++(3pt,3pt);
    % % aristotelian
    \foreach \point in {s,ts}
    \filldraw [color=gray!90!black,fill=gray!10!white] (\point) circle (1.5pt);
    % %
    % % legend
    % %
    \begin{scope}[shift={(0.2cm,6.5cm)}]
      \draw [line width=1pt,color=gray!50!black] (3.2,1.4) rectangle  (7.2,3.6);
      \filldraw [color=red!70!black,fill=red!70!black] (3.5,3.25) ++(-1.5pt,-1.5pt) rectangle ++(3pt,3pt);
      \draw (3.5,3.25) node[color=black,anchor=west] {\small riemannian/lorentzian};
      \filldraw [color=green!70!black,fill=green!70!black] (3.5,2.75) circle (1.5pt) node[color=black,anchor=west] {\small galilean = carrollian};
      \filldraw [color=gray!90!black,fill=gray!10!white] (3.5,2.25) circle (1.5pt) node[color=black,anchor=west] {\small aristotelian};
      \filldraw [color=DarkOrange,fill=DarkOrange] (3.5,1.75) ++(-1.5pt,-1pt) -- ++(3pt,0pt) -- ++(-1.5pt,2.6pt) -- cycle;
      \draw (3.5,1.75) node[color=black,anchor=west] {\small exotic};
    \end{scope}
  \end{tikzpicture}
  \caption{Two-dimensional homogeneous spacetimes and their limits}
  \label{fig:d=2-graph}
\end{figure}

\subsection{Organisation of the paper}
\label{sec:organisation-paper}

The rest of this paper is organised as follows.

Section~\ref{sec:kinem-spac} contains the basic definitions and a
summary of the main results in the paper.  Readers who are pressed for
time should perhaps read that section and then skip to
Section~\ref{sec:some-geom-prop}.  In
Section~\ref{sec:basic-definitions} we define the main objects of
interest: kinematical Lie algebras, their homogeneous spacetimes and
their infinitesimal description in terms of Lie pairs, relegating to
Appendix~\ref{sec:infin-descr-homog} a more careful treatment
including the proof that (geometrically realisable, effective) Lie
pairs are in one-to-one correspondence with simply-connected
homogeneous spaces.  In Section~\ref{sec:summary-main-results} we
summarise the main results and list the resulting isomorphism classes
of simply-connected kinematical and aristotelian homogeneous
spacetimes in Tables~\ref{tab:spacetimes} and \ref{tab:aristotelian}.

Sections~\ref{sec:class-kinem-lie} and \ref{sec:class-kinem-spac}
contains the details leading up to Table~\ref{tab:spacetimes}.
(Table~\ref{tab:aristotelian} is the result of the classification of
aristotelian Lie algebras in Appendix~\ref{sec:class-arist-lie}.)  In
Section~\ref{sec:class-kinem-lie} we classify the isomorphism classes
of kinematical Lie pairs.  This is achieved by going one by one
through the isomorphism classes of kinematical Lie algebras and
determining for each one the possible Lie pairs up to isomorphism.  We
do this in turn for generic kinematical Lie algebras in $D\geq 3$ in
Section~\ref{sec:generic-dimension}, kinematical Lie algebras unique
to $D=3$ in Section~\ref{sec:three-dimensions}, kinematical Lie
algebras in $D=2$ in \ref{sec:two-dimensions} and finally those in
$D=1$ in Section~\ref{sec:one-dimension}.  The classification of
kinematical Lie algebras are summarised in Tables~\ref{tab:generic},
\ref{tab:3d}, \ref{tab:2d} and \ref{tab:bianchi} and that of their
corresponding Lie pairs are summarised in
Tables~\ref{tab:summarydge3}, \ref{tab:summarydeq3},
\ref{tab:summary2d-real} and \ref{tab:summary1d}.  Some Lie pairs in
$D\leq 2$ can be seen to be the low-dimensional avatars of some Lie
pairs which exist for all $D$.  Tables~\ref{tab:deq2correspondence}
and \ref{tab:deq1correspondence} describe this correspondence.  In
Section~\ref{sec:class-kinem-spac} we select from the Lie pairs
obtained in Section~\ref{sec:class-kinem-lie} those which are
effective and then show, using a variety of methods, that all
effective Lie pairs are geometrically realisable.  The end result of
this section and one of the  main results of this work
is the list of simply-connected homogeneous kinematical spacetimes in
Table~\ref{tab:spacetimes}.

The classification of homogeneous spacetimes provides the ``objects''
in Figure~\ref{fig:generic-d-graph}. The ``arrows'' between these
objects is provided by limits. In Section~\ref{sec:limits} we flesh
out Figures~\ref{fig:generic-d-graph}, \ref{fig:d=3-graph} and
\ref{fig:d=2-graph} by exhibiting the limits relating different
spacetimes in Tables~\ref{tab:spacetimes} and \ref{tab:aristotelian}.

Section~\ref{sec:some-geom-prop} starts the geometric study of the
homogeneous spacetimes in Tables~\ref{tab:spacetimes} and
\ref{tab:aristotelian}. In Section~\ref{sec:basic-notions} we review
some basic notions like reductivity, symmetry, the linear isotropy
representation and the invariant structures of interest: lorentzian,
riemannian, galilean, carrollian and aristotelian.  In
Section~\ref{sec:invar-conn} we briefly review the notion of an
invariant (affine) connection on a homogeneous space and explore how
properties (or absence of) canonical invariant connections help us
characterise the spacetimes.  In Section~\ref{sec:invar-homog-spac} we
list the resulting spacetimes dividing them into classes depending on
the whether they are non-reductive, flat symmetric, non-flat symmetric
and reductive torsional.  Table~\ref{tab:spacetimes-props} summarises
these results for the homogeneous spacetimes in
Tables~\ref{tab:spacetimes} and \ref{tab:aristotelian}.

Finally, Section~\ref{sec:conclusions} we offer some conclusions and
point to further work.  In addition, there are two appendices in the
paper.  In Appendix~\ref{sec:class-arist-lie} we classify aristotelian
Lie algebras and in Appendix~\ref{sec:infin-descr-homog} we prove that
our infinitesimal approach actually classifies simply-connected
homogeneous spacetimes.

\subsection{Reader's guide}
\label{sec:readers-guide}

It is our intention that this paper should be useful to the community,
but realise that there is a risk that the results are hard to extract
from the details of how we arrived at them.  We have therefore tried
to write the paper in a way that a reader who is happy to believe the
classification can reach it quickly without having to wade through the
details of how we arrived at it. The main information content of this
paper is contained in Sections~\ref{sec:kinem-spac} and
\ref{sec:some-geom-prop} and a reader who is pressed for time should
perhaps concentrate on those two sections at a first reading.  In
particular, Tables~\ref{tab:spacetimes} and \ref{tab:aristotelian}
contain the list of simply-connected homogeneous kinematical and
aristotelian spacetimes and Table~\ref{tab:spacetimes-props} lists
their basic geometrical properties: whether they are
reductive/symmetric/affine, whether they admit parity and/or time
reversal symmetry, the type of invariant structure that they
possess (if any): lorentzian, riemannian, galilean, carrollian or
aristotelian, and, for the reductive examples, whether the canonical
invariant connection is flat and/or torsion-free.  In addition,
Figures~\ref{fig:generic-d-graph}, \ref{fig:d=3-graph} and
\ref{fig:d=2-graph} illustrate limits between the spacetimes and it too
contains useful information.  Busy readers should be able to read the
introduction and Section~\ref{sec:kinem-spac} and then skip to
Section~\ref{sec:some-geom-prop} where we discuss some geometrical
properties of the homogeneous spacetimes.

Many of the tables and figures contain hyperlinks to ease navigation.
Let us explain here how to use them. The spacetimes in
Figures~\ref{fig:generic-d-graph}, \ref{fig:d=3-graph} and
\ref{fig:d=2-graph} as well as our summary table of the properties,
Table~\ref{tab:spacetimes-props}, are hyperlinked to
Tables~\ref{tab:spacetimes} and \ref{tab:aristotelian} which provide
the explicit kinematical and aristotelian spacetimes. Every row in
these tables starts with a label corresponding to one of the
kinematical spacetimes. These labels are hyperlinked to the
corresponding label in Table~\ref{tab:effective}. That table lists for
each spacetime the corresponding effective Lie pairs. These Lie pairs
are themselves hyperlinked to the relevant tables of Lie pairs:
Table~\ref{tab:summarydge3} for $D\geq 3$,
Table~\ref{tab:summary2d-real} for $D=2$ and Table~\ref{tab:summary1d}
for $D=1$. Those tables also contain the information of the
isomorphism class of kinematical Lie algebra associated to that Lie
pair and the label of the Lie algebra is hyperlinked to the relevant
tables of kinematical Lie algebras: Table~\ref{tab:generic} for
$D\geq 3$, Table~\ref{tab:2d} for $D=2$ and Table~\ref{tab:bianchi}
for $D=1$.  In addition, Table~\ref{tab:LAs-to-spacetimes} shows which
homogeneous kinematical spacetimes are associated with which
kinematical Lie algebra.

For example, if we click on $\hyperlink{S16l}{\zLC}$ in
Figure~\ref{fig:generic-d-graph} we are taken to
Table~\ref{tab:spacetimes}. Clicking on \flc\ in
Table~\ref{tab:spacetimes}, we are taken to Table~\ref{tab:effective}
where we see that it can be described by several Lie pairs, depending
on dimension: LP17 ($D\geq 3$), LP45 ($D=2$) or LP79 ($D=1$).
Clicking, say, on LP17, we are taken to Table~\ref{tab:summarydge3},
where we see that it comes from Lie algebra LA11 and clicking on LA11
we are taken to Table~\ref{tab:generic}, where we see that it
corresponds to $\so(D+1,1)$ with the Lie brackets given there.

\section{Kinematical spacetimes}
\label{sec:kinem-spac}

The well-known spacetimes in Figure~\ref{fig:state-of-prior-art} are
all symmetric homogeneous spaces of kinematical Lie groups.  In this
section we will review this description.

\subsection{Basic definitions}
\label{sec:basic-definitions}

Before we properly define the main objects of interest, let us
motivate our definitions and summarise the philosophy of the
construction. By definition, a kinematical Lie algebra has the same
dimension as the Lie algebra of isometries of a maximally symmetric
riemannian/lorentzian spacetime, and it consists of spatial rotations
and additional generators we call boosts and (space and time)
translations. However before we specify how the kinematical Lie
algebra acts on the spacetime, we cannot make a precise identification
of which generators are boosts and/or translations. This is the reason
to refine the discussion from the kinematical Lie algebras to the
homogeneous spacetimes. The spacetimes are constructed in such a way
that the stabiliser of any point contains the rotations about that
point. This spatial isotropy implies that all invariant tensors are
isotropic (i.e., $\so(D)$ rotationally invariant). The vectorial
generators in the stabiliser are interpreted as boosts, whereas the
additional generators are interpreted as translations. The resulting
spacetime is by construction homogeneous, which roughly means that
every point of the manifold looks like any other point. We now provide
the precise details.

\begin{definition}
  A \textbf{kinematical Lie algebra} (with $D$-dimensional space
  isotropy) is a real Lie algebra $\k$ satisfying the following two
  properties:
  \begin{enumerate}
  \item $\k$ contains a Lie subalgebra $\r \cong \so(D)$, the
    Lie algebra of rotations of $D$-dimensional euclidean space; and
  \item $\k$ decomposes as $\k = \r \oplus 2 V \oplus S$ as a
    representation of $\r$, where $2 V$ are two copies of the
    $D$-dimensional (vector) irreducible representation of $\so(D)$
    and $S$ is the one-dimensional (scalar) trivial representation of
    $\so(D)$.
  \end{enumerate}
  By a \textbf{kinematical Lie group} we mean any Lie group whose Lie
  algebra is a kinematical Lie algebra.
\end{definition}

It follows from this definition that we can describe a kinematical Lie
algebra explicitly in terms of a basis $J_{ab} = - J_{ba}$ for the
rotational subalgebra $\r$, $V_a^{(i)}$ with $i=1,2$ for the two
copies of $V$ and $H$ for $S$.  The definition implies that the Lie
brackets of $\k$ in this basis include the following:
\begin{equation}
  \label{eq:defkinbr}
  \begin{split}
    [J_{ab}, J_{cd}] &= \delta_{bc} J_{ad} - \delta_{ac} J_{bd} - \delta_{bd} J_{ac} + \delta_{ad} J_{bc}\\
    [J_{ab}, V_c^{(i)}] &= \delta_{bc} V_a^{(i)} - \delta_{ac} V_b^{(i)}\\
    [J_{ab}, H ] &= 0,
  \end{split}
\end{equation}
and any other Lie brackets are subject only to the Jacobi identity,
which implies, in particular, equivariance under $\r$.  It is
convenient to relabel $V^{(1)}_a$ as $B_a$ and $V^{(2)}_a$ as $P_a$, a
notation reminiscent of the boosts and translations in kinematics;
although it must be stressed that there is no \emph{a priori} geometrical
interpretation of these generators: they only acquire such an
interpretation when we realise them in a spacetime on which the
kinematical Lie group acts.

Kinematical Lie algebras have been classified up to Lie algebra
isomorphism \cite{MR0238545,MR857383,Figueroa-OFarrill:2017ycu,
  Figueroa-OFarrill:2017tcy, Andrzejewski:2018gmz}.  These
classifications are the starting point to the classification of
homogeneous spacetimes and hence they will be briefly recalled at the
start of Sections~\ref{sec:generic-dimension},
\ref{sec:three-dimensions}, \ref{sec:two-dimensions} and
\ref{sec:one-dimension} and contained in Tables~\ref{tab:generic},
\ref{tab:3d}, \ref{tab:2d} and \ref{tab:bianchi}.

In this paper we obtain a classification of simply-connected
homogeneous \emph{spacetimes} of kinematical Lie groups (up to
isomorphism).  Not every homogeneous space of a kinematical Lie group
is a spacetime, so this requires a definition.  Notice that, in
particular, a homogeneous spacetime of a kinematical Lie group with
$D$-dimensional space isotropy is ($D+1$)-dimensional.

\begin{definition}
  By a \textbf{homogeneous kinematical spacetime} we mean a
  homogeneous space $M$ of a kinematical Lie group $\Kgr$, satisfying
  the following properties:
  \begin{enumerate}
  \item $M$ is a connected smooth manifold, and
  \item $\Kgr$ acts transitively and locally effectively\footnote{See
      Appendix~\ref{sec:trans-acti-lie} for the basic definitions, if
      needed.} on $M$ with stabiliser $\Hgr$, where
  \item $\Hgr$ is a closed subgroup of $\Kgr$ whose Lie algebra $\h$
    contains a rotational subalgebra $\r \cong \so(D)$ and decomposes as
    $\h = \r \oplus V$ as an adjoint $\r$-module, where $V$ is an
    irreducible $D$-dimensional vector representation of $\so(D)$.
  \end{enumerate}
  It follows that $M$ is $\Kgr$-equivariantly diffeomorphic to
  $\Kgr/\Hgr$.  Within the confines of this paper, and because we will
  need to refer to them often, we will say that $\Hgr$ and $\h$ are
  \textbf{admissible}.
\end{definition}

It may be convenient to keep an example in mind, so let us consider
the Poincaré Lie group, whose Lie algebra is defined relative to the
basis $(J_{ab}, B_a, P_a, H)$ by the nonzero Lie brackets
\begin{equation}
  \begin{aligned}\relax
    [J_{ab}, J_{cd}] &= \delta_{bc} J_{ad} - \delta_{ac} J_{bd} - \delta_{bd} J_{ac} + \delta_{ad} J_{bc}\\
    [J_{ab}, B_c] &= \delta_{bc} B_a - \delta_{ac} B_b\\
    [J_{ab}, P_c] &= \delta_{bc} P_a - \delta_{ac} P_b
  \end{aligned}
  \qquad\qquad
  \begin{aligned}\relax
    [B_a, B_b] &= J_{ab}\\
    [B_a, P_b] &= \delta_{ab} H\\
    [B_a,H] &= P_a.
  \end{aligned}
\end{equation}
This is obtained from the more familiar expression
\begin{equation}
  \begin{split}
    [J_{\mu\nu}, J_{\rho\sigma}] &= \eta_{\nu\rho} J_{\mu\sigma} - \eta_{\mu\rho} J_{\nu\sigma} -  \eta_{\nu\sigma} J_{\mu\rho} + \eta_{\mu\sigma} J_{\nu\rho} \\
    [J_{\mu\nu}, P_\rho] &= \eta_{\nu\rho} P_\mu - \eta_{\mu\rho} P_\nu, 
  \end{split}
\end{equation}
by decomposing $J_{\mu\nu} = (J_{ab}, B_a := J_{0a})$ and $P_\mu =
(P_a, H:=P_0)$, where $\eta_{00} = -1$ and $\eta_{ab} = \delta_{ab}$.
Let us take $\Kgr$ to be the Poincaré Lie group and $\Hgr$ the Lorentz
subgroup; that is, the (admissible) subgroup generated by the Lie
subalgebra spanned by $(J_{ab}, B_a)$. Then $\Kgr/\Hgr$ is
diffeomorphic to Minkowski spacetime, as is well known.

Perhaps not so well known is the fact that the Poincaré group admits a
second homogeneous spacetime.  If we now let $\Hgr'$ denote the
(again, admissible) subgroup generated by the Lie subalgebra $\h'$
spanned by $(J_{ab},P_a)$, then $\Kgr/\Hgr'$ is diffeomorphic to the
\emph{carrollian anti de~Sitter spacetime} ($\hyperlink{S15l}{\mathsf{AdSC}}$) (also known as
\emph{para-Minkowski spacetime}), as we will see below.  To more
easily distinguish between these two homogeneous spacetimes of the
Poincaré group, it is convenient to change basis in the Poincaré Lie
algebra in such a way that the admissible Lie subalgebra $\h'$ is also
spanned by $(J_{ab},B_a)$.  Doing so we arrive at what is often termed
the \emph{para-Poincaré Lie algebra}, with nonzero Lie brackets
\begin{equation}
  \begin{aligned}\relax
    [J_{ab}, J_{cd}] &= \delta_{bc} J_{ad} - \delta_{ac} J_{bd} - \delta_{bd} J_{ac} + \delta_{ad} J_{bc}\\
    [J_{ab}, B_c] &= \delta_{bc} B_a - \delta_{ac} B_b\\
    [J_{ab}, P_c] &= \delta_{bc} P_a - \delta_{ac} P_b
  \end{aligned}
  \qquad\qquad
  \begin{aligned}\relax
    [B_a, P_b] &= \delta_{ab} H\\
    [H,P_a] &= B_a\\
    [P_a, P_b] &= J_{ab}.
  \end{aligned}
\end{equation}
Of course, this Lie algebra is isomorphic to the Poincaré Lie algebra,
but not in a way which fixes the admissible subalgebra.  We see from
the above Lie brackets that translations no longer commute, signalling
that this spacetime is not flat; although, as we will see, the
non-flat connection, which is the canonical Poincaré-invariant
connection on this symmetric homogeneous space, is not a metric
connection since $\hyperlink{S15l}{\mathsf{AdSC}}$ does not admit a Poincaré-invariant metric.
If it did, by dimension, it would have be maximally symmetric and
hence isometric to Minkowski spacetime, where translations do commute.

As we explain in Appendix~\ref{sec:infin-descr-homog}, the
classification of isomorphism classes of simply-connected homogeneous
kinematical spacetimes can be arrived at infinitesimally, by
classifying isomorphism classes of (geometrically realisable,
effective) kinematical Lie pairs.  This too requires a definition.

\begin{definition}
  A \textbf{(kinematical) Lie pair} is a pair $(\k,\h)$ consisting of
  a kinematical Lie algebra $\k$ and an admissible subalgebra $\h$.
  Two Lie pairs $(\k_1,\h_1)$ and $(\k_2,\h_2)$ are
  \textbf{isomorphic} if there is a Lie algebra isomorphism
  $\varphi: \k_1 \to \k_2$ with $\varphi(\h_1) = \h_2$.  A Lie pair
  $(\k,\h)$ is \textbf{effective} if $\h$ does not contain any nonzero
  ideals of $\k$.  It is said to be \textbf{(geometrically)
    realisable} if there exists a connected Lie group $\Kgr'$ with Lie
  algebra $\k'$ and a \emph{closed} Lie subgroup $\Hgr'$ with Lie
  algebra $\h'$ with $(\k',\h')$ isomorphic to $(\k,\h)$.  The
  homogeneous space $\Kgr'/\Hgr'$ is said to be a \textbf{geometric
    realisation} of $(\k,\h)$.
\end{definition}

The relationship between Lie pairs and homogeneous spaces extends the
relationship between Lie algebras and Lie groups.  Recall that
associated with every finite-dimensional Lie algebra $\k$ there exists
a unique (up to isomorphism) simply-connected Lie group whose Lie
algebra is isomorphic to $\k$.  A similar correspondence exists
between homogeneous spaces and Lie pairs, except that we need to
restrict to Lie pairs which are geometrically realisable (existence)
and effective (uniqueness).  This is explained in detail in
Appendix~\ref{sec:infin-descr-homog}.

In this paper we will classify isomorphism classes of geometrically
realisable, effective Lie pairs $(\k,\h)$, where $\k$ is a kinematical
Lie algebra and $\h$ an admissible Lie subalgebra.  As just explained,
this is equivalent to classifying isomorphism classes of
simply-connected homogeneous kinematical spacetimes.  We will actually
exploit the classification of kinematical Lie algebras up to
isomorphism and in this way fix a kinematical Lie algebra $\k$ in each
isomorphism class and classify isomorphism classes of (geometrically
realisable, effective) Lie pairs $(\k,\h)$ with $\h$ an admissible
subalgebra, with two such Lie pairs $(\k,\h_1)$ and $(\k,\h_2)$
declared to be isomorphic if there is an \emph{automorphism} of $\k$
sending $\h_1$ to $\h_2$.

\subsection{Summary of main results}
\label{sec:summary-main-results}

The classification is described in Sections~\ref{sec:class-kinem-lie}
and \ref{sec:class-kinem-spac}, but we think it might be helpful to
the reader to collect here already the results.  In the tables below
we use an abbreviated notation in which we do not write the $\so(D)$
indices explicitly.  We write $\J$, $\B$, $\P$ and $H$ for the
generators of the kinematical Lie algebra $\k$ and write the
kinematical Lie brackets of \eqref{eq:defkinbr} as
\begin{equation}
  \label{eq:kin}
  [\J,\J] = \J \qquad [\J,\B] = \B \qquad [\J, \P] = \P \qquad\text{and}\qquad [\J, H] = 0.
\end{equation}
For $D\neq 2$, any other brackets can be reconstructed unambiguously
from the abbreviated expression since there is only one way to
reintroduce indices in an $\so(D)$-equivariant fashion; that is, using
only the $\so(D)$-invariant tensors $\delta_{ab}$ and $\epsilon_{a_{1}
  \cdots a_{D}}$ on the right hand side of the brackets. For example,
\begin{equation}
  [H, \B] = \P \quad\text{stands for}\quad [H, B_a] = P_a \qquad\text{and}\qquad
  [\B,\P] = H + \J \quad\text{for}\quad [B_a, P_b] = \delta_{ab} H + J_{ab}.
\end{equation}
In $D=3$ we may also have brackets of the form
\begin{equation}
 [\P,\P] = \P \qquad\text{which we take to mean}\qquad [P_a,P_b] =
 \epsilon_{abc} P_c,
\end{equation}
where we employ Einstein's summation convention.
Similarly, for $D=2$, $\epsilon_{ab}$ is rotationally invariant and can
appear in Lie brackets.  So we will write, e.g.,
\begin{equation}
  [H, \B] = \B + \Pt  \qquad\text{for}\qquad  [H, B_a] = B_a +
  \epsilon_{ab} P_b,
\end{equation}
et cetera.

Table~\ref{tab:spacetimes} summarises all the hard work in
Sections~\ref{sec:class-kinem-lie} and \ref{sec:class-kinem-spac} and
lists all the inequivalent geometrically realisable effective Lie
pairs $(\k,\h)$ with $\k$ a kinematical Lie algebra and $\h$ an
admissible Lie subalgebra. A basis has been chosen in such a way that
$\h$ is spanned by $(J_{ab}, B_a)$. In this way, Lie pairs are
uniquely characterised by specifying the Lie brackets of $\k$ in this
basis. The kinematical Lie brackets \eqref{eq:kin} are common to all
kinematical Lie algebras, so that we need only specify those Lie
brackets which do not involve the rotations. The table is divided into
five sections, separated by horizontal rules. From top to bottom, the
first four correspond to the lorentzian, riemannian, galilean and
carrollian spacetimes. The final section corresponds to ``exotic''
two-dimensional spacetimes admitting no such structures. Two remarks
are in order about this table and both concern the case of
$D=1$. Since there are no rotations in $D=1$, in any row where $\J$
appears, we are tacitly assuming that we set $\J=0$ if $D=1$. Also,
some of the $D\geq 1$ spacetimes become accidentally pairwise
isomorphic in $D=1$: namely, carrollian and galilean, de~Sitter and
anti~de~Sitter, carrollian dS and galilean AdS, and carrollian AdS and
galilean dS. This explains why we write $D\geq2$ for carrollian,
de~Sitter, and carrollian (anti) de~Sitter. In this way no two rows
are isomorphic and hence every row in the table specifies a unique
isomorphism class of simply-connected homogeneous kinematical
spacetime.  Perhaps it bears repeating that, as mentioned already in
the introduction, the galilean (A)dS spacetimes are often called the
Newton--Hooke spacetimes.

\begin{table}[h!]
  \centering
  \caption{Simply-connected homogeneous ($D+1$)-dimensional kinematical spacetimes}
  \label{tab:spacetimes}
  \rowcolors{2}{blue!10}{white}
  \resizebox{\textwidth}{!}{
    \begin{tabular}{l|>{$}c<{$}|*{5}{>{$}l<{$}}|l}\toprule
      \multicolumn{1}{c|}{Label} & D & \multicolumn{5}{c|}{Nonzero Lie brackets in addition to $[\J,\J] = \J$, $[\J, \B] = \B$, $[\J,\P] = \P$} & \multicolumn{1}{c}{Comments}\\\midrule
      % lorentzian
      \hypertarget{S1l}{}\hyperlink{S1}{\mink} & \geq 1 & [H,\B] = -\P & & [\B,\B] = \J & [\B,\P] = H & & Minkowski ($\MM$)\\
      \hypertarget{S2l}{}\hyperlink{S2}{\ds} & \geq 2 & [H,\B] = -\P & [H,\P] = -\B & [\B,\B]= \J & [\B,\P] = H & [\P,\P]= - \J & de~Sitter ($\zdS$)\\
      \hypertarget{S3l}{}\hyperlink{S3}{\ads} & \geq 1 & [H,\B] = -\P & [H,\P] = \B & [\B,\B]= \J & [\B,\P] = H & [\P,\P] = \J & anti~de~Sitter ($\zAdS$)\\\midrule
      % riemannian
      \hypertarget{S4l}{}\hyperlink{S4}{\euc} & \geq 1 & [H,\B] = \P & & [\B,\B] = -\J & [\B,\P] = H & & euclidean ($\EE$)\\
      \hypertarget{S5l}{}\hyperlink{S5}{\sph} & \geq 1 & [H,\B] = \P & [H,\P] = -\B & [\B,\B]= -\J & [\B,\P] = H & [\P,\P]= - \J & sphere ($\SS$)\\
      \hypertarget{S6l}{}\hyperlink{S6}{\hyp} & \geq 1 & [H,\B] = \P & [H,\P] = \B & [\B,\B]= -\J & [\B,\P] = H & [\P,\P] = \J & hyperbolic ($\HH$)\\\midrule
      % galilean
      \hypertarget{S7l}{}\hyperlink{S7}{\gal} & \geq 1 & [H,\B] = -\P & & & & & galilean ($\zG$)\\
      \hypertarget{S8l}{}\hyperlink{S8}{\dsg} & \geq 1 & [H,\B] = -\P & [H,\P] = -\B & & & & galilean de~Sitter ($\zdSG= \ztdSG_{\gamma=-1}$)\\
      \hypertarget{S9l}{}\hyperlink{S9}{\tdsg$_\gamma$} & \geq 1 & [H,\B] = -\P & [H,\P] = \gamma\B + (1+\gamma)\P & & & & torsional galilean de~Sitter $\ztdSG_{\gamma\in (-1,1]}$ \\
      \hypertarget{S10l}{}\hyperlink{S10}{\adsg} & \geq 1 & [H,\B] =  -\P & [H,\P] = \B & & & & galilean anti~de~Sitter  ($\zAdSG = \ztAdSG_{\chi=0}$)\\
      \hypertarget{S11l}{}\hyperlink{S11}{\tadsg$_\chi$} & \geq 1 & [H,\B] = -\P & [H,\P] = (1+\chi^2) \B + 2\chi \P & & & & torsional galilean anti~de~Sitter $\ztAdSG_{\chi>0}$ \\
      \hypertarget{S12l}{}\hyperlink{S12}{\twodgal$_{\gamma,\chi}$} & 2 & [H,\B] = -\P & [H,\P] = (1+\gamma) \P - \chi\Pt + \gamma \B - \chi \Bt & & & &  $\gamma\in [-1,1), \chi >0$\\\midrule
      % carrollian
      \hypertarget{S13l}{}\hyperlink{S13}{\car} & \geq 2 & & & & [\B,\P] = H & & carrollian ($\zC$)\\
      \hypertarget{S14l}{}\hyperlink{S14}{\dsc} & \geq 2 & & [H,\P] = -\B & & [\B,\P] = H & [\P,\P] = -\J & carrollian de~Sitter ($\zdSC$) \\
      \hypertarget{S15l}{}\hyperlink{S15}{\adsc} & \geq 2 & & [H,\P] = \B & & [\B,\P] = H & [\P,\P] = \J & carrollian anti~de~Sitter ($\zAdSC$) \\
      \hypertarget{S16l}{}\hyperlink{S16}{\flc} & \geq 1 & [H,\B] = \B & [H,\P] = -\P & & [\B,\P] = H + \J & & carrollian light cone ($\zLC$)\\\midrule
      % no structure at all
      \hypertarget{S17l}{}\hyperlink{S17}{\xone} & 1 & [H,B] = -P & & & [B,P] = - H - 2 P & & \\
      \hypertarget{S18l}{}\hyperlink{S18}{\xtwo} & 1 & [H,B] = H & & & [B,P] = -P & & \\
      \hypertarget{S19l}{}\hyperlink{S19}{\xthree$_\chi$} & 1 & [H,B] = (1+\chi) H & & & [B,P]= (1-\chi)P & & $\chi>0$\\
      \hypertarget{S20l}{}\hyperlink{S20}{\xfour$_\chi$} & 1 & [H,B] = -P & & & [B,P] = -(1+\chi^2) H - 2\chi P & & $\chi>0$ \\ \bottomrule
    \end{tabular}
  }
\\[10pt]
    \caption*{Geometrical properties of the spacetimes are provided in Table \ref{tab:spacetimes-props}.}
\end{table}

Table~\ref{tab:aristotelian} lists the Lie pairs corresponding to the
aristotelian spacetimes.  Many of them arise from non-effective Lie
pairs $(\k,\h)$ with $\k$ a kinematical Lie algebra.  If such a pair
is not effective, it means that $\h$ contains a nonzero ideal of $\k$.
For $\k$ a kinematical Lie algebra, this cannot be other than the span
$\b$ of the $(B_a)$, assuming they do form an ideal.  When this is the
case, we may quotient both $\k$ and $\h$ by the ideal $\b$ and arrive
at an effective (by construction) Lie pair $(\k/\b, \h/\b) = (\a,\r)$,
where $\a = \k/\b$ is an aristotelian Lie algebra (see Appendix
\ref{sec:class-arist-lie}) and $\r = \h/\b$ is a rotational subalgebra
of $\a$.

To fix ideas, let us consider the example of the galilean algebra
$\g$.  This is the kinematical Lie algebra with nonzero Lie brackets
$[H,\B] = -\P$ in addition to those in \eqref{eq:kin}.  As we will
see, there are two isomorphism classes of Lie pairs associated to
$\g$.  If we choose bases for $\g$ so that in both cases the
admissible subalgebra is the span of $(J_{ab},B_a)$, then the two Lie
pairs are described by the following Lie brackets for $\g$ in addition
to those in \eqref{eq:kin}:
\begin{itemize}
\item $[H,\B] = -\P$, which is the standard galilean spacetime, and
\item $[H,\P] = \B$.
\end{itemize}
The latter Lie pair is not effective because $\h$ contains the ideal
$\b$ spanned by the $(B_a)$.  If we quotient this Lie pair by $\b$,
which boils down to discarding any $B_a$ in the Lie brackets, we see
that now $[H,\P] = 0$ and we arrive at the static aristotelian Lie
algebra.

Aristotelian Lie pairs are always geometrically realisable because the
rotational subalgebra $\r$ generates a compact subgroup and compact
subgroups are always closed.  Furthermore, since $\r$ is fixed, every
aristotelian Lie algebra gives rise to a unique aristotelian Lie pair,
so that the classification in Appendix~\ref{sec:class-arist-lie} is
also a classification of simply-connected aristotelian spacetimes up
to isomorphism.

\begin{table}[h!]
  \centering
  \caption{Simply-connected homogeneous ($D+1$)-dimensional aristotelian spacetimes}
  \label{tab:aristotelian}
  \rowcolors{2}{blue!10}{white}
  \begin{tabular}{l|>{$}c<{$}|*{2}{>{$}l<{$}}|l} \toprule
    \multicolumn{1}{c|}{Label} & D & \multicolumn{2}{c|}{Nonzero Lie brackets in addition to $[\J,\J] = \J $ and $[\J,\P] = \P$}& \multicolumn{1}{c}{Comments}\\\midrule
    \hypertarget{A21}{\st} & \geq 0 & & & static ($\zS$)\\
    \hypertarget{A22}{\tst} & \geq 1 & [H,\P] = \P & & torsional static ($\zTS$)\\
    \hypertarget{A23p}{\athree$_{+1}$} & \geq 2 & & [\P,\P] = \J & $\RR\times\HH^D$\\
    \hypertarget{A23m}{\athree$_{-1}$} & \geq 2 & & [\P,\P] = - \J & $\RR\times\SS^D$\\
    \hypertarget{A24}{\twoda} & 2 & & [\P,\P] = H & \\\bottomrule
  \end{tabular}
\end{table}

The next two sections contain the classifications leading up to
Table~\ref{tab:spacetimes}, whereas Section~\ref{sec:limits} contains
the details leading up to Figures~\ref{fig:generic-d-graph},
\ref{fig:d=3-graph} and \ref{fig:d=2-graph}. The busy reader may wish
to skip them at a first reading and go directly to
Section~\ref{sec:some-geom-prop} where we explore some of the
geometrical properties of these spacetimes, culminating in
Table~\ref{tab:spacetimes-props}.

\section{Classification of kinematical Lie pairs}
\label{sec:class-kinem-lie}

This section, which can be skipped at a first reading, contains the
details of the classification of kinematical Lie pairs $(\k,\h)$ up to
isomorphism, where $\k$ is a kinematical Lie algebra and $\h$ is an
admissible Lie algebra.  At this stage we will not worry about whether
the resulting Lie pairs are effective and/or geometrically realisable.
The results of this section are collected in
Tables~\ref{tab:summarydge3}, \ref{tab:summarydeq3},
\ref{tab:summary2d-real} and \ref{tab:summary1d}, which will be the
starting point for the analysis in Section~\ref{sec:class-kinem-spac},
where we extract from those tables the effective Lie pairs which admit
a geometric realisation to arrive at Table~\ref{tab:spacetimes}.

We will simplify the analysis by making use of the classification of
kinematical Lie algebras, which we briefly review.  Kinematical Lie
algebras (with $D$-dimensional space isotropy) have been classified up
to Lie algebra isomorphism.  In $D=0$ there is a unique
one-dimensional Lie algebra, whereas in $D=1$ there are no rotations,
so any three-dimensional Lie algebra is kinematical.  These were
classified by Bianchi \cite{Bianchi} (see \cite{MR1900159} for an
English translation) in the context of his classification of
three-dimensional homogeneous spaces; although here they will play the
rôle of symmetries of two-dimensional homogeneous spaces.  The other
classic case is $D=3$, where the kinematical Lie algebras were
classified by Bacry and Nuyts \cite{MR857383} refining earlier work of
Bacry and Lévy-Leblond \cite{MR0238545}.  The remaining cases $D>3$
and $D=2$ are recent classifications.  Following earlier work on the
galilean algebra \cite{JMFGalilean}, the $D=3$ classification was
recovered recently in \cite{Figueroa-OFarrill:2017ycu} using methods
of deformation theory.  These methods were then extended to arrive at
the classifications for $D>3$ \cite{Figueroa-OFarrill:2017tcy} and
$D=2$ \cite{Andrzejewski:2018gmz}.  These classifications are recalled
in this paper and are contained in Tables~\ref{tab:generic},
\ref{tab:3d}, \ref{tab:2d} and \ref{tab:bianchi}.

The classification of kinematical Lie algebras allows us to fix
a kinematical Lie algebra $\k$ and then consider Lie pairs $(\k,\h)$,
with two such pairs $(\k,\h_1)$ and $(\k,\h_2)$ being isomorphic if
there is an automorphism of $\k$ taking $\h_1$ to $\h_2$.  This
suggests the following methodology, which is how we will proceed in
this paper.

For each kinematical Lie algebra $\k$ in the classification, we will
determine the admissible subalgebras $\h$ and, if necessary, the
action of the automorphism group $A = \Aut(\k)$ on them.  We will then
pick one representative from each $A$-orbit.  Finally, we will change
basis (if needed) for $\k$ so that $\h$ is always spanned by $J_{ab}$
and $B_a$.  The resulting Lie pair is then described uniquely by
specifying the Lie brackets of $\k$ in this basis.  Furthermore,
whenever possible, we choose a basis in which the Lie pair is
manifestly reductive.

It follows from the classification of kinematical Lie algebras that
there are kinematical Lie algebras which exist for all $D\geq 1$, but
there are additional Lie algebras for $D\leq 3$, due to accidents in
small dimension: namely, the existence of a rotationally invariant
vector product in $D=3$, a rotationally invariant symplectic structure
in $D=2$, and the absence of rotations in $D=1$.  As explained in
\cite{Andrzejewski:2018gmz}, the case $D=2$ is special and it is
convenient to work with the complexified Lie algebra. The analysis in
that case does not embed easily into the general discussion of $D\geq
3$ and therefore we will have to do it separately. Similarly the case
$D=1$ is special in that any three-dimensional Lie algebra is
kinematical and any one-dimensional subalgebra is admissible. We also
treat this case separately. These considerations suggest first
treating the generic case (restricted to $D\geq 3$), which we do in
Section~\ref{sec:generic-dimension}. Then in
Section~\ref{sec:three-dimensions} we consider the additional
kinematical Lie algebras which are unique to $D=3$. In
Section~\ref{sec:two-dimensions} we determine the Lie pairs associated
to the kinematical Lie algebras in $D=2$. As we will see, there are
two kinematical Lie algebras unique to $D=2$ which do not admit any
Lie pairs.  Then in Section~\ref{sec:one-dimension} we consider the
Lie pairs associated to the three-dimensional Lie algebras.

\subsection{Lie pairs for $D\geq 3$}
\label{sec:generic-dimension}

Let us consider those kinematical Lie algebras which have analogs for
all $D\geq 3$. (They also have analogs for $D<3$, but we will discuss
them again separately in that context.) These are listed in
Table~\ref{tab:generic}, which is borrowed from
\cite[Table~17]{Figueroa-OFarrill:2017tcy} with some small
modification. We only list those nonzero Lie brackets in addition to
\eqref{eq:kin}, which are understood as given. The numbering in this
and other tables is somewhat arbitrary, but might help in referring to
those Lie algebras which are otherwise nameless.  Although we are
aware that referring to some of these kinematical Lie algebras by name
presupposes a geometrical interpretation of the generators which may
or may not be the right interpretation depending on the actual
homogeneous space under consideration, we do feel that it helps to
orient the reader if we point out to which named Lie algebras (when
the name exists) the Lie algebras in Table~\ref{tab:generic} are
isomorphic. In this spirit, let us also mention in this and other
tables the isomorphism type of the Lie algebra, when known.  To this
end let us introduce in Table~\ref{tab:notation} some notation for the
named Lie algebras appearing in this section.

\begin{table}[h!]
  \centering
  \caption{Notation for Lie algebras}
  \label{tab:notation}
  \begin{tabular}{>{$}c<{$}|l} \toprule
    \multicolumn{1}{c|}{Notation} & \multicolumn{1}{c}{Name}\\\midrule
    \s & static\\
    \n_+ & (elliptic) Newton\\
    \n_- & (hyperbolic) Newton\\
    \e & euclidean\\\bottomrule
  \end{tabular}
  \hspace{3cm}
  \begin{tabular}{>{$}c<{$}|l} \toprule
    \multicolumn{1}{c|}{Notation} & \multicolumn{1}{c}{Name}\\\midrule
    \p & Poincaré\\
    \so & orthogonal\\
    \g & galilean\\
    \c & Carroll\\ \bottomrule
  \end{tabular}
\end{table}

\begin{table}[h!]
  \centering
  \caption{Generic kinematical Lie algebras for $D\geq 3$}
  \label{tab:generic}
  \rowcolors{2}{blue!10}{white}
  \resizebox{\textwidth}{!}{
  \begin{tabular}{l|>{$}l<{$}|*{5}{>{$}l<{$}}|l} \toprule
    \multicolumn{1}{c|}{LA\#} & \multicolumn{1}{c|}{$\cong$} &  \multicolumn{5}{c|}{Nonzero Lie brackets in addition to $[\J,\J] = \J$, $[\J, \B] = \B$ and $[\J,\P] = \P$} & \multicolumn{1}{c}{Comments} \\\midrule
    \hypertarget{LA1}{1} & \s & & & & & & \\
    \hypertarget{LA2}{2} & \g & [H,\B] = \P & & & & & \\
    \hypertarget{LA3}{3$_\gamma$} & & [H,\B] = \gamma \B & [H,\P] = \P & & & & $\gamma \in (-1,1)$ \\
    \hypertarget{LA4}{4} & & [H,\B] = \B & [H,\P] = \P & & & & \\
    \hypertarget{LA5}{5} & \n_- & [H,\B] = - \B & [H,\P] = \P & & & & \\
    \hypertarget{LA6}{6} & & [H,\B] = \B + \P & [H, \P] = \P & & & & \\
    \hypertarget{LA7}{7$_\chi$} & & [H,\B] = \chi \B + \P & [H,\P] = \chi \P - \B &  & & & $\chi > 0$ \\
    \hypertarget{LA8}{8} & \n_+ & [H,\B] = \P & [H,\P] = - \B & & & & \\
    \hypertarget{LA9}{9} & \c & & & & [\B,\P] = H & & \\
    \hypertarget{LA10}{10$_\varepsilon$} & \choice{\p}{\e} & [H,\B] = -\varepsilon \P & &  [\B,\B]= \varepsilon \J & [\B,\P] = H & & $\varepsilon = \pm 1$ \\
    \hypertarget{LA11}{11} & \so(D+1,1) & [H,\B] = \B & [H,\P] = -\P & &  [\B,\P] = H + \J & & \\
    \hypertarget{LA12}{12$_\varepsilon$} & \choice{\so(D,2)}{\so(D+2)} & [H,\B] = -\varepsilon \P & [H,\P] = \varepsilon \B &  [\B,\B]= \varepsilon \J & [\B,\P] = H &  [\P,\P] = \varepsilon \J & $\varepsilon = \pm 1$ \\ \bottomrule
  \end{tabular}
  }
\end{table}

In discussing the automorphisms of kinematical Lie algebras we will
fix the rotational subalgebra $\r$ and concentrate on automorphisms
which are the identity on $\r$ and in this way focus only on their
action on the non-rotational generators.  Since automorphisms must
commute with the action of the rotations and the vector representation
is irreducible (also upon complexification) they are necessarily of
the form\footnote{%
  In $D=3$ one might conceive of adding multiples of the rotation
  generators $\J$ to the transformation properties of $\B$ and $\P$,
  but one can show that in automorphisms of $\k$ no such terms
  appear.}
\begin{equation}
  (\B,\P,H ) \mapsto (\B, \P, H )
  \begin{pmatrix} 
    a & b & \zero \\ c & d & \zero \\ \zero & \zero & \Delta
  \end{pmatrix} = (a \B + c \P, b \B + d \P, \Delta H ),
\end{equation}
for some $a,b,c,d,\Delta \in \RR$ and $\Delta (ad - bc) \neq 0$. In
other words, the automorphism group of $\k$ will be a subgroup of
$\GL(2,\RR) \times \RRnZ$.

\subsubsection{Lie pairs associated to Lie algebra \#1}
\label{sec:spac-assoc-stat}

Kinematical Lie algebra \#1 in Table~\ref{tab:generic} is isomorphic
to the static Lie algebra, with all brackets zero except for those in
equation~\eqref{eq:kin}.  The automorphism group is the full
$\GL(2,\RR) \times \RRnZ$.  An admissible Lie subalgebra is
spanned by the rotations and the $D$ vectors $\alpha B_a + \beta P_a$,
for some $\alpha,\beta \in \RR$ not both zero.  We will abbreviate the
span of these vectors as $\alpha \B + \beta \P$, but let us not forget
that we mean a $D$-dimensional subspace.  Under an automorphism,
\begin{equation}
  \alpha \B + \beta \P \mapsto (a \alpha + b \beta) \B + (c \alpha + d
  \beta) \P,
\end{equation}
so that all that happens is that the vector $(\alpha,\beta) \in \RR^2$
is transformed under $\GL(2,\RR)$ according to the defining
representation.  Given any nonzero vector $(\alpha,\beta) \in \RR^2$,
there is a change of basis which sends it to the elementary vector
$(1,0)$, so that up to the action of $\GL(2,\RR)$, we can take the Lie
subalgebra $\h$ to be spanned by the rotations and $\B$.

\subsubsection{Lie pairs associated to Lie algebra \#2}
\label{sec:spac-assoc-galil}

Kinematical Lie algebra \#2 in Table~\ref{tab:generic} is isomorphic
to the galilean Lie algebra.  The automorphism group $A$ is easily
determined to be
\begin{equation}
  A = \left\{
    \begin{pmatrix}  a & \zero & \zero \\ c & a \Delta & \zero \\
      \zero & \zero & \Delta \end{pmatrix} \quad \middle | \quad
    a,c,\Delta \in \RR,~a \Delta \neq 0\right\}.
\end{equation}
Any vectorial subspace $W$ of the form $\alpha \B + \beta \P$ defines
a subalgebra, since $\B$ and $\P$ commute.  Under an automorphism,
\begin{equation}
  \begin{pmatrix}
    \alpha \\ \beta
  \end{pmatrix} \mapsto
  \begin{pmatrix}
    a \alpha \\ a \Delta \beta + c \alpha
  \end{pmatrix}.
\end{equation}
If $\alpha \neq 0$, we can choose $a = \alpha^{-1}$ and $c =- a^2
\Delta \beta$ and hence bring $(\alpha, \beta) \mapsto (1,0)$, so that
$W$ is the span of $\B$.

If $\alpha = 0$, then $\beta \neq 0$ and we can choose $a\Delta =
\beta^{-1}$ so that $(0,\beta) \mapsto (0,1)$.  This means that now
$W$ is the span of $\P$.  We change basis in the Lie algebra so that
$W$ is the span of $\B$.  In the new basis, the galilean algebra
has the additional bracket
\begin{equation}
  [H,\P] = \B.
\end{equation}
This is often known as the \emph{para-galilean} algebra, but it
is of course isomorphic to the galilean algebra.

\subsubsection{Lie pairs associated to Lie algebra \#3$_\gamma$}
\label{sec:spac-assoc-kinem-3}

The automorphism group of Lie algebra \#3$_\gamma$ in
Table~\ref{tab:generic} is determined to be
\begin{equation}
  A = \left\{
    \begin{pmatrix}  a & \zero & \zero \\ \zero & d & \zero \\
      \zero & \zero & 1 \end{pmatrix} \quad \middle | \quad
    a,d \in \RR,~a d \neq 0\right\}.
\end{equation}
Any vectorial subspace $W$ of the form $\alpha \B + \beta \P$ defines
a subalgebra, since $\B$ and $\P$ commute.  Under an automorphism,
\begin{equation}
  \begin{pmatrix}
    \alpha \\ \beta
  \end{pmatrix} \mapsto
  \begin{pmatrix}
    a \alpha \\ d \beta
  \end{pmatrix}.
\end{equation}
We have three cases to consider, depending on whether $\alpha = 0$,
$\beta = 0$ or neither are zero.

If $\beta = 0$, then we can choose $a$ so that $a\alpha =1$, bringing
$(\alpha, 0) \mapsto (1,0)$.  Here $W$ is the span of $\B$.

If $\alpha = 0$, then similarly $W$ is the span of $\P$.  We change
basis so that $W$ is again the span of $\B$, which means that the
additional Lie brackets are now
\begin{equation}
  [H,\B] = \B \qquad\text{and}\qquad [H,\P] = \gamma \P.
\end{equation}

Finally, if $\alpha\beta \neq 0$, then we can choose $a,d$ so that
$(\alpha,\beta) \mapsto (1,1)$ and $W$ is spanned by $\B + \P$.  We
change basis to that $W$ is spanned by $\B$, bringing the additional
Lie brackets to the form
\begin{equation}
  [H, \B] = \P \qquad\text{and}\qquad [H,\P] = -\gamma \B +
  (1+\gamma)\P.
\end{equation}

\subsubsection{Lie pairs associated to Lie algebra \#4}
\label{sec:spac-assoc-kinem-4}

The automorphism group of Lie algebra \#4 in
Table~\ref{tab:generic} is determined to be
\begin{equation}
  A = \left\{
    \begin{pmatrix}  a & b & \zero \\ c & d & \zero \\
      \zero & \zero & 1 \end{pmatrix} \quad \middle | \quad
    a,b,c,d \in \RR,~a d - bc \neq 0\right\} \cong \GL(2,\RR).
\end{equation}
Any vectorial subspace $\alpha \B + \beta \P$ is a subalgebra, but
under the automorphisms we can always bring it to $\B$.

\subsubsection{Lie pairs associated to Lie algebra \#5}
\label{sec:spac-assoc-kinem-5}

The automorphism group of Lie algebra \#5 in
Table~\ref{tab:generic} is determined to be
\begin{equation}
  A = \left\{
    \begin{pmatrix}  a & \zero & \zero \\ \zero & d & \zero \\
      \zero & \zero & 1 \end{pmatrix} \quad \middle | \quad
    a,d \in \RR,~a d \neq 0\right\}
  \bigcup
  \left\{
    \begin{pmatrix}  \zero & b & \zero \\ c & \zero & \zero \\
      \zero & \zero & -1 \end{pmatrix} \quad \middle | \quad
    b,c \in \RR,~b c \neq 0\right\}.
\end{equation}
Any vectorial subspace $\alpha \B + \beta \P$ is a subalgebra.  The
analysis of the action of automorphisms is very similar to that of
algebra \#3, except that using automorphisms not in the identity
component, we can relate the subspace spanned by $\B$ to that spanned
by $\P$.  Therefore we have two inequivalent Lie pairs, corresponding
to the Lie algebra in the original basis:
\begin{equation}
  [H,\B] = -\B \qquad\text{and}\qquad [H,\P] = \P,
\end{equation}
and the one corresponding to
\begin{equation}
  [H, \B] = \P \qquad\text{and}\qquad [H,\P] = \B.
\end{equation}

\subsubsection{Lie pairs associated to Lie algebra \#6}
\label{sec:spac-assoc-kinem-6}

The automorphism group of Lie algebra \#6 in
Table~\ref{tab:generic} is determined to be
\begin{equation}
  A = \left\{
    \begin{pmatrix}  a & \zero & \zero \\ c & a & \zero \\
      \zero & \zero & 1 \end{pmatrix} \quad \middle | \quad
    a,c \in \RR,~a \neq 0\right\}.
\end{equation}
Any vectorial subspace $\alpha \B + \beta \P$ is a subalgebra.  Using
the automorphisms, we may send $(\alpha,\beta) \mapsto (a \alpha,
c\alpha + a\beta)$.  We can distinguish between two cases.  If $\alpha
\neq 0$, then take $a = \alpha^{-1}$ and $c = -\alpha^{-2}\beta$ so
arrive at $(1,0)$, so that $W$ is spanned by $\B$.  If $\alpha = 0$,
then $W$ is spanned by $\P$.  In the former case, we have the original Lie
brackets
\begin{equation}
  [H,\B] = \B + \P \qquad\text{and}\qquad [H,\P] = \P;
\end{equation}
although we prefer to change basis so that $\P \mapsto -(\P + \B)$.  In
that basis, the brackets are now
\begin{equation}
  [H,\B] = -\P \qquad\text{and}\qquad [H,\P] = B + 2 \P.
\end{equation}
In the latter case, we change basis so that $W$ is again
spanned by $\B$, but then the Lie brackets are now
\begin{equation}
  [H, \B] = \B \qquad\text{and}\qquad [H,\P] = \B + \P.
\end{equation}

\subsubsection{Lie pairs associated to Lie algebra \#7$_\chi$}
\label{sec:spac-assoc-kinem-7}

The automorphism group of Lie algebra \#7$_\chi$ in
Table~\ref{tab:generic} is determined to be
\begin{equation}
  A = \left\{
    \begin{pmatrix}  a & b & \zero \\ -b & a & \zero \\
      \zero & \zero & 1 \end{pmatrix} \quad \middle | \quad
    a,b \in \RR,~a^2 + b^2 \neq 0\right\}.
\end{equation}
Any vectorial subspace $\alpha \B + \beta \P$ is a subalgebra.  Using
the automorphisms we can bring any $(\alpha,\beta)$ to $(1,0)$, so
that $W$ is spanned by $\B$.  We prefer to change basis in the Lie
algebra so that $\P \mapsto -(\P + \chi \B)$.  Doing so, the Lie
brackets become
\begin{equation}
  [H,\B] = -\P \qquad\text{and}\qquad [H,\P] = (1+\chi^2)\B + 2 \chi \P.
\end{equation}

\subsubsection{Lie pairs associated to Lie algebra \#8}
\label{sec:spac-assoc-kinem-8}

The automorphism group of Lie algebra \#8 in
Table~\ref{tab:generic} is determined to be
\begin{equation}
  A = \left\{
    \begin{pmatrix}  a & b & \zero \\ -b & a & \zero \\
      \zero & \zero & 1 \end{pmatrix} \quad \middle | \quad
    a,b \in \RR,~a^2 + b^2 \neq 0\right\}
  \bigcup
   \left\{
    \begin{pmatrix}  a & b & \zero \\ b & -a & \zero \\
      \zero & \zero & -1 \end{pmatrix} \quad \middle | \quad
    a,b \in \RR,~a^2 + b^2 \neq 0\right\}.
\end{equation}
Any vectorial subspace $\alpha \B + \beta \P$ is a subalgebra. As in
the previous case, using only the automorphisms in the identity
component we can already bring any $(\alpha,\beta)$ to $(1,0)$, so
that $W$ is spanned by $\B$.  There is thus a unique Lie pair
associated to this Lie algebra.

\subsubsection{Lie pairs associated to Lie algebra \#9}
\label{sec:spac-assoc-carroll}

Lie algebra \#9 in Table~\ref{tab:generic} is isomorphic to the
Carroll algebra.  Its automorphism group is determined to be
\begin{equation}
  A = \left\{
    \begin{pmatrix}  a & b & \zero \\ c & d & \zero \\
      \zero & \zero & a d - b c \end{pmatrix} \quad \middle | \quad
    a,b,c,d \in \RR,~ad - bc \neq 0\right\} \cong \GL(2,\RR).
\end{equation}
Any vectorial subspace $\alpha \B + \beta \P$ is a subalgebra since
\begin{equation}
  [\alpha B_a + \beta P_a, \alpha B_b + \beta P_b] = \alpha\beta (
  [B_a, P_b] + [P_a,B_b]) = \alpha \beta (\delta_{ab} H - \delta_{ba}
  H) = 0.
\end{equation}
However up to the automorphisms we can always bring $(\alpha,\beta)$
to $(1,0)$ and hence $W$ is spanned by $\B$.  There is a unique
homogeneous Lie pair associated to the Carroll algebra.

\subsubsection{Lie pairs associated to Lie algebra \#10$_\varepsilon$}
\label{sec:spac-assoc-euc-poin}

Lie algebra \#10$_{-1}$ is isomorphic to the euclidean Lie algebra,
whereas \#10$_{+1}$ is isomorphic to the Poincaré Lie algebra.  The
automorphism group is determined to be
\begin{equation}
    A = \left\{
    \begin{pmatrix}  1 & \zero & \zero \\ c & d & \zero \\
      \zero & \zero & d \end{pmatrix} \quad \middle | \quad
    c,d \in \RR,~d \neq 0\right\}
  \bigcup
   \left\{
    \begin{pmatrix}  -1 & \zero & \zero \\ c & d & \zero \\
      \zero & \zero & -d \end{pmatrix} \quad \middle | \quad
    c,d \in \RR,~d \neq 0\right\}.
\end{equation}
Any vectorial subspace $\alpha \B + \beta \P$ is admissible, since
\begin{equation}
  [\alpha\B + \beta\P, \alpha\B + \beta\P] = \varepsilon\alpha^2 \J~.
\end{equation}
Under the automorphisms, we
may bring $(\alpha, \beta)$ to $(\pm\alpha, c \alpha + d \beta)$.  We
must distinguish between two cases.  If $\alpha \neq 0$, we can bring
$(\alpha,\beta)$ to $(\alpha,0)$ with $\alpha >0$, which says that
$W$ is the span of $\B$.  If $\alpha = 0$, then $W$ is spanned by
$\P$.  In the former case, we have the original Lie brackets
\begin{equation}
  [H,\B] = -\varepsilon \P \qquad [\B,\P] = H \qquad\text{and}\qquad
  [\B,\B] = \varepsilon \J,
\end{equation}
whereas in the latter case we change basis so that $W$ is again
spanned by $\B$, but this changes the Lie brackets to
\begin{equation}
  [H,\P] = \varepsilon \B \qquad [\B,\P] = H \qquad\text{and}\qquad
  [\P,\P] = \varepsilon \J,
\end{equation}
where we have also changed the sign of $H$ in order to keep the
$[\B,\P] = H$ bracket uniform.  These Lie algebras are often known as
the \emph{para-Poincaré} and \emph{para-euclidean} algebras, depending
on the sign of $\varepsilon$.

\subsubsection{Lie pairs associated to Lie algebra \#11}
\label{sec:spac-assoc-sod}

Lie algebra \#11 in Table~\ref{tab:generic} is isomorphic to
$\so(D+1,1)$.  The automorphism group consists of two connected
components:
\begin{equation}
    A = \left\{
    \begin{pmatrix}  a & \zero & \zero \\ \zero & a^{-1} & \zero \\
      \zero & \zero & 1 \end{pmatrix} \quad \middle | \quad
    a \in \RRnZ\right\}
  \bigcup
   \left\{
    \begin{pmatrix} \zero & b & \zero \\ b^{-1} & \zero & \zero \\
      \zero & \zero & -1 \end{pmatrix} \quad \middle | \quad
    b \in \RRnZ \right\}.
\end{equation}
Every vectorial subspace $\alpha \B + \beta \P$ is admissible, since
\begin{equation}
  [\alpha\B + \beta\P, \alpha\B + \beta\P] = 2\alpha\beta \J~.
\end{equation}
If either $\alpha = 0$ or $\beta = 0$, then we can use the
automorphisms to bring $(\alpha,\beta)$ to $(1,0)$, which says that
$W$ is spanned by $\B$.  If $\alpha\beta \neq 0$, then we can use the
automorphisms to bring $(\alpha,\beta)$ to
$(\sqrt{|\alpha\beta|},\pm \sqrt{|\alpha\beta|})$, depending on whether
$\alpha\beta$ is positive or negative.  This says that $W$ is spanned
by $\B\pm \P$.  Changing basis so that $W$ is again spanned by $\B$,
and redefining $H$, we arrive at the following brackets
\begin{equation}
  [H,\B] = \varepsilon\P \qquad [H,\P] = \varepsilon\B \qquad [\B,\B] = - \varepsilon \J
  \qquad [\P,\P] = \varepsilon \J \qquad\text{and}\qquad [\B,\P] = H,
\end{equation}
where $\varepsilon = \pm 1$ according to whether $\pm \alpha\beta >0$.

\subsubsection{Lie pairs associated to Lie algebra \#12$_\varepsilon$}
\label{sec:spac-assoc-sod-1}

Lie algebra \#12$_\varepsilon$ in Table~\ref{tab:generic} is
isomorphic to $\so(D,2)$ for $\varepsilon=1$ and to $\so(D+2)$ for
$\varepsilon=-1$.  The automorphism group is isomorphic to $O(2)$ and
hence has two connected components:
\begin{equation}
  A = \left\{
    \begin{pmatrix}  a & b & \zero \\ -b & a & \zero \\
      \zero & \zero & 1 \end{pmatrix} \quad \middle | \quad
    a,b \in \RR,~a^2 + b^2 = 1\right\}
  \bigcup
   \left\{
    \begin{pmatrix}  a & b & \zero \\ b & -a & \zero \\
      \zero & \zero & -1 \end{pmatrix} \quad \middle | \quad
    a,b \in \RR,~a^2 + b^2 = 1\right\}.
\end{equation}
Any vectorial subspace $\alpha \B + \beta \P$ is admissible, since
\begin{equation}
  [\alpha \B + \beta \P, \alpha \B + \beta \P] = \varepsilon
  (\alpha^2+\beta^2) \J.
\end{equation}
However using only the automorphisms connected to the identity, we can
rotate $(\alpha,\beta)$ to $(\sqrt{\alpha^2+\beta^2},0)$ and hence $W$
is spanned by $\B$.  Therefore each of these kinematical Lie algebras
has a unique Lie pair.

\subsubsection{Summary}
\label{sec:summary-dgt3}

We summarise the above results in Table~\ref{tab:summarydge3}, which
lists the equivalence classes of kinematical Lie pairs which exist for
all $D\geq 3$.  In this and other similar tables throughout the paper,
each row consists of an incremental label ``LP\#'' for the Lie pair
for easy reference in the rest of the paper and also a label ``LA\#''
of the isomorphism type of kinematical Lie algebra $\k$ to which the
Lie pair is associated.  The rest of the row contains the Lie brackets
of $\k$ in a basis where the $\h$ is spanned by $J_{ab}$ and $B_a$ and
perhaps some relevant comments.  In some cases we have made changes of
basis (leaving alone the subalgebra $\h$) in order to arrive at a more
uniform description.

\begin{table}[h!]
  \centering
  \caption{Lie pairs for kinematical Lie algebras ($D\geq 3$)}
  \label{tab:summarydge3}
  \rowcolors{2}{blue!10}{white}
  \resizebox{\textwidth}{!}{
  \begin{tabular}{l|l|*{5}{>{$}l<{$}}|l} \toprule
   \multicolumn{1}{c|}{LP\#} & \multicolumn{1}{c|}{LA\#} & \multicolumn{5}{c|}{Nonzero Lie brackets in addition to $[\J,\J] = \J $, $[\J, \B] = \B$, $[\J,\P] = \P$}& \multicolumn{1}{c}{Comments}\\\midrule
    \hypertarget{LP1}{1} & \hyperlink{LA1}{1} & & & & & & static \\
    \hypertarget{LP2}{2} & \hyperlink{LA2}{2} & [H,\B] = -\P & & & & & galilean            \\
    \hypertarget{LP3}{3} & \hyperlink{LA2}{2} & & [H,\P] = \B & & & &   \\
    \hypertarget{LP4}{4$_\gamma$} & \hyperlink{LA3}{3$_\gamma$} & [H,\B] = \gamma \B & [H,\P] = \P & & & & $\gamma \in (-1,1)$ \\
    \hypertarget{LP5}{5$_\gamma$} & \hyperlink{LA3}{3$_\gamma$} & [H,\B] = \B & [H,\P] = \gamma \P & & & & $\gamma \in (-1,1)$ \\
    \hypertarget{LP6}{6$_\gamma$} & \hyperlink{LA3}{3$_\gamma$} & [H,\B] = -\P & [H,\P] = \gamma\B + (1+\gamma)\P & & & & $\gamma \in (-1,1)$ \\
    \hypertarget{LP7}{7} & \hyperlink{LA4}{4} & [H,\B] = \B & [H,\P] = \P & & & & \\
    \hypertarget{LP8}{8} & \hyperlink{LA5}{5} & [H,\B] = -\B & [H,\P] = \P & & & & \\
    \hypertarget{LP9}{9} & \hyperlink{LA5}{5} & [H,\B] = -\P & [H,\P] = -\B & & & & galilean dS  \\
    \hypertarget{LP10}{10} & \hyperlink{LA6}{6} & [H,\B] = -\P & [H, \P] = \B + 2 \P & & & &                     \\
    \hypertarget{LP11}{11} & \hyperlink{LA6}{6} & [H,\B] = \B & [H, \P] = \B + \P & & & &                     \\
    \hypertarget{LP12}{12$_\chi$} & \hyperlink{LA7}{7$_\chi$} & [H,\B] = - \P & [H,\P] = (1+\chi^2) \B + 2 \chi \P & & & & $\chi > 0$        \\
    \hypertarget{LP13}{13} & \hyperlink{LA8}{8} & [H,\B] = -\P & [H,\P] = \B & & & & galilean  AdS \\
    \hypertarget{LP14}{14} & \hyperlink{LA9}{9} & & & & [\B,\P] = H & & carrollian    \\
    \hypertarget{LP15p}{15$_{+1}$} & \hyperlink{LA10}{10$_{+1}$} & [H,\B] = -\P & & [\B,\B] = \J & [\B,\P] = H & & Minkowski    \\
    \hypertarget{LP15m}{15$_{-1}$} & \hyperlink{LA10}{10$_{-1}$} & [H,\B] = \P & & [\B,\B] = -\J & [\B,\P] = H & & euclidean            \\
    \hypertarget{LP16p}{16$_{+1}$} & \hyperlink{LA10}{10$_{+1}$} & & [H,\P] = \B & & [\B,\P] = H & [\P,\P] =  \J & carrollian AdS      \\
    \hypertarget{LP16m}{16$_{-1}$} & \hyperlink{LA10}{10$_{-1}$} & & [H,\P] = -\B & & [\B,\P] = H & [\P,\P] = -\J & carrollian dS  \\
    \hypertarget{LP17}{17} & \hyperlink{LA11}{11} & [H,\B] = \B & [H,\P] = -\P & & [\B,\P] = H + \J & &                     \\
    \hypertarget{LP18p}{18$_{+1}$} & \hyperlink{LA11}{11} & [H,\B] = \P & [H,\P] = \B & [\B,\B]= - \J & [\B,\P] =  H & [\P,\P]= \J & hyperbolic          \\
    \hypertarget{LP18m}{18$_{-1}$} & \hyperlink{LA11}{11} & [H,\B] = -\P & [H,\P] = -\B & [\B,\B]= \J & [\B,\P] = H & [\P,\P]= - \J & de~Sitter           \\
    \hypertarget{LP19p}{19$_{+1}$} & \hyperlink{LA12}{12$_{+1}$} & [H,\B] = -\P & [H,\P] = \B & [\B,\B]= \J & [\B,\P] = H & [\P,\P] = \J & anti~de~Sitter    \\
    \hypertarget{LP19m}{19$_{-1}$} & \hyperlink{LA12}{12$_{-1}$} & [H,\B] = \P & [H,\P] = -\B & [\B,\B]= - \J & [\B,\P] =  H & [\P,\P] = - \J & sphere     \\ \bottomrule
  \end{tabular}
   }
\end{table}

\subsection{Lie pairs unique to $D=3$}
\label{sec:three-dimensions}

Table~\ref{tab:3d} lists those kinematical Lie algebras which are
unique to $D=3$.  It is a sub-table of
\cite[Table~1]{Figueroa-OFarrill:2017ycu}, with small
modifications.\footnote{In \cite[Table~1]{Figueroa-OFarrill:2017ycu}
  the rotational generator $\boldsymbol{R}$ is a vector and is related
  to $\J$ in Table~\ref{tab:3d} by $J_{ab} = - \epsilon_{abc} R_c$.}
As usual we only list any nonzero Lie brackets in addition to
\eqref{eq:kin}.

\begin{table}[h!]
  \centering
  \caption{Kinematical Lie algebras unique to $D=3$}
  \label{tab:3d}
  \rowcolors{2}{blue!10}{white}
  \setlength{\extrarowheight}{2pt}  
  \begin{tabular}{l|*{4}{>{$}l<{$}}|l} \toprule
    \multicolumn{1}{c|}{LA\#} & \multicolumn{4}{c|}{Nonzero brackets in addition to $[\J,\J] = \J$, $[\J, \B] = \B$, $[\J,\P] = \P$} & Comments \\\midrule
    \hypertarget{LA13}{13$_\varepsilon$} & & & [\B,\B]= \B &  [\P,\P] = \varepsilon (\B-\J) & $\varepsilon=\pm1$ \\
    \hypertarget{LA14}{14} & & & [\B,\B] = \B & &  \\
    \hypertarget{LA15}{15} & & & [\B, \B] = \P & &  \\
    \hypertarget{LA16}{16} & & [H ,\P] = \P & [\B,\B] = \B & & \\
    \hypertarget{LA17}{17} & [H ,\B] = -\P & & [\B,\B] = \P & & \\
    \hypertarget{LA18}{18} & [H ,\B] = \B & [H ,\P] = 2\P & [\B,\B] = \P & & \\ \bottomrule
  \end{tabular}
\end{table}

For the kinematical Lie algebras in this table, the condition on the
vectorial subspace $W$ to be admissible is very restrictive and there
is no need to worry about the action of the automorphisms.  Therefore
we have no need to determine the automorphism groups.

\subsubsection{Lie pairs associated to Lie algebra \#13$_\varepsilon$}
\label{sec:four-dimens-spac-13}

The only linear combinations $\alpha \B + \beta \P$ which are
admissible are those with $\beta = 0$, so that $W$ is already spanned
by $\B$. Therefore there is a unique spacetime for each of these
kinematical Lie algebras.

\subsubsection{Lie pairs associated to Lie algebra \#14}
\label{sec:four-dimens-spac-14}

There are two admissible subspaces: the span of $\B$ and the span of
$\P$.  In the first case, the Lie brackets are as shown in the table,
whereas in the second case, changing basis so that $W$ is spanned by
$\B$, we find
\begin{equation}
  [\P, \P] = \P.
\end{equation}

\subsubsection{Lie pairs associated to Lie algebra \#15}
\label{sec:four-dimens-spac-15}

The only admissible subspace is the one spanned by $\P$.  Changing
basis so that it spanned by $\B$, we arrive at the Lie bracket
\begin{equation}
  [\P,\P] = \B.
\end{equation}

\subsubsection{Lie pairs associated to Lie algebra \#16}
\label{sec:four-dimens-spac-16}

There are two admissible subspaces $W$: the one spanned by $\B$ and
the one spanned by $\P$.  In the former case, the Lie brackets are as
shown in the table, whereas in the latter case, changing basis so that
$W$ is spanned by $\B$ again, we arrive at
\begin{equation}
  [H,\B]=\B \qquad\text{and}\qquad [\P,\P] = \P.
\end{equation}

\subsubsection{Lie pairs associated to Lie algebra \#17}
\label{sec:four-dimens-spac-17}

There is a unique admissible subspace: the one spanned by $\P$.
Changing basis so that it is spanned by $\B$, we arrive at the Lie
bracket
\begin{equation}
  [H,\P] =-\B \qquad\text{and}\qquad [\P,\P] = \B.
\end{equation}

\subsubsection{Lie pairs associated to Lie algebra \#18}
\label{sec:four-dimens-spac-18}

The span of $\P$ is the only admissible subspace.  Changing basis so
that it is spanned by $\B$, we arrive at
\begin{equation}
  [H,\B] = 2 \B \qquad [H,\P] = \P \qquad\text{and}\qquad [\P,\P] = \B.
\end{equation}

\subsubsection{Summary}
\label{sec:summary-deq3}

Table~\ref{tab:summarydeq3} lists the Lie pairs associated to the
kinematical Lie algebras which are unique to $D=3$.

\begin{table}[h!]\small
  \centering
  \caption{Lie pairs for kinematical Lie algebras unique to $D=3$}
  \label{tab:summarydeq3}
  \rowcolors{2}{blue!10}{white}
  \begin{tabular}{l|l|*{4}{>{$}l<{$}}|l} \toprule
   \multicolumn{1}{c|}{LP\#} & \multicolumn{1}{c|}{LA\#} &\multicolumn{4}{c|}{Nonzero Lie brackets in addition to $[\J,\J] = \J$, $[\J, \B] = \B$, $[\J,\P] = \P$} & \multicolumn{1}{c}{Comments} \\\midrule
    \hypertarget{LP20}{20$_\varepsilon$} & \hyperlink{LA13}{13$_\varepsilon$} & & & [\B,\B]= \B & [\P,\P]= \varepsilon(\B- \J) & $\varepsilon=\pm 1$\\
    \hypertarget{LP21}{21} & \hyperlink{LA14}{14} & & &  [\B,\B]= \B & & \\
    \hypertarget{LP22}{22} & \hyperlink{LA14}{14} & & & & [\P,\P]= \P & \\
    \hypertarget{LP23}{23} & \hyperlink{LA15}{15} & & & & [\P,\P]= \B & \\
    \hypertarget{LP24}{24} & \hyperlink{LA16}{16} & & [H,\P] = \P &  [\B,\B]= \B & &\\
    \hypertarget{LP25}{25} & \hyperlink{LA16}{16} &  [H,\B] = \B & & & [\P,\P] = \P & \\
    \hypertarget{LP26}{26} & \hyperlink{LA17}{17} & & [H,\P] = -\B & & [\P,\P]= \B & \\    
    \hypertarget{LP27}{27} & \hyperlink{LA18}{18} &  [H,\B] = 2\B & [H,\P] = \P & &  [\P,\P] = \B & \\ \bottomrule
  \end{tabular}
\end{table}

\subsection{Lie pairs for $D=2$}
\label{sec:two-dimensions}

Table~\ref{tab:2d} lists the kinematical Lie algebras with $D=2$ in
complex form.  It is borrowed partially from
\cite[Table~1]{Andrzejewski:2018gmz}, where we explain the rationale
for working with the complexified Lie algebras.  If $\k$ is a
kinematical Lie algebra in $D=2$, we let its complexification $\k_\CC$
be the complex span of $H,J,\B,\P,\Bbar,\Pbar$, where
$\B = B_1 + i B_2$, $\P = P_1 + i P_2$, $\Bbar = B_1 - i B_2$ and
$\Pbar = P_1 - i P_2$.  The standard rotational generator $J_{ab}$ is
related to $J$ by $J_{ab} = - \epsilon_{ab} J$, from where we see that
on a vectorial generator $[J,V_a] = \epsilon_{ab} V_b$, or
equivalently $[J,\V] = -i \V$ in complex form.  The Lie
bracket of $\k_\CC$ is obtained from that of $\k$ by extending it
complex linearly.  Given $\k_\CC$, we recover $\k$ as the real Lie
algebra which is fixed under the conjugation (i.e., complex antilinear
involutive automorphism) $\star$ defined on generators by
$H^\star = H$, $J^\star = J$, $\B^\star = \Bbar$ and
$\P^\star = \Pbar$.  To see how to translate between the complex and
real notations, the bracket $[\B,\Bbar] = i H$ is equivalent to
\begin{equation}
  [B_a,B_b] = -\tfrac12 \epsilon_{ab} H,
\end{equation}
whereas $[H,\B] = i \B$ is equivalent to
\begin{equation}
  [H,B_a] = - \epsilon_{ab} B_b,
\end{equation}
et cetera.  As usual, the Lie brackets \eqref{eq:kin}, characterising
kinematical Lie algebras, are implicit in every case.

\begin{table}[h!]
  \centering
  \rowcolors{2}{blue!10}{white}
  \setlength{\extrarowheight}{2pt}  
  \caption{Kinematical Lie algebras for $D=2$ (complex form)}
  \label{tab:2d}
    \resizebox{\textwidth}{!}{
  \begin{tabular}{l|>{$}l<{$}|*{5}{>{$}l<{$}}|l} \toprule
    \multicolumn{1}{c|}{LA\#} & \multicolumn{1}{c|}{$\cong$} & \multicolumn{5}{c|}{Nonzero Lie brackets in addition to $[J,\B]=-i\B$, $[J,\P] = - i \P$ and their c.c.} & \multicolumn{1}{c}{Comments} \\\midrule
    \hypertarget{LA19}{19} & \s & & & & & & \\
    \hypertarget{LA20}{20} & \g & [H,\B] = -\P & & & & & \\
    \hypertarget{LA21}{21} & & [H,\B] = \B & [H,\P] = \B + \P & & & & \\
    \hypertarget{LA22}{22$_{\gamma+ i\chi}$} & & [H,\B] = \B & [H,\P] = (\gamma + i \chi) \P & & & & $\gamma \in [-1,1]$, $\chi \geq 0$, $\gamma + i \chi \neq -1$ \\
    \hypertarget{LA23}{23} & \n_- & [H,\B] = \B & [H,\P] = -\P & & & & \\
    \hypertarget{LA24}{24} & \n_+ & [H,\B] = i \B & & & & & \\
    \hypertarget{LA25}{25} & \c & & & & [\B,\Pbar] = H & & \\
    \hypertarget{LA26}{26$_\varepsilon$} & \choice{\p}{\e} & & [H,\P] = \varepsilon\B & & [\B,\Pbar] = 2 H & [\P,\Pbar] = \varepsilon 2i J & $\varepsilon=\pm1$ \\
    \hypertarget{LA27}{27} & \so(3,1)  & [H,\B] = \B & [H,\P] = -\P & & [\B,\Pbar] = 2(J - i H) & & \\
    \hypertarget{LA28}{28$_\varepsilon$} & \choice{\so(4)}{\so(2,2)} & [H,\B] = \varepsilon \P & [H,\P] =  -\varepsilon \B & [\B,\Bbar] =  -\varepsilon 2 i J & [\B,\Pbar]= 2 H & [\P,\Pbar] =  -\varepsilon2i J & $\varepsilon=\pm 1$ \\
    \hypertarget{LA29}{29} & & & & [\B, \Bbar]=i H & & [\P, \Pbar] = i H & \\
    \hypertarget{LA30}{30} & & [H,\B] = i \B & & [\B,\Bbar] = i H & & [\P,\Pbar] = i (H+J) & \\
    \hypertarget{LA31}{31} & & & & [\B,\Bbar] = i H & & & \\
    \hypertarget{LA32}{32} & & [H,\B] = \P & & [\B,\Bbar] = i H & & & \\
    \hypertarget{LA33}{33$_\varepsilon$} & & [H,\B] = i \varepsilon \B & & [\B,\Bbar] = i H & & & $\varepsilon = \pm 1$ \\ \bottomrule
  \end{tabular}
  } 
\end{table}

Automorphisms are also different than for $D>2$ due to the fact that
$\r$ is one-dimensional.  If $\k$ is a kinematical Lie algebra in
Table~\ref{tab:2d}, then the group of automorphisms is a subgroup of
$\GL(2,\CC) \times \GL(2,\RR)$.  It consists of pairs
\begin{equation}
  \begin{pmatrix}
    a & b \\ c & d 
  \end{pmatrix} \in \GL(2,\CC) \qquad\text{and}\qquad
  \begin{pmatrix}
    r & t \\ s & 1
  \end{pmatrix} \in \GL(2,\RR)
\end{equation}
acting on the basis by
\begin{equation}
  (\B, \P) \mapsto (a\B + c\P, b\B + d\P) \qquad\text{and}\qquad (H,J)
  \mapsto (r H + s J, J + t H).
\end{equation}
(This uses the observation that for any of the Lie algebras in
Table~\ref{tab:2d}, the only elements in the real span of $H$
and $J$ which act on $\B$ and $\P$ as rotations have the form $J + t
H$ and $t$ can be nonzero only when $H$ is central, which  is the case
for the static and Carroll algebras and also for the kinematical Lie
algebras \#29 and \#31 in Table~\ref{tab:2d}.)

If $\h$ is an admissible subalgebra, its complexification $\h_\CC$ is
the complex span of
$\J + t H,\alpha\B + \beta \P, \bar\alpha\Bbar + \bar\beta\Pbar$, for
some $\alpha,\beta \in \CC$ not both zero and some $t \in \RR$. (Again
this uses the above mentioned observation that the rotational
generators are of the form $\J + t H$.)  The real subalgebra $\h$ is
then the real span of $\J + t H$ and the real and imaginary parts of
$\alpha \B + \beta \P$.  We let $W$ denote the two-dimensional real
vector space spanned by the real and imaginary parts of
$\alpha \B + \beta \P$, which transforms as a two-dimensional real
(vector) representation of $\r$.  We will, however, abbreviate this by
saying $W$ is spanned by $\alpha\B + \beta\P$.

\subsubsection{Lie pairs associated to Lie algebra \#19}
\label{sec:three-dimens-spac-19}

Any subspace $\alpha \B + \beta \P$ is admissible. The automorphism
group is the full group
\begin{equation}
  A = \left\{
    \begin{pmatrix}
      a & b & & \\ c & d & & \\ & & r & t \\ & & s & 1
    \end{pmatrix}
\quad \middle | \quad a,b,c,d \in \CC,~r,s,t \in \RR,~r(ad-bc) \neq 0\right\}.
\end{equation}
As in the case of $D\geq 3$, we may bring $(\alpha, \beta)$ to $(1,0)$
by an automorphism which preserves $H$ and $J$, so we can take $W$ to
be spanned by $\B$ without loss of generality.  There is thus a unique
spacetime associated to this Lie algebra.

\subsubsection{Lie pairs associated to the Lie algebra \#20}
\label{sec:three-dimens-spac-20}

Any subspace $\alpha \B + \beta \P$ is admissible.  The automorphism
group is now
\begin{equation}
  A = \left\{
    \begin{pmatrix}
      a & \zero & & \\ c & a r & & \\ & & r & \zero \\ & & \zero & 1
    \end{pmatrix}
\quad \middle | \quad a,c \in \CC,~r \in \RR,~a r \neq 0\right\}.
\end{equation}
Using the automorphisms, we can bring any nonzero $(\alpha,\beta) \in
\CC^2$ to either $(1,0)$ or $(0,1)$, depending on whether or not
$\alpha = 0$.  In the former case, $W$ is spanned by $\B$ and the
nonzero brackets are unchanged.  In the latter case, $W$ is spanned by
$\P$.  Changing basis so that it is again spanned by $\B$ changes the
nonzero brackets to $[H,\P] = -\B$.

\subsubsection{Lie pairs associated to Lie algebra \#21}
\label{sec:three-dimens-spac-21}

Again every nonzero subspace $\alpha \B + \beta \P$ is admissible.
The automorphisms are given by
\begin{equation}
  A = \left\{
    \begin{pmatrix}
      a & b & & \\ 0 & a & & \\ & & 1 & \zero \\ & & \zero & 1
    \end{pmatrix}
\quad \middle | \quad a,b\in \CC,~a \neq 0\right\}.
\end{equation}
Using automorphisms we can bring any nonzero $(\alpha,\beta)$ to one
of two vectors: $(1,0)$ or $(0,1)$, depending on whether or not
$\beta = 0$.  In the former case, $W$ is already spanned by $\B$ and
the nonzero brackets are unchanged.  In the latter case, $W$ is
spanned by $\P$, but changing basis so that it is again spanned by
$\B$, results in the following nonzero Lie brackets:
\begin{equation}
  [H,\B] = \B + \P \qquad\text{and}\qquad [H,\P] = \P.
\end{equation}

\subsubsection{Lie pairs associated to Lie algebras \#22$_{\gamma+i \chi}$ and \#23}
\label{sec:three-dimens-spac-22-23}

These two cases can be treated simultaneously, with brackets
\begin{equation}
  [H,\B] = \B \qquad\text{and}\qquad [H,\P] = \xi \P,
\end{equation}
for $\xi \in \CC$ given by $\xi = \gamma + i \chi$ with
$\gamma\in[-1,1]$ and $\chi\geq 0$.  For kinematical Lie algebra
\#22$_{\gamma+i\chi}$, $\xi \neq -1$, which corresponds to Lie algebra
\#23.  It is convenient to distinguish three cases:
\begin{enumerate}
\item $\xi = 1$;
\item $\Re\xi = -1$; and
\item $\xi\neq 1$ and $\Re\xi \neq -1$.
\end{enumerate}
In all cases, any nonzero subspace $\alpha \B + \beta \P$ is
admissible, but the automorphisms differ.

For $\xi = 1$, the automorphisms are given by
\begin{equation}
  A = \left\{
    \begin{pmatrix}
      a & b & & \\ c & d & & \\ & & 1 & \zero \\ & & \zero & 1
    \end{pmatrix}
    \quad \middle | \quad a,b,c,d\in \CC,~ad-bc \neq 0\right\},
\end{equation}
and it is clear that we can bring any nonzero $(\alpha,\beta) \in
\CC^2$ to $(1,0)$.  So without loss of generality we can take $W$ to
be spanned by $\B$ and the brackets are unchanged:
\begin{equation}
  [H,\B] = \B \qquad\text{and}\qquad [H,\P] = \P.
\end{equation}

For $\Re\xi = -1$, the automorphisms are given by
\begin{equation}
  A = \left\{
    \begin{pmatrix}
      a & \zero & & \\ \zero & d & & \\ & & 1 & \zero \\ & & \zero & 1
    \end{pmatrix}
    \quad \middle | \quad a,d\in \CC,~ ad \neq 0\right\} \bigcup \left\{
    \begin{pmatrix}
      \zero & b & & \\ c & \zero & & \\ & & -1 & \zero \\ & & -\chi & 1
    \end{pmatrix}
    \quad \middle | \quad b,c\in \CC,~bc \neq 0\right\}.
\end{equation}
Any nonzero $(\alpha,\beta) \in \CC^2$ may be brought to one of two
normal forms: $(1,0)$ or $(1,1)$ according to whether one of $\alpha$
or $\beta$ are zero or neither are zero, respectively.  In the first 
case, $W$ is already spanned by $\B$ and there is no need to change
basis.  In the second case, $W$ is spanned by $\B + \P$ and after the
change of basis, the new brackets are
\begin{equation}
  [H,\B] = \P \qquad\text{and}\qquad [H,\P] = -\xi \B + (1+\xi) \P.
\end{equation}

Finally, if $\xi \neq 1$ and $\Re\xi \neq -1$, the automorphisms are
\begin{equation}
  A =\left\{
    \begin{pmatrix}
      a & \zero & & \\ \zero & d & & \\ & & 1 & \zero \\ & & \zero & 1
    \end{pmatrix}
    \quad \middle | \quad a,d\in \CC,~ ad \neq 0\right\}.
\end{equation}
Again we can bring any nonzero $(\alpha,\beta)$ to one of three normal
forms $(1,0)$ (if $\beta = 0$), $(0,1)$ (if $\alpha = 0$) or $(1,1)$
otherwise.  In the first case there is no need to change basis.  In
the second case, $W$ is spanned by $\P$ and changing basis results in
\begin{equation}
  [H,\B] = \xi B \qquad\text{and}\qquad [H,\P] = \P.
\end{equation}
Finally, in the last case, $W$ is spanned by $\B + \P$.  Changing
basis results again in
\begin{equation}
  [H,\B] = \P \qquad\text{and}\qquad [H,\P] = -\xi \B + (1+\xi) \P.
\end{equation}

\subsubsection{Lie pairs associated to Lie algebra \#24}
\label{sec:three-dimens-spac-24}

Again every nonzero subspace $W$ is admissible.  The automorphism
group has now two connected components:
\begin{equation}
  A = \left\{
    \begin{pmatrix}  a & \zero & & \\ \zero & d & & \\
      & & 1 & \zero \\ & & \zero & 1 \end{pmatrix} \quad \middle | \quad
    a,d \in \CC,~ad \neq 0\right\}
  \bigcup
   \left\{
    \begin{pmatrix}  \zero & b & & \\ c & \zero & & \\
      & & -1 & \zero \\ & & -1 & 1 \end{pmatrix} \quad \middle | \quad
    b,c \in \CC,~bc \neq 0\right\}.
\end{equation}
Via an automorphism, any nonzero $(\alpha,\beta) \in \CC^2$ can be
brought to either $(1,0)$ if at least one of $\alpha$ or $\beta$ is
zero, or $(1,1)$ otherwise.  In the former case, $W$ is already
spanned by $\B$ and the brackets are unchanged, whereas in the latter
case $W$ is spanned by $\B + \P$ and changing basis changes the
brackets to
\begin{equation}
  [H,\B] = -\P \qquad\text{and}\qquad [H,\P] = \B.
\end{equation}

\subsubsection{Lie pairs associated to Lie algebra \#25}
\label{sec:three-dimens-spac-25}

A subspace $\alpha \B + \beta \P$ is admissible if and only if
$\alpha\bar\beta \in \RR$: indeed,
\begin{equation}
  [\alpha \B + \beta \P, \bar\alpha \Bbar + \bar\beta\Pbar] =
  (\alpha\bar\beta - \bar\alpha \beta) H.
\end{equation}
The condition $\alpha\bar\beta \in \RR$ means that we can choose
$\alpha$ and $\beta$ real and still span the same subspace.  It is
unnecessary to determine the precise automorphism group.  It suffices
to remark that it includes the subgroup
\begin{equation}
  A = \left\{
    \begin{pmatrix}
      a & b & & \\ c & d & & \\ & & ad-bc & \zero \\ & & \zero & 1
    \end{pmatrix}
\quad \middle | \quad a,b,c,d \in \RR,ad-bc \neq 0\right\},
\end{equation}
and that using such an automorphism, we can take any
$(\alpha,\beta)\in\RR^2$ to $(1,0)$.  There is thus a unique
spacetime.

\subsubsection{Lie pairs associated to Lie algebra \#26$_\varepsilon$}
\label{sec:three-dimens-spac-26}

A subspace $\alpha \B + \beta \P$ is admissible if and only if $\alpha
\bar\beta$ is real:
\begin{equation}
  [\alpha \B + \beta \P, \bar\alpha \Bbar + \bar\beta \Pbar] = 2 (\alpha
  \bar\beta - \bar\alpha\beta) H + 2 i \varepsilon |\beta|^2 J.
\end{equation}
As in the previous case, we can take $\alpha,\beta$ both real without
altering their span.

The automorphism group consists of
\begin{equation}
  A = \left\{
    \begin{pmatrix}
      r d & u d & & \\ \zero & d & & \\ & & r & \zero \\ & & \zero & 1
    \end{pmatrix}
\quad \middle | \quad r,u \in\RR,~d \in \CC,~ r\neq 0,~|d|=1\right\}.
\end{equation}
Under automorphisms we can bring every admissible subspace to either
the span of the $\B$ or the span of the $\P$.  In the former case we
do not change the brackets, whereas in the latter case we change basis
and arrive at
\begin{equation}
  [H,\B] = \varepsilon\P \qquad [\B,\Pbar] = -2 H
  \qquad\text{and}\qquad [\B,\Bbar] = 2\varepsilon i J.
\end{equation}

\subsubsection{Lie pairs associated to Lie algebra \#27}
\label{sec:three-dimens-spac-27}

Because of the bracket
\begin{equation}
  [\alpha \B + \beta \P, \bar\alpha \Bbar + \bar\beta \Pbar] = 2
  (\alpha \bar\beta - \bar\alpha \beta) J - 2 i (\alpha \bar\beta +
  \bar\alpha \beta) H,
\end{equation}
a subspace $\alpha \B + \beta \P$ is admissible if and only if
$\alpha\bar\beta + \bar\alpha\beta = 0$.

The automorphism group in this case can be determined to consist of
two connected components
\begin{equation}
  A = \left\{
    \begin{pmatrix}  a & \zero & & \\ \zero & {\bar a}^{-1} & & \\
      & & 1 & \zero \\ & & \zero & 1 \end{pmatrix} \quad \middle | \quad
    a \in \CC,~a \neq 0\right\}
  \bigcup
   \left\{
    \begin{pmatrix}  \zero & b & & \\ -{\bar b}^{-1} & \zero & & \\
      & & -1 & \zero \\ & & \zero & 1 \end{pmatrix} \quad \middle | \quad
    b \in \CC,~b \neq 0\right\}.
\end{equation}

Under an automorphism in the identity component, $(\alpha,\beta)
\mapsto (a \alpha, \beta/\bar a)$, whereas under an automorphism in
the other component $(\alpha,\beta) \mapsto (b\beta, -\alpha/\bar
b)$.  Notice that the product $\alpha\bar\beta$ remains invariant
under automorphisms, provided that $(\alpha,\beta)$ is admissible, so
that $\alpha\bar\beta = - \bar\alpha\beta$.

Suppose that $\alpha\bar\beta = 0$.  Then either $\alpha = 0$ or
$\beta = 0$ and we may apply an automorphism to send such
$(\alpha,\beta)$ to $(1,0)$.  On the other hand, if $\alpha\bar\beta =
i t$ for some real $t \neq 0$, we can always choose an automorphism to
bring $(\alpha,\beta)$ to a multiple of $(1,\pm i)$.  The former case,
$W$ is spanned by $\B$ and in the latter cases by $\B \pm i \P$.  In
these latter cases, changing basis so that $W$ is spanned again by
$\B$, we find the following brackets:
\begin{equation}
  [H,\B] = \P \qquad [H,\P] = \B \qquad [\B,\Bbar] = \mp i J \qquad
  [\P,\Pbar] = \pm i J \qquad\text{and}\qquad [\B,\Pbar] = \pm H.
\end{equation}

\subsubsection{Lie pairs associated to Lie  algebra \#28$_\varepsilon$}
\label{sec:three-dimens-spac-28}

A subspace $\alpha \B + \beta \P$ is admissible if and only if $\alpha
\bar\beta \in \RR$.  Such an admissible subspace is therefore always
of the form $\alpha \B + \beta \P$ with $\alpha,\beta \in \RR$.
It is not necessary to determine the precise form of the automorphism
group, but simply to remark that it contains the following $\SO(2)$
subgroup:
\begin{equation}
  \left\{
    \begin{pmatrix}
      \cos\theta & -\sin\theta & & \\ \sin\theta & \cos\theta & & \\ & & 1 & \zero \\ & & \zero & 1
    \end{pmatrix}
\quad \middle | \quad \theta \in\RR\right\}.
\end{equation}
Any $(\alpha,\beta) \in \RR^2$ can be rotated to
$(\sqrt{\alpha^2+\beta^2},0)$ which corresponds to the subspace
spanned by $\B$.  Therefore there is a unique spacetime for each sign
$\varepsilon$.

\subsubsection{No Lie pairs associated to Lie algebras \#29  and \#30}
\label{sec:three-dimens-spac-29-30}

There is no admissible subalgebra in these cases, since in one case
\begin{equation}
  [\alpha \B + \beta \P, \bar\alpha \Bbar + \bar\beta \Pbar] = i
  (|\alpha|^2 + |\beta|^2) H
\end{equation}
and in the other
\begin{equation}
  [\alpha \B + \beta \P, \bar\alpha \Bbar + \bar\beta \Pbar] = i
  (|\alpha|^2 + |\beta|^2) H + i |\beta|^2 J,
\end{equation}
whereas the right-hand sides do not span a rotational subalgebra.
Therefore there are no spacetimes associated these kinematical Lie
algebras.

\subsubsection{Lie pairs associated to Lie algebras \#31, \#32 and \#33$_\varepsilon$}
\label{sec:three-dimens-spac-31-32-33}

We can treat these cases simultaneously, since they only differ in the
adjoint action of $H$.  In all cases, the only admissible subspace is
the one spanned by $\P$, since
\begin{equation}
  [\alpha \B + \beta \P, \bar\alpha \Bbar + \bar\beta \Pbar] = i
  |\alpha|^2 H,
\end{equation}
and $H$ does not span a rotational subalgebra.  Changing basis so that
it is spanned by $\B$, changes the brackets to
\begin{equation}
  [H,\P] =
  \begin{cases}
    0 \\ \B \\ i \varepsilon \P
  \end{cases}
  \qquad\text{and}\qquad
  [\P,\Pbar] = i H.
\end{equation}

\subsubsection{Summary}
\label{sec:summary-deq2}

In Table~\ref{tab:summary2d-real}, we list the Lie pairs associated to
the kinematical Lie algebras in $D=2$. We have reverted to the real
form of the $D=2$ Lie algebras. The notation is then very similar to
$D\geq 3$.  In particular, the rotational generator $\J$ is again
$J_{ab}$ as in $D\geq 3$.  This eases the comparison with $D\geq 3$
and allows us to quickly determine which Lie pairs are unique to $D=2$
and which are just the $D=2$ avatars of Lie pairs which exist also for
$D>2$.  The main difference between $D=2$ and $D>2$ is that some
brackets feature the $\epsilon$ symbol and we have therefore
introduced $\Pt$ and $\Bt$ as explained at the start of
Section~\ref{sec:summary-main-results}.  We have also changed basis in
some cases in order to uniformise the presentation and ease the
comparison with $D>2$.

The table is divided into two: the ones above the horizontal line are
$D=2$ avatars of Lie pairs which exist also for $D>2$.  The
precise dictionary is shown in Table~\ref{tab:deq2correspondence}.  The
ones below the line exist only for $D=2$.

\begin{table}[h!]
  \centering
  \caption{Lie pairs for kinematical Lie algebras ($D=2$, real form)}
  \label{tab:summary2d-real}
  \resizebox{\textwidth}{!}{
  \rowcolors{2}{blue!10}{white}
  \begin{tabular}{l|l|*{5}{>{$}l<{$}}|l} \toprule
    \multicolumn{1}{c|}{LP\#}&  \multicolumn{1}{c|}{LA\#}& \multicolumn{5}{c|}{Nonzero Lie brackets in addition to $[\J, \B] = \B$ and $[\J,\P] = \P$} & \multicolumn{1}{c}{Comments} \\\midrule
     \hypertarget{LP28}{28} & \hyperlink{LA19}{19} & & & & & & static\\
     \hypertarget{LP29}{29} & \hyperlink{LA20}{20} &  [H,\B] = -\P & & & & & galilean\\
     \hypertarget{LP30}{30} & \hyperlink{LA20}{20} & & [H,\P] = \B & & & & \\
     \hypertarget{LP31}{31} & \hyperlink{LA21}{21} &  [H,\B] = - \P & [H, \P] = \B + 2 \P & & & & \\
     \hypertarget{LP32}{32} & \hyperlink{LA21}{21} &  [H,\B] = \B & [H, \P] = \B + \P & & & & \\
     \hypertarget{LP33}{33} & \hyperlink{LA22}{22$_{1+i0}$} &  [H,\B] = \B & [H, \P] = \P & & & & \\
     \hypertarget{LP34}{34} & \hyperlink{LA23}{23} &  [H,\B] = -\B & [H,\P] = \P & & & & \\
     \hypertarget{LP35}{35} & \hyperlink{LA23}{23} &  [H,\B] = -\P & [H,\P] = -\B & & & & galilean dS\\
     \hypertarget{LP36}{36$_\gamma$} & \hyperlink{LA22}{22$_{\gamma+i0}$} & [H,\B] = \B & [H,\P] = \gamma \P & & & & $-1<\gamma<1$\\
     \hypertarget{LP37}{37$_\gamma$} & \hyperlink{LA22}{22$_{\gamma+i0}$} &  [H,\B] = \gamma \B & [H,\P] = \P & & & &  $-1<\gamma<1$\\
     \hypertarget{LP38}{38$_\gamma$} & \hyperlink{LA22}{22$_{\gamma+i0}$} &  [H,\B] = -\P & [H,\P] = \gamma \B + (1+\gamma) \P & & & &  $-1<\gamma<1$\\
     \hypertarget{LP39}{39} & \hyperlink{LA24}{24} &  [H,\B] = -\P & [H,\P] = \B & & & & galilean AdS\\
     \hypertarget{LP40}{40} & \hyperlink{LA25}{25} & & & & [\B,\P] = H & & carrollian\\
     \hypertarget{LP41}{41} & \hyperlink{LA26}{26$_{+1}$} & & [H,\P] = \B & & [\B,\P] =  H & [\P,\P] = \J & carrollian AdS\\
     \hypertarget{LP42}{42} & \hyperlink{LA26}{26$_{-1}$} & & [H,\P] = - \B & & [\B,\P] = H & [\P,\P] = -\J & carrollian dS\\
     \hypertarget{LP43}{43} & \hyperlink{LA26}{26$_{+1}$} &  [H,\B] = -\P & &  [\B,\B] = \J & [\B,\P] = H & & Minkowski\\
     \hypertarget{LP44}{44} & \hyperlink{LA26}{26$_{-1}$} &  [H,\B] = \P & & [\B,\B] = -\J & [\B,\P] = H & & euclidean\\
     \hypertarget{LP45}{45} & \hyperlink{LA27}{27} &  [H,\B] = \B & [H,\P] = -\P & & [\B,\P] = H + \J & & \\
     \hypertarget{LP46}{46} & \hyperlink{LA27}{27} &  [H,\B] = \P & [H,\P] = \B &  [\B,\B] = - \J & [\B,\P]= H & [\P,\P]= \J& hyperbolic\\
     \hypertarget{LP47}{47} & \hyperlink{LA27}{27} &  [H,\B] = -\P & [H,\P] = -\B &  [\B,\B] = \J & [\B,\P]= H & [\P,\P]= - \J& de~Sitter\\
     \hypertarget{LP48}{48} & \hyperlink{LA28}{28$_{+1}$} &  [H,\B] = \P & [H,\P] = -\B &  [\B,\B]= -\J & [\B,\P] = H &  [\P,\P] = -\J & sphere\\
     \hypertarget{LP49}{49} & \hyperlink{LA28}{28$_{-1}$} &  [H,\B] = -\P & [H,\P] = \B &  [\B,\B]= \J & [\B,\P] = H &  [\P,\P] = \J & anti~de~Sitter\\
     \hypertarget{LP50}{50$_\chi$} & \hyperlink{LA22}{22$_{1+i\chi}$} &  [H,\B] = -\P & [H,\P] = (1+4\chi^{-2}) \B + 4 \chi^{-1}\P & & & &  $\chi > 0$\\\midrule
     \hypertarget{LP51}{51$_\chi$} & \hyperlink{LA22}{22$_{1+i\chi}$} &  [H,\B] = \B - \chi \Bt & [H,\P] = \P & & & &  $\chi > 0$\\
     \hypertarget{LP52}{52$_\chi$} & \hyperlink{LA22}{22$_{1+i\chi}$} & [H,\B] = \B & [H,\P] = \P - \chi \Pt & & & & $\chi > 0$\\
     \hypertarget{LP53}{53$_{\gamma,\chi}$} & \hyperlink{LA22}{22$_{\gamma+i\chi}$} & [H,\B] = \B & [H,\P] = \gamma \P - \chi \Pt & & & & $-1\leq\gamma<1, \chi >0$\\
     \hypertarget{LP54}{54$_{\gamma,\chi}$} & \hyperlink{LA22}{22$_{\gamma+i\chi}$} &  [H,\B] = \gamma \B - \chi \Bt & [H,\P] = \P & & & &  $-1\leq\gamma<1, \chi >0$\\
     \hypertarget{LP55}{55$_{\gamma,\chi}$} & \hyperlink{LA22}{22$_{\gamma+i\chi}$} &  [H,\B] = - \P & [H,\P] = (1+\gamma) \P - \chi \Pt + \gamma \B - \chi \Bt & & & &  $-1\leq\gamma<1, \chi >0$\\
     \hypertarget{LP56}{56} & \hyperlink{LA24}{24} &  [H,\B] = - \Bt & & & & & \\
     \hypertarget{LP57}{57} & \hyperlink{LA31}{31} & & & & & [\P,\P] = H &\\
     \hypertarget{LP58}{58} & \hyperlink{LA32}{32} & & [H,\P] = \B & & & [\P,\P] = - H &\\
     \hypertarget{LP59}{59} & \hyperlink{LA33}{33$_{+1}$} & & [H,\P] = \Pt & & & [\P,\P] = H &\\
     \hypertarget{LP60}{60} & \hyperlink{LA33}{33$_{-1}$} & & [H,\P] = \Pt & & & [\P,\P] = - H &\\ \bottomrule
  \end{tabular}}
\end{table}

\begin{table}[h!]
  \centering
  \caption{Correspondence between $D=2$ and $D\geq 3$ Lie pairs}
  \label{tab:deq2correspondence}
  \begin{tabular}{l|l} \toprule
    \multicolumn{1}{c|}{$D=2$} & \multicolumn{1}{c}{$D\geq 3$}\\\midrule
    28 & 1\\
    29 & 2\\
    30 & 3\\
    31 & 10\\
    32 & 11\\
    33 & 7\\
    34 & 8\\
    35 & 9 \\\bottomrule
  \end{tabular}
  \hspace{3cm}
  \begin{tabular}{l|l} \toprule
    \multicolumn{1}{c|}{$D=2$} & \multicolumn{1}{c}{$D\geq 3$}\\\midrule
    36$_\gamma$ & 5$_\gamma$\\
    37$_\gamma$ & 4$_\gamma$\\
    38$_\gamma$ & 6$_\gamma$\\
    39 & 13\\
    40 & 14\\
    41 & 16$_{+1}$\\
    42 & 16$_{-1}$\\ \bottomrule
  \end{tabular}
  \hspace{3cm}
  \begin{tabular}{l|l} \toprule
    \multicolumn{1}{c|}{$D=2$} & \multicolumn{1}{c}{$D\geq 3$}\\\midrule
    43 & 15$_{+1}$\\
    44 & 15$_{-1}$\\
    45 & 17\\
    46 & 18$_{+1}$\\
    47 & 18$_{-1}$\\
    48 & 19$_{-1}$\\
    49 & 19$_{+1}$\\
    50$_\chi$ & 12$_{2/\chi}$ \\\bottomrule
  \end{tabular}
\end{table}

\subsection{Lie pairs for $D=1$}
\label{sec:one-dimension}

Finally, we consider the Lie pairs associated to the $D=1$ kinematical
Lie algebras.  Since there are no rotations in $D=1$, every
three-dimensional real Lie algebra is kinematical.  Such Lie algebras
were classified by Bianchi \cite{Bianchi} as part of his
classification of three-dimensional homogeneous manifolds.  In this
section, however, they will be associated to two-dimensional
homogeneous spaces.  The Bianchi classification of three-dimensional
real Lie algebras is recalled in Table~\ref{tab:bianchi}.  We have
omitted Bianchi III because it is isomorphic to Bianchi VI$_1$.

\begin{table}[h!]
  \caption{Kinematical Lie algebras for $D=1$}
  \label{tab:bianchi}
  \rowcolors{2}{blue!10}{white}
  \centering
  \begin{tabular}{l|>{$}l<{$}|*{3}{>{$}l<{$}}|l} \toprule
    \multicolumn{1}{c|}{Bianchi} & \multicolumn{1}{c|}{$\cong$} &
    \multicolumn{3}{c|}{Nonzero brackets} & \multicolumn{1}{c}{Comments} \\\midrule
    \hypertarget{BianchiI}{I} & \s & & & & \\
    \hypertarget{BianchiII}{II} & \g \cong \c & & & [\be_2,\be_3]=\be_1 & \\
    \hypertarget{BianchiIV}{IV} & & & [\be_1,\be_3]=\be_1 &  [\be_2,\be_3]=\be_1+\be_2 & \\
    \hypertarget{BianchiV}{V} & & & [\be_1,\be_3]=\be_1 &  [\be_2,\be_3]=\be_2 & \\
    \hypertarget{BianchiVI0}{VI$_0$} & \n_- \cong \p & & [\be_1,\be_3]=- \be_1 & [\be_2,\be_3]=\be_2 & \\
    \hypertarget{BianchiVI}{VI$_\chi$} & & & [\be_1,\be_3]=(\chi-1) \be_1 & [\be_2,\be_3]=(\chi+1) \be_2 & $\chi>0$ \\
    \hypertarget{BianchiVII0}{VII$_0$} & \n_+ \cong \e & & [\be_1,\be_3]=- \be_2 & [\be_2,\be_3]=\be_1 & \\
    \hypertarget{BianchiVII}{VII$_\chi$} & &  & [\be_1,\be_3]=\chi \be_1 - \be_2 & [\be_2,\be_3]=\be_1 + \chi \be_2 & $\chi>0$ \\
    \hypertarget{BianchiVIII}{VIII} & \so(1,2) &  [\be_1,\be_2]=-\be_3 &  [\be_1,\be_3]=-\be_2 & [\be_2,\be_3]=\be_1 & \\
    \hypertarget{BianchiIX}{IX} & \so(3) &  [\be_1,\be_2]=\be_3 & [\be_1,\be_3]=-\be_2 & [\be_2,\be_3]=\be_1 & \\ \bottomrule
  \end{tabular}
\end{table}

Because of dimension and the absence of rotations, every
one-dimensional subspace is an admissible subalgebra.  Therefore to
classify pairs $(\k,\h)$ we need, for each $\k$ in the Bianchi
classification, classify the one-dimensional subspaces up to the
action of automorphisms.  In other words, we must determine the orbits
of the action of the automorphism group of $\k$ on the projective
space $\PP\k$ of one-dimensional subspaces of $\k$.  We will now do
this for every Bianchi type in turn.  We will choose a basis for $\k$
where $\h$ is spanned by $B$ and will choose $P,H$, if possible, so
that their span is $\h$-invariant.  The resulting Lie pairs are
summarised in Table~\ref{tab:summary1d} below.  Those Lie pairs below
the horizontal line correspond to Lie algebras unique to $D=1$;
although the Lie pair might actually be the $1+1$ case of a Lie pair
which exists for all $D$.

One notational remark about the automorphism groups.  Since we have
chosen a basis $(\be_1,\be_2,\be_3)$ for $\k$, we can identify the
automorphism group $A = \Aut(\k)$ with a subgroup of $\GL(3,\RR)$ and
therefore we will be describing it as a set of matrices.

\subsubsection{Bianchi I}
\label{sec:bianchi-i}

Since $\k$ is abelian, every invertible linear map is an
automorphism.  The general linear group acts transitively on the
projective space, so we can take $\h$ to be spanned by $B= \be_1$,
say.

\subsubsection{Bianchi II}
\label{sec:bianchi-ii}

Here $\k$ is the Heisenberg algebra, whose automorphisms can be
determined to be
\begin{equation}
  A = \left\{
    \begin{pmatrix}
      ad - bc & \alpha & \beta \\
      \zero & a & b\\
      \zero & c & d\\
    \end{pmatrix}
    \middle | a,b,c,d,\alpha,\beta \in \RR,~ad \neq bc \right\}.
\end{equation}

A nonzero vector $x^1 \be_1 + x^2 \be_2 + x^3 \be_3$ determines a
one-dimensional subspace and hence an admissible subalgebra.  Under a
typical automorphism, this vector transforms as
\begin{equation}
  \begin{pmatrix}
    x^1 \\ x^2 \\ x^3
  \end{pmatrix} \mapsto
  \begin{pmatrix}
    ad - bc & \alpha & \beta \\
    \zero & a & b\\
    \zero & c & d\\
  \end{pmatrix}
  \begin{pmatrix}
    x^1 \\ x^2 \\ x^3
  \end{pmatrix} =
  \begin{pmatrix}
    (ad - bc) x^1 + \alpha x^2 + \beta x^3  \\ a x^2 + b x^3 \\ c x^2
    + d x^3
  \end{pmatrix} 
\end{equation}

The line spanned by $\be_1$ is sent to itself, whereas any other line
can be transformed to the line spanned by $\be_2$, say.

In the first case we let $B = \be_1$, $P= \be_2$ and $H = \be_3$,
arriving at the Lie algebra with nonzero brackets
\begin{equation}
  [P,H] = B.
\end{equation}
  
In the second case we let $B = \be_2$, $P = \be_1$ and $H = \be_3$,
arriving at
\begin{equation}
  [B,H] = P.
\end{equation}

\subsubsection{Bianchi IV}
\label{sec:bianchi-iv}

The automorphism group is now
\begin{equation}
  A = \left\{
    \begin{pmatrix}
      a & b & \alpha \\
      \zero & a & \beta\\
      \zero & \zero & 1\\
    \end{pmatrix}
    \middle | a,b,\alpha,\beta \in \RR,~a \neq 0 \right\}.
\end{equation}
There are three $A$-orbits in $\PP\k$, labelled by the vectors
$(1,0,0)$, $(0,1,0)$ and $(0,0,1)$.

In the first case, $B = \be_1$, $P = \be_2$ and $H = \be_3$, with
nonzero brackets
\begin{equation}
  [B,H] = B \qquad\text{and}\qquad [P,H] = B + P.
\end{equation}

In the second case, $B= \be_2$, $P= \be_1$ and $H= \be_3$, with
nonzero brackets
\begin{equation}
  [B, H] = P \qquad\text{and}\qquad [P,H] = 2P - B.
\end{equation}

In the final case, $B = \be_3$, $P= \be_1$ and $H = \be_2$, with
nonzero brackets
\begin{equation}
  [B,P] = -P \qquad\text{and}\qquad [B,H] = - P - H.
\end{equation}

\subsubsection{Bianchi V}
\label{sec:bianchi-v}

The automorphism group in this case is
\begin{equation}
    A = \left\{
    \begin{pmatrix}
      a & b & \alpha\\
      c & d & \beta\\
      \zero & \zero & 1\\
    \end{pmatrix}
    \middle | a,b,c,d,\alpha,\beta \in \RR,~ad \neq bc \right\}.
\end{equation}

There are two $A$-orbits in $\PP\k$ labelled by the vectors $(1,0,0)$
and $(0,0,1)$.

In the first case, $B = \be_1$, $P = \be_2$ and $H = \be_3$, with
nonzero brackets
\begin{equation}
  [B,H] = B \qquad\text{and}\qquad [P,H] = P.
\end{equation}

In the other case, $B = \be_3$, $P = \be_1$ and $H = \be_2$, with
nonzero brackets
\begin{equation}
  [B,P] = - P \qquad\text{and}\qquad [B,H] = - H.
\end{equation}

\subsubsection{Bianchi VI$_0$}
\label{sec:bianchi-vi0}

The automorphism group is this case is
\begin{equation}
    A = \left\{
    \begin{pmatrix}
      a & \zero & \alpha \\
      \zero & d & \beta\\
      \zero & \zero & 1\\
    \end{pmatrix}
    \middle | a,d,\alpha,\beta \in \RR,~ad \neq 0 \right\} \bigcup
   \left\{
    \begin{pmatrix}
      \zero & b & \alpha \\
      c & \zero & \beta\\
      \zero & \zero & -1\\
    \end{pmatrix}
    \middle | b,c,\alpha,\beta \in \RR,~bc \neq 0 \right\}.
\end{equation}

There are three $A$-orbits in $\PP\k$, labelled by $(0,0,1)$,
$(1,0,0)$ and $(1,1,0)$.

In the first case, $B = \be_3$, $P= \be_1$ and $H = \be_2$, with
nonzero brackets
\begin{equation}
  [B,H] = - H \qquad\text{and}\qquad [B,P] = P.
\end{equation}

In the second case, $B = \be_1$, $P = \be_2$ and $H = \be_3$, with
nonzero brackets
\begin{equation}
  [B,H] = - B \qquad\text{and}\qquad [P,H] = P.
\end{equation}
This spacetime is not reductive.

In the final case, $B = \be_1 + \be_2$, $P = -\be_1 + \be_2$ and $H
= \be_3$, with nonzero brackets
\begin{equation}
  [B,H] = P \qquad\text{and}\qquad [P,H] = B.
\end{equation}

\subsubsection{Bianchi VI$_{\chi>0}$}
\label{sec:bianchi-vic}

The automorphism group in this case is the identity component of the
automorphism group of Bianchi VI$_0$:
\begin{equation}
  A = \left\{
    \begin{pmatrix}
      a & \zero & \alpha \\
      \zero & d & \beta\\
      \zero & \zero & 1\\
    \end{pmatrix}
    \middle | a,d,\alpha,\beta \in \RR,~ad \neq 0 \right\}.
\end{equation}
There are four $A$-orbits in $\PP\k$ labelled by the vectors
$(0,0,1)$, $(1,0,0)$, $(0,1,0)$ and $(1,1,0)$.

In the first case, $B = \be_3$, $P= \be_1$ and $H = \be_2$, with
nonzero brackets
\begin{equation}
  [B,H] = - (1+\chi) H \qquad\text{and}\qquad [B,P] = (1-\chi) P.
\end{equation}

In the second case, $B = \be_1$, $P= \be_2$ and $H = \be_3$, with
nonzero brackets
\begin{equation}
  [B,H] = (\chi-1) B \qquad\text{and}\qquad [P,H] = (\chi+1) P.
\end{equation}

In the third case, $B = \be_2$, $P= \be_1$ and $H = \be_3$, with
nonzero brackets
\begin{equation}
  [B,H] = (\chi+1) B \qquad\text{and}\qquad [P,H] = (\chi-1) P.
\end{equation}

In the final case, $B = \be_1 + \be_2$, $P = \be_2 - \be_1$ and $H =
\be_3$, with nonzero brackets
\begin{equation}
  [B,H] = \chi B + P \qquad\text{and}\qquad [P,H] = \chi P +  B.
\end{equation}

\subsubsection{Bianchi VII$_0$}
\label{sec:bianchi-vii0}

The automorphism group has two connected components:
\begin{equation}
  A = \left\{
    \begin{pmatrix}
      a & b & \alpha \\
      -b & a & \beta\\
      \zero & \zero & 1\\
    \end{pmatrix}
    \middle | a,b,\alpha,\beta \in \RR,~a^2 + b^2 \neq 0 \right\} \bigcup
   \left\{
    \begin{pmatrix}
      a & b & \alpha \\
      b & -a & \beta\\
      \zero & \zero & -1\\
    \end{pmatrix}
    \middle | a,b,\alpha,\beta \in \RR,~a^2+b^2 \neq 0 \right\}.
\end{equation}
Using only the identity component, we can bring any line in $\PP\k$ to
of of two lines: the one spanned by $(0,0,1)$ and the one spanned by
$(1,0,0)$.

In the first case, $B = \be_3$, $P= \be_1$ and $H = \be_2$, with
nonzero brackets
\begin{equation}
  [B,H] = -P \qquad\text{and}\qquad [B,P] = H.
\end{equation}

In the other case, $B = \be_1$, $P= \be_2$ and $H = \be_3$, with
nonzero brackets
\begin{equation}
  [B,H] = -P \qquad\text{and}\qquad [P,H] = B.
\end{equation}

\subsubsection{Bianchi VII$_{\chi>0}$}
\label{sec:bianchi-viic}

The automorphism group in this case is the identity component of the
automorphism group for Bianchi VII$_0$:
\begin{equation}
  A = \left\{
    \begin{pmatrix}
      a & b & \alpha \\
      -b & a & \beta\\
      \zero & \zero & 1\\
    \end{pmatrix}
    \middle | a,b,\alpha,\beta \in \RR,~a^2 + b^2 \neq 0 \right\}.
\end{equation}
This was all that was needed to bring every line to one of two lines,
spanned by either $(0,0,1)$ and $(1,0,0)$.

In the first case, $B = \be_3$, $P= \be_1$ and $H = \be_2$, with
nonzero brackets
\begin{equation}
  [B,H] = -P - \chi H \qquad\text{and}\qquad [B,P] = H - \chi P.
\end{equation}

In the other case, $B = \be_1$, $P= \be_2$ and $H = \be_3$, with
nonzero brackets
\begin{equation}
  [B,H] = \chi B - P \qquad\text{and}\qquad [P,H] = B + \chi P.
\end{equation}

\subsubsection{Bianchi VIII}
\label{sec:bianchi-viii}

This Lie algebra is isomorphic to $\so(2,1)$ and the automorphism
group is the adjoint group, so isomorphic to $SO(2,1)$.  It acts on
$\k$ via proper Lorentz transformations and hence has three orbits in
the space of lines, corresponding to timelike, spacelike and null
lines.  The inner product on $\k$ is the Killing form $\kappa$, which
in the chosen basis is diagonal with components $\kappa(\be_1,\be_1) =
\kappa(\be_2,\be_2) = 2$ and $\kappa(\be_3,\be_3) = -2$.

We can take for our spacelike line, the one spanned by $(1,0,0)$.
Here $B = \be_1$, $P = \be_2$ and $H = \be_3$, with brackets
\begin{equation}
  [B,H] = - P, \qquad [B,P] = - H \qquad\text{and}\qquad [P, H]
  = B.
\end{equation}

For the timelike line, we take the one spanned by $(0,0,1)$.  Here $B
= \be_3$, $P = \be_1$ and $H = \be_2$, with brackets
\begin{equation}
  [B,H] = - P, \qquad [B,P] = H \qquad\text{and}\qquad [P, H]
  = -B.
\end{equation}

Finally, we take the span of $(1,0,1)$ for the null line.  Here $B =
\be_1 + \be_3$, $P = \be_2$ and $H = \be_1 - \be_3$, with brackets
\begin{equation}
  [B,H] = 2 P, \qquad [B,P] = -B \qquad\text{and}\qquad [P, H]
  = -H.
\end{equation}

\subsubsection{Bianchi IX}
\label{sec:bianchi-ix}

Bianchi IX is isomorphic to $\su(2)$, whose automorphism group is the
adjoint group $SO(3)$, which acts transitively on the space of lines.
We can take $B = \be_1$, $P = \be_2$ and $H = \be_3$ with brackets
\begin{equation}
    [B,H] = - P, \qquad [B,P] = H \qquad\text{and}\qquad [P, H]
  = B.
\end{equation}

\subsubsection{Summary}
\label{sec:summary-deq1}

In Table~\ref{tab:summary1d} we list the Lie pairs associated
to the Bianchi Lie algebras.  In that table we have often redefined
$H$ and $P$ linearly in order to ease the comparison with the
Lie pairs in $D\geq 3$.  This comparison leads to a division of the
table into two: the ones above the horizontal line are 
$D=1$ avatars of Lie pairs which exist also for $D>1$.
Table~\ref{tab:deq1correspondence} shows the correspondence between
the $D=1$ Lie pairs and the $D\geq 3$ Lie pairs in those cases where
there is one.  Notice that in $D=1$ there are exceptional isomorphisms
between Lie pairs which might differ in $D\geq 3$.  For example,
exchanging $P$ and $H$ (i.e., re-interpreting what is time and space),
we see that the following pairs of spacetimes are isomorphic:
de~Sitter/anti~de~Sitter, galilean/carrollian, galilean dS/carrollian
AdS and galilean AdS/carrollian dS.

As in previous tables (Tables~\ref{tab:summarydge3},
\ref{tab:summarydeq3} and \ref{tab:summary2d-real}) each row in
Table~\ref{tab:summary1d} is an isomorphism class of kinematical Lie
pair with sequential label ``LP\#''.  Similarly, the label ``Bianchi''
in Table~\ref{tab:summary1d} identifies the Bianchi type
of the Lie algebra in Table~\ref{tab:bianchi}.  It bears reminding
that the Lie brackets in the tables below are expressed in a basis
where the admissible subalgebra $\h$ is spanned by $B$.

\begin{table}[h!]\small
  \centering
  \caption{Lie pairs for kinematical Lie algebras ($D=1$)}
  \label{tab:summary1d}
  \rowcolors{2}{blue!10}{white}
  \begin{tabular}{l|l|*{3}{>{$}l<{$}}|l} \toprule
    \multicolumn{1}{c|}{LP\#}& \multicolumn{1}{c|}{Bianchi} & \multicolumn{3}{c|}{Nonzero Lie brackets} & \multicolumn{1}{c}{Comments} \\\midrule
    \hypertarget{LP61}{61} & \hyperlink{BianchiI}{I} & & & & static\\
    \hypertarget{LP62}{62} & \hyperlink{BianchiII}{II} & & & [H,P]=B & \\
    \hypertarget{LP63}{63} & \hyperlink{BianchiII}{II} & [H,B] = -P & & & galilean/carrollian\\
    \hypertarget{LP64}{64} & \hyperlink{BianchiIV}{IV} & [H,B] = B & & [H,P] = B + P & \\
    \hypertarget{LP65}{65} & \hyperlink{BianchiIV}{IV} & [H,B] = -P & & [H,P] = B + 2P & \\
    \hypertarget{LP66}{66} & \hyperlink{BianchiV}{V} & [H,B] = B & & [H,P] = P & \\
    \hypertarget{LP67}{67} & \hyperlink{BianchiVI0}{VI$_0$} & [H,B] = -B & & [H,P] = P & \\
    \hypertarget{LP68}{68} & \hyperlink{BianchiVI0}{VI$_0$} & [H,B] = -P & & [H,P] = -B & galilean dS/carrollian AdS\\
    \hypertarget{LP69}{69$_\chi$} & \hyperlink{BianchiVI}{VI$_\chi$} & [H,B] = (\chi-1) B & & [H,P]= (1+\chi)P & $\chi>0$\\
    \hypertarget{LP70}{70$_\chi$} & \hyperlink{BianchiVI}{VI$_\chi$} & [H,B] = (\chi+1) B & & [H,P]= (\chi-1)P & $\chi>0$\\
    \hypertarget{LP71}{71$_\chi$} & \hyperlink{BianchiVI}{VI$_\chi$} & [H,B] = -\chi B + P & & [H,P] = B - \chi P & $\chi>0$\\
    \hypertarget{LP72}{72} & \hyperlink{BianchiVII0}{VII$_0$} & [H,B] = P & [B,P] = H & & euclidean\\
    \hypertarget{LP73}{73} & \hyperlink{BianchiVII0}{VII$_0$} & [H,B] = -P & & [H,P] = B & galilean AdS/carrollian dS\\
    \hypertarget{LP74}{74$_\chi$} & \hyperlink{BianchiVI}{VI$_\chi$} & [H,B] = - P & & [H,P] = (1+\chi^2) B + 2  \chi P & $\chi>0$\\
    \hypertarget{LP75}{75} & \hyperlink{BianchiVIII}{VIII} & [H,B] = - P & [B,P] = H & [H,P] = B & (anti) de Sitter \\
    \hypertarget{LP76}{76} & \hyperlink{BianchiVIII}{VIII} & [H,B] =  P & [B,P] = H & [H,P] = B & hyperbolic \\
    \hypertarget{LP77}{77} & \hyperlink{BianchiIX}{IX} & [H,B] = P & [B,P] =  H & [H,P] = -B & sphere \\
    \hypertarget{LP78}{78} & \hyperlink{BianchiVI0}{VI$_0$} & [H,B] = H & [B,P] = P & & Minkowski\\
    \hypertarget{LP79}{79} & \hyperlink{BianchiVIII}{VIII} & [H,B] = - P & [B,P] = B & [H,P] = -H  & \\\midrule
    \hypertarget{LP80}{80} & \hyperlink{BianchiIV}{IV} & [H,B] = P + H & [B,P] = - P & & \\
    \hypertarget{LP81}{81} & \hyperlink{BianchiV}{V} & [H,B] = H & [B,P] = -P & & \\
    \hypertarget{LP82}{82$_\chi$} & \hyperlink{BianchiVI}{VI$_\chi$} & [H,B] = (1+\chi) H & [B,P]= (1-\chi)P & & $\chi>0$ \\
    \hypertarget{LP83}{83$_\chi$} & \hyperlink{BianchiVII}{VII$_\chi$} & [H,B] = P + \chi H & [B,P] = H - \chi P & & $\chi>0$ \\ \bottomrule
  \end{tabular}
\end{table}

\begin{table}[h!]
  \centering
  \caption{Correspondence between $D=1$ and $D\geq 3$ Lie pairs}
  \label{tab:deq1correspondence}
  \begin{tabular}{l|l} \toprule
    \multicolumn{1}{c|}{$D=1$} & \multicolumn{1}{c}{$D\geq 3$}\\\midrule
    61 & 1\\
    62 & 3\\
    63 & 2 and 14\\
    64 & 11\\
    65 & 10\\
    66 & 7\\ \bottomrule
  \end{tabular}
  \hspace{2cm}
  \begin{tabular}{l|l} \toprule
    \multicolumn{1}{c|}{$D=1$} & \multicolumn{1}{c}{$D\geq 3$}\\\midrule
    67 & 8\\
    68 & 9 and 16$_{+1}$\\
    69$_\chi$ & 4$_{\frac{\chi-1}{\chi+1}}$\\
    70$_\chi$ & 5$_{\frac{\chi-1}{\chi+1}}$\\
    71$_\chi$ & 6$_{\frac{\chi-1}{\chi+1}}$\\
    72 & 15$_{-1}$\\ \bottomrule
  \end{tabular}
  \hspace{2cm}
  \begin{tabular}{l|l} \toprule
    \multicolumn{1}{c|}{$D=1$} & \multicolumn{1}{c}{$D\geq 3$}\\\midrule
    73 & 13 and 16$_{-1}$\\
    74$_\chi$ & 12$_\chi$\\
    75 & 18$_{-1}$ and 19$_{+1}$\\
    76 & 18$_{+1}$\\
    77 & 19$_{-1}$\\
    78 & 15$_{+1}$\\
    79 & 17 \\ \bottomrule
  \end{tabular}
\end{table}

\section{Classification of simply-connected homogeneous spacetimes}
\label{sec:class-kinem-spac}

In Section~\ref{sec:class-kinem-lie} we have classified the
equivalence classes of Lie pairs $(\k,\h)$ where $\k$ is a kinematical
Lie algebra and $\h$ is an admissible subalgebra.  Our aim is to
classify simply-connected homogeneous spacetimes and, as explained in
Appendix~\ref{sec:infin-descr-homog}, this requires selecting those
Lie pairs which are effective (so that $\h$ contains no nonzero ideal
of $\k$) and also geometrically realisable, so that there exists a
kinematical Lie group $\Kgr'$ and an admissible subgroup $\Hgr'$ whose
Lie pair $(\k',\h')$ is isomorphic to $(\k,\h)$.  Let us discuss both
of these selection criteria in turn.

\subsection{Effective Lie pairs}
\label{sec:effective-lie-pairs}

Recall that we have chosen a basis for the kinematical Lie algebra
$\k$ such that the subalgebra $\h$ is spanned by $J_{ab}$ and $B_a$.
This means that the only possible nonzero ideal of $\k$ contained in
$\h$ is the span of the boosts $B_a$.  Therefore to check if a Lie
pair $(\k,\h)$ is effective, all we need to do is inspect the Lie
brackets $[B_a,X]$ for $X$ in the span of $(B_a, P_a, H)$ and see
whether they all lie in the span of the $B_a$ (in which case the Lie
pair is not effective) or not (in which case it is).  Reducing a
non-effective Lie pair by the ideal $\b$ spanned by the boosts, we
arrive at an effective aristotelian Lie pair.  Since it turns out that
not all aristotelian Lie pairs arise in this way, we treat their
classification separately in Appendix~\ref{sec:class-arist-lie}.

\subsubsection{Effective Lie pairs for all $D\geq 3$}
\label{sec:effective-lie-pairs-dge3}

Inspecting Table~\ref{tab:summarydge3} we see that Lie pairs
1, 3, 4$_\gamma$, 5$_\gamma$, 7, 8 and 11 are not effective, since the
span of the $B_a$ define an ideal of $\k$.  Reducing by that ideal we
obtain an aristotelian Lie pair of the ones in
Table~\ref{tab:aristotelian}, which are classified in
Appendix~\ref{sec:class-arist-lie}:
\begin{itemize}
\item Lie pairs 1 and 3 reduce to the static aristotelian Lie pair \st;
  whereas
\item Lie pairs 4$_\gamma$, 5$_\gamma$, 7, 8 and 11 reduce to the
  torsional static aristotelian Lie pair \tst.
\end{itemize}

\subsubsection{No effective Lie pairs unique to $D=3$}
\label{sec:no-effective-lie-pairs-deq3}

Inspecting Table~\ref{tab:summarydeq3} we see that no Lie pairs for
kinematical Lie algebras unique to $D=3$ are effective.  They all
reduce to aristotelian Lie pairs:
\begin{itemize}
\item Lie pairs 21, 23 and 26 reduce to the static aristotelian Lie
  pair \st;
\item Lie pairs 24 and 27 reduce to the torsional static aristotelian
  Lie pair \tst;
\item Lie pair 20$_\varepsilon$ reduces to the aristotelian Lie pair
  \athree$_\varepsilon$ with bracket $[\P,\P] = - \varepsilon \J$; and
\item Lie pairs 22 and 25 reduce to $[\P,\P] = \P$, which is
  isomorphic to \athree$_{+1}$ after changing basis
  $P_a \mapsto \tfrac12 (P_a + \tfrac14 \epsilon_{abc} J_{bc})$.
\end{itemize}

\subsubsection{Effective Lie pairs for $D=2$}
\label{sec:effective-lie-pairs-deq2}

We inspect Table~\ref{tab:summary2d-real} and concentrate on the Lie
pairs which are unique to $D=2$, since for those which exist for $D>2$
the calculations for $D>2$ are valid also for $D=2$.  We see that only
the Lie pair 55$_{\gamma,\chi}$ is effective.  The others reduce to
aristotelian Lie pairs:
\begin{itemize}
\item 56 to \st;
\item 51$_\chi$, 52$_\chi$, 53$_{\gamma,\chi}$ and 54$_{\gamma,\chi}$
  to \tst;
\item 57 and 58 to \twoda; and
\item 59 and 60, respectively, to the aristotelian Lie pairs
  \athree$_{-1}$ and  \athree$_{+1}$, after redefining $H$ and $J$.
\end{itemize}

\subsubsection{Effective Lie pairs for $D=1$}
\label{sec:effective-lie-pairs-deq1}

Inspecting Table~\ref{tab:summary1d}, but concentrating only on the
Lie pairs unique to $D=1$ (those below the horizontal line), we see
that they are all effective.

\subsubsection{Summary}
\label{sec:summary-1}

Table~\ref{tab:effective} summarises the effective Lie pairs.  Some
classes of Lie pairs exist for all $D\geq 1$ and we collect them in
the same row.  This is possibly the most navigationally useful table
in the paper, in that it shows the correspondence between the
spacetimes and their Lie pairs.  The table is hyperlinked for ease of
navigation.

\begin{table}[h!]
  \centering
  \caption{Effective kinematical Lie pairs}
  \label{tab:effective}
  \rowcolors{2}{blue!10}{white}
  \begin{tabular}{l|l|l|l|l} \toprule
    \multicolumn{1}{c|}{Label} & \multicolumn{1}{c|}{$D\geq 3$}& \multicolumn{1}{c|}{$D=2$} & \multicolumn{1}{c|}{$D=1$} & \multicolumn{1}{c}{Comments} \\\midrule
    \hypertarget{S1}{\mink} & \hyperlink{LP15p}{15$_{+1}$} & \hyperlink{LP43}{43} & \hyperlink{LP78}{78} & Minkowski ($\MM$)\\
    \hypertarget{S2}{\ds} & \hyperlink{LP18m}{18$_{-1}$} & \hyperlink{LP47}{47} & \hyperlink{LP75}{75} & de~Sitter ($\zdS$)\\
    \hypertarget{S3}{\ads} & \hyperlink{LP19p}{19$_{+1}$} & \hyperlink{LP49}{49} & \hyperlink{LP75}{75} & anti~de~Sitter ($\zAdS$)\\\midrule
    \hypertarget{S4}{\euc} & \hyperlink{LP15m}{15$_{-1}$} & \hyperlink{LP44}{44} & \hyperlink{LP72}{72} & euclidean ($\EE$)\\
    \hypertarget{S5}{\sph} & \hyperlink{LP19m}{19$_{-1}$} & \hyperlink{LP48}{48} & \hyperlink{LP77}{77} & sphere ($\SS$)\\
    \hypertarget{S6}{\hyp} & \hyperlink{LP18p}{18$_{+1}$} & \hyperlink{LP46}{46} & \hyperlink{LP76}{76} & hyperbolic ($\HH$)\\\midrule
    \hypertarget{S7}{\gal} & \hyperlink{LP2}{2} & \hyperlink{LP29}{29} & \hyperlink{LP63}{63} & galilean ($\zG$)\\
    \hypertarget{S8}{\dsg} & \hyperlink{LP9}{9} & \hyperlink{LP35}{35} & \hyperlink{LP68}{68} & galilean dS ($\zdSG = \ztdSG_{-1}$)\\
    \hypertarget{S9}{\tdsg$_\gamma$} & \hyperlink{LP6}{6$_\gamma$} & \hyperlink{LP38}{38$_\gamma$} & \hyperlink{LP71}{71$_{\frac{1+\gamma}{1-\gamma}}$} & $\ztdSG_{\gamma\in(-1,1)}$\\
    \hypertarget{S9}{\tdsg$_1$} & \hyperlink{LP10}{10} & \hyperlink{LP31}{31} & \hyperlink{LP65}{65} & $\ztdSG_1 = \ztAdSG_\infty$\\
    \hypertarget{S10}{\adsg} & \hyperlink{LP13}{13} & \hyperlink{LP39}{39} & \hyperlink{LP73}{73} & galilean AdS ($\zAdSG = \ztAdSG_0$)\\
    \hypertarget{S11}{\tadsg$_\chi$} & \hyperlink{LP12}{12$_\chi$} & \hyperlink{LP50}{50$_{\frac2\chi}$}  & \hyperlink{LP74}{74$_\chi$} &$\ztAdSG_{\chi>0}$\\
    \hypertarget{S12}{\twodgal$_{\gamma,\chi}$} & $-$ & \hyperlink{LP55}{55$_{\gamma,\chi}$} & $-$ & $\gamma \in[-1,1)$, $\chi>0$\\\midrule
    \hypertarget{S13}{\car} & \hyperlink{LP14}{14} & \hyperlink{LP40}{40} & \hyperlink{LP63}{63} & carrollian ($\zC$)\\
    \hypertarget{S14}{\dsc} & \hyperlink{LP16m}{16$_{-1}$} & \hyperlink{LP42}{42} & \hyperlink{LP73}{73} & carrollian dS ($\zdSC$)\\
    \hypertarget{S15}{\adsc} & \hyperlink{LP16p}{16$_{+1}$} & \hyperlink{LP41}{41} & \hyperlink{LP68}{68} & carrollian AdS ($\zAdSC$)\\
    \hypertarget{S16}{\flc} & \hyperlink{LP17}{17} & \hyperlink{LP45}{45} & \hyperlink{LP79}{79} & carrollian light cone ($\zLC$)\\\midrule
    \hypertarget{S17}{\xone} & $-$ & $-$ & \hyperlink{LP80}{80} & \\
    \hypertarget{S18}{\xtwo} & $-$ & $-$ & \hyperlink{LP81}{81} & \\
    \hypertarget{S19}{\xthree$_\chi$} & $-$ & $-$ & \hyperlink{LP82}{82$_\chi$} & $\chi>0$\\
    \hypertarget{S20}{\xfour$_\chi$} & $-$ & $-$ & \hyperlink{LP83}{83$_\chi$} & $\chi>0$\\ \bottomrule
  \end{tabular}
\end{table}

Table~\ref{tab:LAs-to-spacetimes} is included for convenience.  We
have found this list useful at times and thus we think it might be
useful to other readers.  Table~\ref{tab:LAs-to-spacetimes} lists
which simply-connected homogeneous kinematical or aristotelian
spacetimes are associated to which kinematical Lie algebras.  We do
not list aristotelian Lie algebras since there is a one-to-one
correspondence between the Lie algebras and the simply-connected
homogenous aristotelian spacetimes in that case. The table is
separated into three: corresponding to the Lie algebras for $D\geq 3$,
$D=2$ and $D=1$.  The Lie algebras below the horizontal line in the
first part of the table exist only for $D=3$.  Lie algebras
\hyperlink{LA29}{29} and \hyperlink{LA30}{30}, which exist only in
$D=2$, have no admissible Lie pairs at all.

\begin{table}[h!]
  \centering
  \caption{Kinematical Lie algebras and their homogeneous spacetimes}
  \label{tab:LAs-to-spacetimes}
  \resizebox{\textwidth}{!}{
  \rowcolors{2}{blue!10}{white}
  \begin{tabular}{l|lll} \toprule
    \multicolumn{1}{c|}{Lie algebra} & \multicolumn{3}{c}{Spacetimes} \\\midrule
    \hyperlink{LA1}{1} & \hyperlink{A21}{\st} & & \\
    \hyperlink{LA2}{2} & \hyperlink{S7}{\gal} & \hyperlink{A21}{\st} & \\
    \hyperlink{LA3}{3$_\gamma$} & \hyperlink{S9}{\tdsg$_\gamma$} & \hyperlink{A22}{\tst} & \\
    \hyperlink{LA4}{4} & \hyperlink{A22}{\tst} & &  \\
    \hyperlink{LA5}{5} & \hyperlink{S8}{\dsg} & \hyperlink{A22}{\tst} & \\
    \hyperlink{LA6}{6} & \hyperlink{S9}{\tdsg$_1$} & \hyperlink{A22}{\tst} & \\
    \hyperlink{LA7}{7$_\chi$} & \hyperlink{S11}{\tadsg$_\chi$} & & \\
    \hyperlink{LA8}{8} & \hyperlink{S10}{\adsg} & & \\
    \hyperlink{LA9}{9} & \hyperlink{S13}{\car} & & \\
    \hyperlink{LA10}{10$_{+1}$} & \hyperlink{S1}{\mink} & \hyperlink{S15}{\adsc} & \\
    \hyperlink{LA10}{10$_{-1}$} & \hyperlink{S4}{\euc} & \hyperlink{S14}{\dsc} & \\
    \hyperlink{LA11}{11} & \hyperlink{S2}{\ds} & \hyperlink{S6}{\hyp} & \hyperlink{S16}{\flc}\\
    \hyperlink{LA12}{12$_{+1}$} & \hyperlink{S3}{\ads} & & \\
    \hyperlink{LA12}{12$_{-1}$} & \hyperlink{S5}{\sph} & & \\\midrule
    \hyperlink{LA13}{13$_\varepsilon$} & \hyperlink{A23p}{\athree$_\varepsilon$} & & \\
    \hyperlink{LA14}{14} & \hyperlink{A21}{\st} & \hyperlink{A23p}{\athree$_{+1}$} & \\
    \hyperlink{LA15}{15} & \hyperlink{A21}{\st} & & \\
    \hyperlink{LA16}{16} & \hyperlink{A22}{\tst} & \hyperlink{A23p}{\athree$_{+1}$} & \\
    \hyperlink{LA17}{17} & \hyperlink{A21}{\st} & & \\
    \hyperlink{LA18}{18} & \hyperlink{A22}{\tst} & &  \\\bottomrule
  \end{tabular}
  \qquad
  \rowcolors{2}{blue!10}{white}
  \begin{tabular}{l|lll} \toprule
    \multicolumn{1}{c|}{Lie algebra} & \multicolumn{3}{c}{Spacetimes} \\\midrule
    \hyperlink{LA19}{19} & \hyperlink{A21}{\st} & & \\
    \hyperlink{LA20}{20} & \hyperlink{S7}{\gal} &  \hyperlink{A21}{\st} & \\
    \hyperlink{LA21}{21} & \hyperlink{S9}{\tdsg$_1$} &  \hyperlink{A22}{\tst} & \\
    \hyperlink{LA22}{22$_{1+i0}$} & \hyperlink{A22}{\tst} & & \\
    \hyperlink{LA22}{22$_{\gamma + i 0}$} & \hyperlink{S9}{\tdsg$_\gamma$} &  \hyperlink{A22}{\tst} & \\
    \hyperlink{LA22}{22$_{1+i\chi}$} & \hyperlink{S11}{\tadsg$_{\frac2\chi}$} &  \hyperlink{A22}{\tst} &  \\
    \hyperlink{LA22}{22$_{\gamma+i\chi}$} & \hyperlink{S12}{\twodgal$_{\gamma,\chi}$} &  \hyperlink{A22}{\tst} & \\
    \hyperlink{LA23}{23} & \hyperlink{S8}{\dsg} &  \hyperlink{A22}{\tst} & \\
    \hyperlink{LA24}{24} & \hyperlink{S10}{\adsg}  &  \hyperlink{A21}{\st} & \\
    \hyperlink{LA25}{25} & \hyperlink{S13}{\car} & & \\
    \hyperlink{LA26}{26$_{+1}$} & \hyperlink{S1}{\mink} &  \hyperlink{S15}{\adsc} & \\
    \hyperlink{LA26}{26$_{-1}$} & \hyperlink{S4}{\euc} &  \hyperlink{S14}{\dsc} & \\      
    \hyperlink{LA27}{27} & \hyperlink{S2}{\ds} &  \hyperlink{S6}{\hyp} &  \hyperlink{S16}{\flc}\\
    \hyperlink{LA28}{28$_{+1}$} & \hyperlink{S5}{\sph} & & \\
    \hyperlink{LA28}{28$_{-1}$} & \hyperlink{S3}{\ads} & & \\
    \hyperlink{LA31}{31} & \hyperlink{A24}{\twoda} & & \\
    \hyperlink{LA32}{32} & \hyperlink{A23m}{\athree$_{-1}$} & & \\
    \hyperlink{LA33}{33} & \hyperlink{A23p}{\athree$_\varepsilon$} & & \\\bottomrule
  \end{tabular}
  \qquad
  \rowcolors{2}{blue!10}{white}
  \begin{tabular}{c|llll} \toprule
    \multicolumn{1}{c|}{Lie algebra} & \multicolumn{4}{c}{Spacetimes} \\\midrule
    \hyperlink{BianchiI}{I} & \hyperlink{A21}{\st} & & & \\
    \hyperlink{BianchiII}{II} & \hyperlink{S7}{\gal} & \hyperlink{S13}{\car} &  \hyperlink{A21}{\st} & \\
    \hyperlink{BianchiIV}{IV} & \hyperlink{S9}{\tdsg$_1$} & \hyperlink{S17}{\xone} &  \hyperlink{A22}{\tst} & \\
    \hyperlink{BianchiV}{V} & \hyperlink{S18}{\xtwo} & & & \\
    \hyperlink{BianchiVI0}{VI$_0$} & \hyperlink{S1}{\mink} & \hyperlink{S8}{\dsg} & \hyperlink{S15}{\adsc} &  \hyperlink{A22}{\tst} \\
    \hyperlink{BianchiVI}{VI$_\chi$} & \hyperlink{S9}{\tdsg$_{\frac{\chi-1}{\chi+1}}$} & \hyperlink{S11}{\tadsg$_\chi$} & \hyperlink{S19}{\xthree$_\chi$} &  \hyperlink{A22}{\tst} \\
    \hyperlink{BianchiVII0}{VII$_0$} & \hyperlink{S4}{\euc} & \hyperlink{S10}{\adsg} & \hyperlink{S14}{\dsc} & \\
    \hyperlink{BianchiVII}{VII$_\chi$} & \hyperlink{S20}{\xfour$_\chi$} & & & \\
    \hyperlink{BianchiVIII}{VIII} & \hyperlink{S2}{\ds} & \hyperlink{S3}{\ads} & \hyperlink{S6}{\hyp} & \hyperlink{S16}{\flc} \\
    \hyperlink{BianchiIX}{IX} & \hyperlink{S5}{\sph} & & & \\\bottomrule
  \end{tabular}
  }
\end{table}

\subsection{Geometric realisability}
\label{sec:geom-real}

Having selected the effective Lie pairs, we must then select those
which are geometrically realisable, so that they correspond to the Lie
pair of a homogeneous spacetime.  It will turn out that all effective
Lie pairs are geometrically realisable.  Recall that $(\k,\h)$ is
geometrically realisable if there exists a connected Lie group $\Kgr$,
a Lie algebra isomorphism $\varphi: \text{Lie}(\Kgr) \to \k$ and a
\emph{closed} Lie subgroup $\Hgr \subset \Kgr$ whose Lie algebra is
isomorphic to $\h$ under $\varphi$.  There are a number of criteria
which can be brought to bear in order to help decide whether a Lie
pair admits a geometric realisation.

\subsubsection{Riemannian maximally symmetric spaces}
\label{sec:riem-maxim-symm}

We start by showing that the Lie pairs corresponding to the riemannian
symmetric spaces are geometrically realisable; although this is of
course well-known.

\begin{criterion}\label{crit:compact}
  Compact subgroups are closed, so if if $\h$ generates a
  compact subgroup, then $(\k,\h)$ is geometrically realisable.
\end{criterion}

From this criterion we see that the following Lie pairs are
geometrically realisable:
\begin{itemize}
\item Lie pairs 15$_{-1}$, 44 and 72, which we can identify with
  euclidean spaces;
\item Lie pairs 18$_{+1}$, 46, 76, which we can identify with hyperbolic spaces;
  and 
\item Lie pairs 19$_{-1}$, 48, 77, which we can identify with the round spheres.
\end{itemize}

\subsubsection{A sufficient criterion}
\label{sec:sufficient}

Another useful criterion (sufficient, but by no means necessary)
applies to linear Lie algebras; that is, Lie algebras isomorphic to
Lie algebras of matrices.  By Ado's theorem (see, e.g.,
\cite[Ch.~VI]{MR559927}) every kinematical Lie algebra $\k$, being
finite-dimensional, has a faithful linear representation and hence is
a linear Lie algebra.  Exponentiating inside the matrix algebra, we
obtain a connected Lie group $\Kgr$ with $\k$ (or, more precisely, its
isomorphic image in the matrix algebra) as its Lie algebra.

\begin{criterion}
  If the subalgebra $\h$ of a linear Lie algebra $\k$ is its own
  normaliser, so that the only elements $X \in \k$ with
  $[X,\h] \subset \h$ are the elements of $\h$, then the unique
  connected subgroup $\Hgr \subset \Kgr$ to which it exponentiates is
  closed (see, e.g., \cite[Pr.~2.7.4]{MR1889121}) and $\Kgr/\Hgr$ is
  geometric realisation of $(\k,\h)$.
\end{criterion}

In particular, if the Lie pair is reductive, so that
$\k = \h \oplus \m$ as adjoint $\h$-modules, then we may decompose
$X\in\k$ uniquely as $X = X_\h + X_\m$, with $X_\h \in \h$ and
$X_\m \in \m$.  Now, rotations (when present) act reducibly, so
$[J_{ab},X]\in \h$ means that $X_\m = c H$ for some $c$. So we then
need to inspect whether $[B_a,H] \in \h$.

As we can see by inspection, this second criterion allows us to
conclude that all Lie pairs are geometrically realisable with the
following possible exceptions requiring a closer look: Lie pairs 14,
16$_\varepsilon$ and 17 in $D\geq 3$, Lie pairs 40, 41, 42 and 45 in
$D=2$ and Lie pairs 63, 65, 68, 71$_\chi$, 73, 74$_\chi$, 79 and
82$_{\chi=1}$ in $D=1$.

\subsubsection{Two-dimensional spacetimes}
\label{sec:two-dimens-spac}

To show that $D=1$ Lie pairs 63, 65, 68, 71$_\chi$, 73, 74$_\chi$, 79
and 82$_{\chi=1}$ are geometrically realisable, we may use yet a third
criterion for when a one-dimensional subgroup of a matrix group
is (not) closed.

\begin{criterion}
  A one-parameter subgroup of a matrix group is \emph{not} closed if
  and only if the generating matrix in the Lie algebra is similar to
  a diagonal matrix with imaginary entries, at least two of which have
  an irrational ratio (see, e.g., \cite[Pr.~2.7.5]{MR1889121}).
\end{criterion}

All the Bianchi Lie algebras have faithful representations of
dimension 2 or 3.  So it is simply a matter of calculating the
eigenvalues of $B$ in each of these representations to deduce that all
these Lie pairs are geometrically realisable.  Indeed, let us go back
to the description of the Bianchi Lie algebras in terms of the basis
$(\be_1,\be_2,\be_3)$ as in Table~\ref{tab:bianchi} and let us write
the generic element as $X = x \be_1 + y \be_2 + z \be_3$.  For each of
the cases of interest, we will write down the matrix $\rho(X)$
representing $X$ and the matrix $\rho(B)$ representing $B$.  We will
see that in no case does $\rho(B)$ have imaginary eigenvalues.

\begin{enumerate}
\item[(63)] This is Bianchi II,
  \begin{equation}
   \rho(X) =
   \begin{pmatrix}
     \zero & -z & x \\ \zero & \zero & y \\ \zero & \zero & \zero
   \end{pmatrix}
   \qquad\text{so that}\qquad
   \rho(B) = \rho(\be_2) = \begin{pmatrix}
     \zero & \zero & \zero \\ \zero & \zero & 1 \\ \zero & \zero & \zero
   \end{pmatrix}.
 \end{equation}
 This matrix is not diagonalisable.
 
\item[(65)] This is Bianchi IV,
  \begin{equation}
   \rho(X) =
   \begin{pmatrix}
     -z & -z & x \\ \zero & -z & y \\ \zero & \zero & \zero
   \end{pmatrix}
   \qquad\text{so that}\qquad
   \rho(B) = \rho(\be_2) = \begin{pmatrix}
     \zero & \zero & \zero \\ \zero & \zero & 1 \\ \zero & \zero & \zero
   \end{pmatrix},
 \end{equation}
 which is not diagonalisable.

\item[(68)] This is Bianchi VI$_0$, but let us consider the general
  Bianchi VI$_\chi$,
  \begin{equation}
   \rho(X) =
   \begin{pmatrix}
     (1-\chi) z & \zero & x \\ \zero & -(1+\chi)z & y \\ \zero & \zero & \zero
   \end{pmatrix}
   \qquad\text{so that}\qquad
   \rho(B) = \rho(\be_1 + \be_2) = \begin{pmatrix}
     \zero & \zero & 1 \\ \zero & \zero & 1 \\ \zero & \zero & \zero
   \end{pmatrix},
 \end{equation}
 which is not diagonalisable.
  
\item[(71$_\chi$)] This is again Bianchi VI$_\chi$, so that we can reuse
  the previous calculation.  In this case also
  \begin{equation}
    \rho(B) = \rho(\be_1 + \be_2) = \begin{pmatrix}
     \zero & \zero & 1 \\ \zero & \zero & 1 \\ \zero & \zero & \zero
   \end{pmatrix},
 \end{equation}
 which is not diagonalisable.
 
\item[(74$_\chi$)] This is again Bianchi VI$_\chi$ and again $B = \be_1+
  \be_2$, with the same matrix non-diagonalisable as above.

\item[(73)] This is Bianchi VII$_0$, but we will treat the general
  Bianchi VII$_\chi$ whose matrix representation is
  \begin{equation}
   \rho(X) =
   \begin{pmatrix}
    -\chi z & -z & x \\ z & -\chi z & y \\ \zero & \zero & \zero
   \end{pmatrix}
   \qquad\text{so that}\qquad
   \rho(B) = \rho(\be_1) = \begin{pmatrix}
     \zero & \zero & 1 \\ \zero & \zero & \zero \\ \zero & \zero & \zero
   \end{pmatrix},
 \end{equation}
 which is not diagonalisable.

 \item[(82$_1$)] This is Bianchi VII$_1$ and $B = \be_3$, so that the
   matrix representing it is
   \begin{equation}
   \rho(B) = \rho(\be_1) = \begin{pmatrix}
     -1 & -1 & \zero \\ 1 & -1 & \zero \\ \zero & \zero & \zero
   \end{pmatrix}.
 \end{equation}
 This matrix is diagonalisable, but the eigenvalues are not
 imaginary: $-1 \pm i$.

 \item[(79)] This is Bianchi VIII ($\cong \mathfrak{sl}(2,\RR)$) which
   has a faithful two-dimensional representation
   \begin{equation}
     \rho(X) =
     \begin{pmatrix}
       x & y \\ z & -x
     \end{pmatrix}.
   \end{equation}
   Here $B = \be_1 + \be_3$, which has matrix
   \begin{equation}
     \rho(B) =
     \begin{pmatrix}
       1 &  \zero \\ 1 & -1
     \end{pmatrix},
   \end{equation}
   whose eigenvalues are real: $\pm 1$.
 \end{enumerate}

In summary, after this discussion we are left with the carrollian Lie
pairs: 14, 16$_\varepsilon$ and 17 in $D\geq 3$ and 40, 41, 42 and 45
in $D=2$.

\subsubsection{The carrollian light cone}
\label{sec:non-reduct-carr}

It is easy to show that Lie pairs 17 and 45 are geometrically
realisable.  We can treat them together and, in fact, the argument
works also to give an alternative proof of geometric realisability for
the $D=1$ Lie pair 79.  These Lie pairs are all of the form
$(\so(D+1,1),\h)$ where, as we now show, $\h$ is the subalgebra of
$\so(D+1,1)$ corresponding to the stabiliser of a null vector in the
lorentzian vector space $\RR^{D+1,1}$.  Indeed let
$e_\mu = \{e_a, e_+, e_-\}$ be a Witt basis for $\RR^{D+1,1}$ where
the lorentzian inner product $\eta_{\mu\nu}:= \eta(e_\mu,e_\nu)$ is
given by $\eta_{ab} = \delta_{ab}$ and $\eta_{+-} = 1$.  Then
$\so(D+1,1)$ has generators $J_{\mu\nu}$ with Lie brackets:
\begin{equation}
  [J_{\mu\nu}, J_{\rho\sigma}] = \eta_{\nu\rho} J_{\mu\sigma} -
  \eta_{\mu\rho} J_{\nu\sigma} - \eta_{\nu\sigma} J_{\mu\rho} +
  \eta_{\mu\sigma} J_{\nu\rho}.
\end{equation}
Letting $\mu = (a,+,-)$ and decomposing $J_{\mu\nu}$ into $\{J_{ab},
J_{a+}, J_{a-}, J_{+-}\}$, we obtain the following nonzero brackets:
\begin{equation}
  \begin{split}
    [J_{ab}, J_{cd}] &= \delta_{bc} J_{ad} - \delta_{ac} J_{bd} - \delta_{bd} J_{ac} + \delta_{ad} J_{bc}\\
    [J_{ab}, J_{c+}] &= \delta_{bc} J_{a+} - \delta_{ac} J_{b+}\\
    [J_{ab}, J_{c-}] &= \delta_{bc} J_{a-} - \delta_{ac} J_{b-}\\
    [J_{+-}, J_{a+}] &= J_{a+}\\
    [J_{+-}, J_{a-}] &= -J_{a-}\\
    [J_{a+}, J_{b-}] &= - J_{ab} - \delta_{ab} J_{+-},
  \end{split}
\end{equation}
from where we can identify $B_a = J_{a+}$, $P_a = -J_{a-}$ and $H = J_{+-}$.
The Lie pairs 17, 45 and 79 correspond to $(\so(D+1,1), \h)$ where
$\h$ is the span of $J_{ab}$ (if $D\geq 2$) and $B_a$.  The key
observation is that $\h$ is the subalgebra which annihilates the basis
vector $e_+$ under the usual action:
\begin{equation}
  J_{\mu\nu} \cdot e_\rho = \eta_{\nu\rho} e_\mu - \eta_{\mu\rho} e_\nu.
\end{equation}
The connected subgroup $\Hgr$ of $\SO(D+1,1)$ generated by $\h$ is
(the identity component of) the stabiliser of $e_+$ and hence it is a
closed subgroup. Since $\Hgr$ is connected, it is actually a subgroup
of $\SO(D+1,1)_0$, the identity component of $\SO(D+1,1)$. We conclude
that $\SO(D+1,1)_0/\Hgr$ is thus a geometric realisation of
$(\so(D+1,1),\h)$. Geometrically, it corresponds to the future light
cone $\eL_+ \subset \MM$, where $\MM$ is ($D+2$)-dimensional Minkowski
spacetime.\footnote{Strictly speaking this description holds for
  $D\geq 2$.  For $D=1$, the future light cone is not simply-connected
  and hence the simply-connected geometric realisation of the Lie pair
  79 is the universal cover of the future light cone, which we may
  think of as the submanifold of $\RR^3$ consisting of points
  $(r\cos\theta,r\sin\theta, \theta)$ with $r>0$, projecting to the
  light cone in $\RR^{1,2}$, by sending
  $(r\cos\theta,r\sin\theta,\theta) \mapsto
  (r,r\cos\theta,r\sin\theta)$.}  Indeed, as shown in
\cite{Duval:2014uoa} (see also \cite{Hartong:2015xda}), null
hypersurfaces (such as the future light cone in Minkowski spacetime)
are carrollian spacetimes.  This idea turns out to be very fruitful in
order to prove the geometric realisability of the remaining symmetric
carrollian Lie pairs, as we will now see.

\subsubsection{Symmetric carrollian spacetimes}
\label{sec:symm-carr-spac}

Finally, we show that the symmetric carrollian Lie pairs 14 and
16$_\varepsilon$ in $D\geq 3$ and 40, 41 and 42 in $D=2$ are
geometrically realisable.  As mentioned above already, one way to do
this is to construct the geometric realisations explicitly by
exhibiting them as null hypersurfaces in lorentzian manifolds one
dimension higher.  This was done originally for the carrollian
spacetime $\zC$ (\hyperlink{S13l}{\car}) in \cite{Duval:2014uoa}, who
embedded it as a null hypersurface in Minkowski spacetime one
dimension higher.  In a similar way we will construct
the carrollian (anti) de~Sitter spacetimes $\zdSC$
(\hyperlink{S14l}{\dsc}) and $\zAdSC$ (\hyperlink{S15l}{\adsc}) as
null hypersurfaces in (anti) de~Sitter spacetimes one dimension
higher.  All we need to show is that the Lie pairs describing these
null hypersurfaces are the symmetric carrollian Lie pairs 14,
16$_\varepsilon$, 40, 41 and 42.

Introducing a parameter $\varepsilon = 0, \pm 1$, we define the
kinematical Lie algebra $\k_{\varepsilon}$ by the following Lie
brackets in addition to the ones in \eqref{eq:kin}:
\begin{equation}
  [H, \P] = \varepsilon \B,\qquad [\B,\P] = H \qquad\text{and}\qquad
  [\P,\P] = \varepsilon \J.
\end{equation}
We shall let $\h_\varepsilon$ denote the admissible subalgebra spanned
by $\J$ and $\B$.  The Lie pair $(\k_\varepsilon,\h_\varepsilon)$ is
isomorphic to 14 and 40 when $\varepsilon=0$, to 16$_{+1}$ and 41 when
$\varepsilon=+1$ and to 16$_{-1}$ and 42 when $\varepsilon=-1$.  We
will now exhibit homogeneous manifolds whose Lie pairs are isomorphic
to $(\k_\varepsilon,\h_\varepsilon)$ for each value of $\varepsilon$.

We start with $\varepsilon=0$, which is the construction of flat
carrollian space $\zC$ (\hyperlink{S13l}{\car}) in
\cite{Duval:2014uoa}. To this end, let $\MM$ denote
($D+2$)-dimensional Minkowski spacetime with coordinates $x^\mu$ for
$\mu = 0,1,\dots,D+1$ (although we will let $\natural$ stand for
$D+1$) and metric \begin{equation}
  ds^2 = \eta_{\mu\nu} dx^\mu dx^\mu = - (dx^0)^2 + \sum_{a=1}^D
  (dx^a)^2+ (dx^\natural)^2 = 2 dx^+ dx^- + \sum_{a=1}^D (dx^a)^2,
\end{equation}
where $x^\pm = \tfrac{1}{\sqrt2}(x^\natural \pm x^0)$.  Let $\eC_0 \subset
\MM$ denote the null hypersurface defined by $x^- = 0$.  We claim that
$\eC_0$ is a geometric realisation of the Lie pair $(\k_0,\h_0)$.

The Poincaré Lie algebra is spanned by the following vector fields on
$\MM$:
\begin{equation}
  J_{\mu\nu} = x_\mu \frac{\d}{\d x^\nu} - x_\nu \frac{\d}{\d x^\mu}
  \qquad\text{and}\qquad P_\mu = \frac{\d}{\d x^\mu},
\end{equation}
where $x_\mu = \eta_{\mu\rho} x^\rho$.  It is a transitive Lie algebra
for $\MM$, which means that for every $p \in \MM$, the values
$J_{\mu\nu}(p)$ and $P_\mu(p)$ span the tangent space $T_p \MM$.
The subalgebra of the Poincaré algebra consisting of vector fields
which are tangent to $\eC$ is spanned by
\begin{equation}
  J_{ab}, \qquad P_a, \qquad B_a := J_{a0} + J_{a\natural} \qquad
  H:= P_0 + P_\natural \qquad\text{and}\qquad J_{0\natural}.
\end{equation}
The subalgebra spanned by $J_{ab}$,$P_a$, $B_a$ and $H$ satisfy the
Lie brackets of $\k_0$ (the Carroll algebra) and the subalgebra $\h_0$
which vanishes at the point $o \in \eC_0$ with coordinates $x^\mu = 0$
is the span of $J_{ab}$ and $B_a$.  By dimension, $P_a$ and $H$ span
the tangent space $T_o\eC_0$ and indeed the same is true at any other
point of $\eC$ with different (but isomorphic) stabiliser subalgebra.
Hence $\k_0$ is a transitive Lie algebra for $\eC_0$ with stabiliser
$\h_0$ at $o$.  Therefore $\eC_0$ is a geometric realisation of
$(\k_0,\h_0)$.

Now let us consider $\varepsilon=-1$ and let $\MM$ now stand for
Minkowski spacetime of dimension $D+3$, with coordinates $x^\mu$ with
$\mu= 0,1,\dots,D+2$ and metric
\begin{equation}
  ds^2 = \eta_{\mu\nu} dx^\mu dx^\nu = - (dx^0)^2 + \sum_{A=1}^{D+2} (dx^A)^2.
\end{equation}
Let $\eQ_-$ denote the quadric defined by
\begin{equation}
  \eta_{\mu\nu} x^\mu x^\nu = R^2,
\end{equation}
for some (fixed) $R>0$.  Its universal cover (with the induced metric)
is de~Sitter spacetime in dimension $D+2$.  The Lorentz Lie algebra,
spanned by
\begin{equation}
  J_{\mu\nu} = x_\mu \frac{\d}{\d x^\nu} - x_\nu \frac{\d}{\d x^\mu}
\end{equation}
with $x_\mu = \eta_{\mu\rho} x^\rho$, is a transitive Lie algebra for
$\eQ_-$ and isomorphic to $\so(D+2,1)$.  Let $\eN \subset \MM$ denote
the null hypersurface defined by the equation $x^0 = x^{D+2}$ and let
$\eC_- = \eQ_- \cap \eN $.  This is defined by the following two
equations:
\begin{equation}
  x^0 = x^{D+2} \qquad\text{and}\qquad \sum_{i=1}^{D+1} (x^i)^2 = R^2,
\end{equation}
which shows that $\eC_-$ is diffeomorphic to $\RR \times \SS^D$ and
hence is simply-connected for $D\geq 2$.  We claim that $\eC_-$ is a
geometric realisation for the Lie pair $(\k_-,\h_-)$.  We will show
this by determining the subalgebra of the Lorentz Lie algebra
consisting of vector fields tangent to $\eC_-$, a subalgebra of which,
by the same argument as in the previous case, is transitive on
$\eC_-$.  We will show that this subalgebra is isomorphic to $\k_-$
and that the stabiliser at a suitably chosen point $o \in \eC_-$ is
isomorphic to $\h_-$.  The following Lorentz generators are tangent to
$\eC_-$:
\begin{equation}
  J_{ij}, \qquad V_i := J_{i0} + J_{i,D+2} \qquad\text{and}\qquad J_{0,D+2}
\end{equation}
for all $i=1,\dots,D+1$.  The subalgebra spanned by $J_{ij}$ and $V_i$
is isomorphic to the euclidean algebra in dimension $D+1$, with
nonzero Lie brackets
\begin{equation}
  \begin{split}
    [J_{ij},J_{k\ell}] &= \delta_{jk} J_{i\ell} - \delta_{ik} J_{j\ell} - \delta_{j\ell} J_{ik} + \delta_{i\ell} J_{jk}\\
    [J_{ij}, V_k] &= \delta_{jk} V_i - \delta_{ik} V_j.
  \end{split}
\end{equation}
Let $a = 1,\dots,D$ and let $\natural = D+1$.  Then these generators
break up as
\begin{equation}
  J_{ab},\qquad P_a:=J_{a\natural}, \qquad B_a:=V_a \qquad\text{and}\qquad H:=V_\natural,
\end{equation}
which obey
\begin{equation}
  [H,P_a] = - B_a, \qquad [B_a,P_b] = \delta_{ab} H
  \qquad\text{and}\qquad [P_a,P_b] = - J_{ab},
\end{equation}
apart from \eqref{eq:kin}.  We see that this Lie algebra is isomorphic
to $\k_-$.  Now let $o \in \eC_-$ denote the point with coordinates
$x^a=x^0=x^{D+2}=0$ and $x^\natural = R$.  Then the vector fields
which vanish at $o$ are the span of $J_{ab}$ and $B_a$, which is
isomorphic to $\h_-$.  Therefore $\eC_-$ is a geometric realisation of
$(\k_-,\h_-)$.

Finally, we consider the case $\varepsilon=+1$.  Now $\widetilde{\EE}$ is
pseudo-euclidean space with signature $(D+1,2)$ with coordinates
$x^\mu$ for $\mu =0,1,\dots,D+2$ and metric
\begin{equation}
  ds^2 = \eta_{\mu\nu} dx^\mu dx^\nu = - (dx^0)^2 + \sum_{i=1}^{D+1} (dx^i)^2 -
  (dx^{D+2})^2
\end{equation}
metric as above.  Now we fix $R>0$ and let $\eQ_+\subset
\widetilde{\EE}$ denote the quadric defined by the equation
\begin{equation}
  \eta_{\mu\nu} x^\mu x^\nu = -R^2.
\end{equation}
Its universal cover (with the induced metric) is anti de~Sitter
spacetime in dimension $D+2$.  The Lie algebra spanned by the vector
fields
\begin{equation}
  J_{\mu\nu} = x_\mu \frac{\d}{\d x^\nu} - x_\nu \frac{\d}{\d x^\mu}
\end{equation}
with $x_\mu = \eta_{\mu\rho} x^\rho$, is a transitive Lie algebra for
$\eQ_+$ and isomorphic to $\so(D+1,2)$.  Let $\eN \subset
\widetilde{\EE}$ denote now the null hypersurface with equation
$x^{D+1} = x^{D+2}$.  The intersection $\eQ_+ \cap \eN$ is described
by the two equations:
\begin{equation}
  x^{D+1} = x^{D+2} \qquad\text{and}\qquad (x^0)^2 = R^2 +
  \sum_{a=1}^D (x^a)^2,
\end{equation}
which has two connected components determined by the sign of $x^0$,
which is never zero.  Let $\eC_+$ denote the component where $x^0>0$.
The following vector fields generate the subalgebra of the span of the
$J_{\mu\nu}$ which are tangent to $\eC_+$:
\begin{equation}
  J_{ab}, \qquad P_a := - J_{0a}, \qquad B_a := J_{a,D+1} + J_{a,D+2},
  \qquad H:= J_{0,D+1} + J_{0,D+2} \qquad\text{and}\qquad J_{D+1,D+2},
\end{equation}
for $a=1,\dots,D$.  The subalgebra spanned by $J_{ab}$, $P_a$, $B_a$
and $H$ has, in addition to \eqref{eq:kin}, the
following nonzero Lie brackets
\begin{equation}
[H,P_a] = B_a   \qquad [B_a, P_b] = \delta_{ab} H  \qquad\text{and}\qquad [P_a,P_b] = J_{ab},
\end{equation}
showing that it is isomorphic to $\k_+$.  The stabiliser subalgebra at the point
$o\in\eC_+$ with coordinates $x^a = x^{D+1} = x^{D+2}=0$ and $x^0 = R$, is the
span of $J_{ab}$ and $B_a$, so isomorphic to $\h_+$.  As before,
dimension says that $\k_+$ is transitive at $o\in \eC_+$ with stabiliser $\h_+$
and also at other points with isomorphic stabilisers.  Therefore
$\eC_+$ is a geometric realisation of $(\k_+,\h_+)$.

It is worth remarking that in all the homogeneous carrollian
spacetimes discussed here, the boost generators can be interpreted as
null rotations in the ambient Minkowski or pseudo-euclidean spaces.

\section{Limits between homogeneous spacetimes}
\label{sec:limits}

In the previous two sections we have classified the simply-connected
homogeneous spacetimes.  This provides the objects in
Figure~\ref{fig:generic-d-graph}, which also contains arrows between
the spacetimes.  These arrows are explained by limits between
spacetimes and in this section we will discuss these limits and in
this way explain Figure~\ref{fig:generic-d-graph}.  We will also
explain how the picture gets modified in $D\leq 2$ and explain
Figures~\ref{fig:d=3-graph} and \ref{fig:d=2-graph}.

In the infinitesimal description of the homogeneous spacetimes in
terms of Lie pairs, most (but not all) limits between homogeneous
spacetimes manifest themselves as contractions of the underlying
kinematical Lie algebras.

\subsection{Contractions}
\label{sec:contractions}

Recall that a (finite-dimensional, real) Lie algebra consists of a
vector space $V$, together with a linear map
$\varphi: \Lambda^2 V \to V$ satisfying the Jacobi identity.  The
Jacobi identity defines an algebraic variety
$\eJ \subset \Lambda^2V^*\otimes V$, every point of which is a Lie
algebra structure on $V$.  The general linear group $\GL(V)$ acts on
$V$ and hence also on the vector space $\Lambda^2V^*\otimes V$ and
since the action is tensorial it preserves the variety $\eJ$.  If
$\varphi \in \eJ$ defines a Lie algebra $\g$ and $g \in \GL(V)$, then
$g \cdot \varphi \in \eJ$ and, by definition, the Lie algebra it
defines is isomorphic to $\g$.  Indeed, the $\GL(V)$ orbit of
$\varphi$ consists of all Lie algebras on $V$ which are isomorphic to
$\g$.  This orbit may not be closed relative to the induced topology
on $\eJ$.  The closure of the orbit may contain Lie algebras which are
not isomorphic to $\g$: they are said to be ``degenerations'' of $\g$.
A special class of degenerations are the \emph{contractions} of $\g$,
which are limit points of curves in the $\GL(V)$-orbit of $\varphi$.
More precisely, let $t \in (0,1]$ and let $g_t$ be a continuous curve
in $\GL(V)$ with $g_1$ the identity.  Define $\varphi_t := g_t \cdot
\varphi$. Explicitly, the Lie bracket $[-,-]_t$ associated to
$\varphi_t$ is given by
\begin{equation}
  [X,Y]_t = g_t \cdot [g_t^{-1} \cdot X, g_t^{-1} \cdot Y].
\end{equation}
Then $\varphi_1 = \varphi$ and for every $t \in (0,1]$, $\varphi_t$
defines a Lie algebra on $V$ isomorphic to $\g$.  By continuity, if
the limit $\varphi_0 := \lim_{t \to 0} \varphi_t$ exists, it defines a
Lie algebra, but since the linear transformation $g_0 := \lim_{t \to
  0} g_t$ of $V$ (even if it exists) need not be invertible, the Lie
algebra defined by $\varphi_0$ need not be isomorphic to $\g$.  It is,
however, by definition a \textbf{contraction} of $\g$.

We will now explicitly exhibit contractions between kinematical Lie
algebras which induce the limits between the $(D+1)$-dimensional
homogeneous spaces in Figure~\ref{fig:generic-d-graph} for $D\geq 3$.
We will also explain other non-contraction limits in that figure, as
well as in $D\leq 2$.

\subsection{$D\geq 3$}
\label{sec:dgeq-3}

We will start with the lorentzian and riemannian space forms with
nonzero curvature, whose kinematical Lie algebras are the (semi)simple
Lie algebras: $\so(D+1,1)$ for de~Sitter spacetime (\ds) and hyperbolic
space (\hyp), $\so(D,2)$ for anti~de~Sitter spacetime (\ads) and
$\so(D+2)$ for the round sphere (\sph).

Let $\RR^{D+2}$ have basis $e_\mu = (e_0,e_a,e_\natural)$, with $a
=1,\dots,D$ and let $\eta$ denote the inner product with coefficients
$\eta_{\mu\nu} = \eta(e_\mu,e_\nu)$ given by $\eta_{ab} = \delta_{ab}$
and all other components zero except for $\eta_{00}$ and
$\eta_{\natural\natural}$.   The generators of $\so(\RR^{D+2},\eta)$
are $J_{\mu\nu} = \{J_{ab}, B_a := J_{0a}, P_a := J_{a\natural}, H := J_{0\natural}\}$
subject to the following Lie brackets
\begin{equation}
  \begin{aligned}\relax
    [J_{ab}, J_{cd}] &= \delta_{bc} J_{ad} - \delta_{ac} J_{bd} - \delta_{bd} J_{ac} + \delta_{ad} J_{bc}\\
    [J_{ab}, B_c] &= \delta_{bc} B_a - \delta_{ac} B_b\\
    [J_{ab}, P_c] &= \delta_{bc} P_a - \delta_{ac} P_b\\
    [H, B_a] &= \eta_{00} P_a
  \end{aligned}
  \qquad\qquad
  \begin{aligned}\relax
    [H, P_a] &= -\eta_{\natural\natural} B_a\\
    [B_a, B_b] &= -\eta_{00} J_{ab}\\
    [P_a, P_b] &= -\eta_{\natural\natural} J_{ab}\\
    [B_a, P_b] &= \delta_{ab} H.
  \end{aligned}
\end{equation}
The first three brackets are the standard kinematical Lie brackets, so
we will focus attention on the remaining brackets and we will change
to shorthand notation where the $\so(D)$ indices are implicit.

Let us consider a three-parameter $(\kappa,c,\tau)$ family of linear
transformations $g_{\kappa,c,\tau}$ defined on generators by
\begin{equation}
  g_{\kappa,c,\tau} \cdot \J = \J, \qquad
  g_{\kappa,c,\tau} \cdot \B = \tfrac{\tau}{c} \B, \qquad
  g_{\kappa,c,\tau} \cdot \P = \tfrac{\kappa}{c} \P
  \qquad\text{and}\qquad
  g_{\kappa,c,\tau} \cdot H = \tau\kappa H.
\end{equation}
The transformed Lie brackets are such that the common kinematical Lie
brackets involving $\J$ are unchanged and the remaining brackets are
\begin{equation}
  \label{eq:simple-params}
  \begin{split}
    [H, \B] &= \tau^2 \eta_{00} \P\\
    [H, \P] &= -\kappa^2 \eta_{\natural\natural} \B\\
    [\B, \B] &= -\left(\tfrac{\tau}{c}\right)^2 \eta_{00} \J\\
    [\P, \P] &= -\left(\tfrac{\kappa}{c}\right)^2 \eta_{\natural\natural} \J\\
    [\B, \P] &= \left(\tfrac{1}{c}\right)^2 H.
  \end{split}
\end{equation}
The \textbf{flat limit} corresponds to taking $\kappa \to 0$, the
\textbf{non-relativistic limit} to $c \to \infty$ and the
\textbf{ultra-relativistic limit} to $\tau \to 0$.  We may take any
two of the three limits or, indeed, all limits at once; although
whenever we combine a non-relativistic limit with an ultra-relativistic
limit we arrive at a non-effective Lie pair reducing to the
aristotelian static spacetime (\st), denoted $\hyperlink{A21}{\mathsf{S}}$ in
Figure~\ref{fig:generic-d-graph}.

Let us first take the flat limit $\kappa \to 0$ of the Lie brackets
in \eqref{eq:simple-params} to arrive at
\begin{equation}
  \begin{split}
    [H, \B] &= \tau^2 \eta_{00} \P\\
    [\B, \B] &= -\left(\tfrac{\tau}{c}\right)^2 \eta_{00} \J\\
    [\B, \P] &= \left(\tfrac{1}{c}\right)^2 H.
  \end{split}
\end{equation}
For $\frac{\tau}{c} \neq 0$, this is either the Poincaré Lie algebra
for $\eta_{00} = -1$ or the euclidean Lie algebra for $\eta_{00}=1$.
The corresponding Lie pairs are those of Minkowski spacetime (\mink) and
euclidean space (\euc).  In Figure~\ref{fig:generic-d-graph} only
$\eta_{00}=-1$ is considered and these limits explain the arrows
$\hyperlink{S3l}{\mathsf{AdS}}\to\hyperlink{S1l}{\mathbb{M}}$ and $\hyperlink{S2l}{\mathsf{dS}}\to\hyperlink{S1l}{\mathbb{M}}$, but in
Figure~\ref{fig:state-of-prior-art} this also explains the arrows
$\hyperlink{S6l}{\mathbb{H}} \to \hyperlink{S4l}{\mathbb{E}}$ and $\hyperlink{S5l}{\mathbb{S}} \to \hyperlink{S4l}{\mathbb{E}}$.  We may now take
the non-relativistic limit $c \to \infty$ to arrive at the galilean
algebra (after rescaling $H$ by $-1/(\eta_{00} \tau^2)$),
\begin{equation}
    [H, \B] = - \P,
\end{equation}
or alternatively the ultra-relativistic limit $\tau \to 0$ to arrive at
the Carroll algebra (after setting $c=1$):
\begin{equation}
    [\B, \P] = H.
\end{equation}
The corresponding Lie pairs are the galilean (\gal) and carrollian
(\car) spacetimes.  This explains the arrows $\hyperlink{S1l}{\mathbb{M}} \to \hyperlink{S7l}{\mathsf{G}}$
and $\hyperlink{S1l}{\mathbb{M}} \to \hyperlink{S13l}{\mathsf{C}}$ in Figure~\ref{fig:generic-d-graph} and
the arrows $\hyperlink{S4l}{\mathbb{E}} \to \hyperlink{S7l}{\mathsf{G}}$ and $\hyperlink{S4l}{\mathbb{E}} \to \hyperlink{S13l}{\mathsf{C}}$ in
Figure~\ref{fig:state-of-prior-art}.

Taking now the non-relativistic limit of the Lie brackets in
\eqref{eq:simple-params}, we have
\begin{equation}
  \begin{split}
    [H, \B] &= \tau^2 \eta_{00} \P\\
    [H, \P] &= -\kappa^2 \eta_{\natural\natural} \B.
  \end{split}
\end{equation}
For $\tau \kappa \neq 0$, we obtain the Lie pairs corresponding to the
galilean de~Sitter spacetime (\dsg), if
$\eta_{\natural\natural} \eta_{00} = -1$, or the galilean
anti~de~Sitter spacetime (\adsg), if
$\eta_{\natural\natural} \eta_{00} = 1$, thus explaining the arrows
$\hyperlink{S2l}{\mathsf{dS}} \to \hyperlink{S8l}{\mathsf{dSG}}$ and $\hyperlink{S3l}{\mathsf{AdS}} \to \hyperlink{S10l}{\mathsf{AdSG}}$ in
Figure~\ref{fig:generic-d-graph} and also the arrows
$\hyperlink{S6l}{\mathbb{H}} \to \hyperlink{S8l}{\mathsf{dSG}}$ and $\hyperlink{S5l}{\mathbb{S}} \to \hyperlink{S10l}{\mathsf{AdSG}}$ in
Figure~\ref{fig:state-of-prior-art}.  If we then take the flat
limit $\kappa \to 0$, we obtain the galilean spacetime, thus
explaining the arrows $\hyperlink{S8l}{\mathsf{dSG}} \to \hyperlink{S7l}{\mathsf{G}}$ and
$\hyperlink{S10l}{\mathsf{AdSG}} \to \hyperlink{S7l}{\mathsf{G}}$.  If instead we take the
ultra-relativistic limit we obtain a non-effective Lie pair reducing
to the aristotelian static spacetime (\st).

Finally, let us start by taking the ultra-relativistic limit
($\tau\to0$) in the brackets~\eqref{eq:simple-params}, to arrive at
\begin{equation}
  \begin{split}
    [H, \P] &= -\kappa^2 \eta_{\natural\natural} \B\\
    [\P, \P] &= -\left(\tfrac{\kappa}{c}\right)^2 \eta_{\natural\natural} \J\\
    [\B, \P] &= \left(\tfrac{1}{c}\right)^2 H.
  \end{split}
\end{equation}
If $\frac{\kappa}c \neq 0$, we obtain either the carrollian de~Sitter
spacetime (\dsc) if $\eta_{\natural\natural} = 1$ or the carrollian
anti~de~Sitter spacetime (\adsc) if $\eta_{\natural\natural} = -1$.
This explains the arrows $\hyperlink{S2l}{\mathsf{dS}} \to \hyperlink{S14l}{\mathsf{dSC}}$ and
$\hyperlink{S3l}{\mathsf{AdS}} \to \hyperlink{S15l}{\mathsf{AdSC}}$ in Figure~\ref{fig:generic-d-graph} and
the arrows $\hyperlink{S6l}{\mathbb{H}} \to \hyperlink{S15l}{\mathsf{AdSC}}$ and $\hyperlink{S5l}{\mathbb{S}} \to \hyperlink{S14l}{\mathsf{dSC}}$
in Figure~\ref{fig:state-of-prior-art}.  If we now take the flat
limit we arrive at the carrollian spacetime, which explains the arrows
$\hyperlink{S14l}{\mathsf{dSC}} \to \hyperlink{S13l}{\mathsf{C}}$ and $\hyperlink{S15l}{\mathsf{AdSC}} \to \hyperlink{S13l}{\mathsf{C}}$.  If
instead we take the non-relativistic limit, we arrive at a
non-effective Lie pair reducing to the aristotelian static spacetime
(\st).

Of all the arrows to the aristotelian static spacetime, only
$\hyperlink{S7l}{\mathsf{G}} \to \hyperlink{A21}{\mathsf{S}}$ and $\hyperlink{S13l}{\mathsf{C}} \to \hyperlink{A21}{\mathsf{S}}$ are shown
explicitly in Figure~\ref{fig:generic-d-graph}.  Taking any two limits
in the brackets \eqref{eq:simple-params}, the resulting Lie algebra
does not depend on the order in which we take the limits.  This means
that the arrows in \eqref{eq:simple-params} ``commute'' and thus, for
instance, that the arrow (not shown) $\hyperlink{S14l}{\mathsf{dSC}} \to \hyperlink{A21}{\mathsf{S}}$ is to
be understood as the composition of the arrows (shown)
$\hyperlink{S14l}{\mathsf{dSC}} \to \hyperlink{S13l}{\mathsf{C}} \to \hyperlink{A21}{\mathsf{S}}$.  Similarly, the arrows (not
shown) $\hyperlink{S15l}{\mathsf{AdSC}} \to \hyperlink{A21}{\mathsf{S}}$, $\hyperlink{S8l}{\mathsf{dSG}} \to \hyperlink{A21}{\mathsf{S}}$ and
$\hyperlink{S10l}{\mathsf{AdSG}} \to \hyperlink{A21}{\mathsf{S}}$ can be understood as compositions of
arrows which are shown: $\hyperlink{S15l}{\mathsf{AdSC}} \to \hyperlink{S13l}{\mathsf{C}} \to \hyperlink{A21}{\mathsf{S}}$,
$\hyperlink{S8l}{\mathsf{dSG}} \to \hyperlink{S7l}{\mathsf{G}} \to \hyperlink{A21}{\mathsf{S}}$ and $\hyperlink{S10l}{\mathsf{AdSG}} \to \hyperlink{S7l}{\mathsf{G}}
\to \hyperlink{A21}{\mathsf{S}}$, respectively.

We have so far explained the limits in
Figure~\ref{fig:generic-d-graph} (or even
Figure~\ref{fig:state-of-prior-art}) corresponding to the known
symmetric spacetimes and it now remains to explain the limits from the
new spacetimes in our classification.

\subsubsection{$\ztAdSG_\chi \to \zG$}
\label{sec:adsg-to-g}

Let $t \in (0,1]$ and let $g_t$ be defined by
\begin{equation}
  g_t \cdot \J = \J, \qquad g_t \cdot \B = \B, \qquad g_t \cdot \P =
  t \P \qquad\text{and}\qquad g_t \cdot H = t H.
\end{equation}
The new brackets are now
\begin{equation}
  [H,\B] = - \P \qquad\text{and}\qquad [H,\P] = t^2 (1+ \chi^2) \B + 2
  t \chi \P~,
\end{equation}
so that taking the limit $t \to 0$, gives the galilean algebra $[H,\B]
= - \P$.

\subsubsection{$\ztdSG_\gamma \to \zG$}
\label{sec:dsg-to-g}

This is just like the previous case.  Under the same $g_t$ as before,
the new brackets are now
\begin{equation}
  [H,\B] = - \P \qquad\text{and}\qquad [H,\P] = t^2 \gamma \B + t (1+
  \gamma) \P~,
\end{equation}
so that taking the limit $t \to 0$, gives the galilean algebra.

\subsubsection{$\zLC \to \zC$}
\label{sec:flc-to-c}

Taking $g_t$ as in the previous two cases, the brackets become
\begin{equation}
  [H,\B] = t \B, \qquad [H,\P] = t \P \qquad\text{and}\qquad [\B,\P] = H + t \J.
\end{equation}
Taking the limit $t \to 0$ we recover the Carroll algebra $[\B,\P] = H$.

\subsubsection{$\zLC \to \zTS$}
\label{sec:flc-to-ts}

Let $t \in (0,1]$ and let $g_t$ be defined by
\begin{equation}
  g_t \cdot \J = \J, \qquad g_t \cdot \B = \B, \qquad g_t \P = t \P
  \qquad\text{and}\qquad g_t \cdot H =  H,
\end{equation}
so that the brackets become
\begin{equation}
  [H,\B] = \B, \qquad [H,\P] = \P \qquad\text{and}\qquad [\B,\P] = t H + t \J.
\end{equation}
Taking the limit $t\to 0$ gives
\begin{equation}
  [H,\B] = \B \qquad\text{and}\qquad [H,\P] = \P.
\end{equation}
The resulting Lie pair is not effective because the span of the $B_a$
is an ideal.  Quotienting by this ideal gives the aristotelian Lie
algebra defined by $[H,\P] = \P$, whose associated spacetime is the
torsional static spacetime \tst.

\subsubsection{$\zTS \to \zS$}
\label{sec:ts-to-s}

Let $t \in (0,1]$ and let $g_t \cdot H = t H$ and $g_t \cdot P = P$.
The new bracket is $[H,P] = t P$, which vanishes in the limit $t \to
0$.

\subsubsection{A non-contracting limit}
\label{sec:add-lim}

Finally, we discuss a limit which does not come from a contraction of
Lie algebras.  The Lie algebra of $\hyperlink{S11l}{\mathsf{AdSG}}_\chi$ depends on a
parameter $\chi\geq 0$ and this parameter determines the isomorphism
class of the Lie algebra.  A natural question is what spacetime
corresponds to $\hyperlink{S11l}{\mathsf{AdSG}}_\chi$ in the limit $\chi \to \infty$.  The
answer turns out to be that
$\lim_{\chi\to\infty} \hyperlink{S11l}{\mathsf{AdSG}}_\chi = \hyperlink{S9l}{\mathsf{dSG}}_1$.
To see this we start with the Lie algebra corresponding to
$\hyperlink{S11l}{\mathsf{AdSG}}_\chi$
\begin{equation}
  [H,\B] = -\P \qquad\text{and}\qquad [H,\P] = (1+\chi^2) \B + 2\chi \P
\end{equation}
and we change basis to
\begin{equation}
  H' = \chi^{-1} H \qquad \B' = \B \qquad\text{and}\qquad \P' =
  \chi^{-1} \P.
\end{equation}
This is a vector space isomorphism for any $\chi >0$, but becomes
singular in the limit $\chi \to \infty$.  In this sense this is
reminiscent of a contraction, but it is not a contraction since we are
changing the isomorphism type of the Lie algebra as we change $\chi$.
In the new basis,
\begin{equation}
  [H',\B'] = -\P' \qquad\text{and}\qquad [H',\P'] = (1+\chi^{-2}) \B' + 2 \P'
\end{equation}
and now taking $\chi \to \infty$ we arrive at
\begin{equation}
  [H',\B'] = -\P' \qquad\text{and}\qquad [H',\P'] = \B' + 2 \P',
\end{equation}
which is the Lie algebra corresponding to $\hyperlink{S9l}{\mathsf{dSG}}_1$.

\subsection{$D=2$}
\label{sec:d=2}

In $D=2$ there is an additional two-parameter family of spacetimes not
present in $D\geq 3$: namely, spacetime \twodgal$_{\gamma,\chi}$, for
$\gamma \in [-1,1)$ and $\chi>0$.  The Lie brackets are given by
\eqref{eq:kin} and in addition
\begin{equation}\label{eq:twodgal}
  [H,\B] = -\P  \qquad\text{and}\qquad  [H,\P] = (1+\gamma) \P -
  \chi\Pt + \gamma \B - \chi \Bt,
\end{equation}
or in complex form
\begin{equation}\label{eq:twodgal-complex}
  [H,\B] = -\P  \qquad\text{and}\qquad  [H,\P] = (1+z) \P + z \B,
\end{equation}
where $z = \gamma + i \chi \in \CC$ lies in the infinite vertical
strip in the upper-half plane defined by $-1\leq\Re
z<1$.

Parenthetically, let us mention that for $z = -1 + i \chi$,
the complexified Lie pair for spacetime \twodgal$_{-1,\chi}$ is isomorphic
to the complexification of the Lie pair for $\hyperlink{S8l}{\mathsf{dSG}}$.  This can be seen
by the following complex change of basis:
\begin{equation}
  H' = \tfrac{2}{2+i\chi} (H + \tfrac\chi2 J) \qquad \B' = \B
  \qquad\text{and}\qquad  \P' = \tfrac{2}{2+i\chi}(\P + i \tfrac\chi2 \B),
\end{equation}
so that
\begin{equation}
  [H', \B'] = - \P' \qquad\text{and}\qquad [H',\P'] = -\B'.
\end{equation}
This provides an example of a finite-dimensional complex Lie algebra
having a continuum of non-isomorphic real forms.

The region where $z$ lives has two additional boundaries: $z = \gamma
\in [-1,1]$ and $z= 1 + i \chi$ with $\chi\geq 0$.  The horizontal
boundary $z = \gamma \in [-1,1]$ corresponds to $\chi = 0$ in
equation~\eqref{eq:twodgal}:
\begin{equation}
  [H,\B] = -\P  \qquad\text{and}\qquad  [H,\P] = \gamma \B + (1+\gamma) \P,
\end{equation}
which corresponds to $\hyperlink{S9l}{\mathsf{dSG}}_\gamma$.

Let us change basis (for $\chi > 0$) from $(H,\B,\P)$ to
$(H':=\frac2\chi H + J, \B':=\B, \P':=i \B + \frac2\chi \P)$ in such a
way that the Lie brackets become
\begin{equation}
  [H',\B'] = -\P' \qquad\text{and}\qquad [H',\P'] =
  \frac{2(1+\gamma)}\chi \P' + \left(1 + \frac{4\gamma^2}{\chi^2} +
    \frac{2(1-\gamma)i}\chi \right) \B'.
\end{equation}
When $\gamma = 1$, this corresponds to $\hyperlink{S11l}{\mathsf{AdSG}}_{2/\chi}$.  Now let
us consider the (singular) limit $\chi \to\infty$, so that the Lie
brackets become
\begin{equation}
  [H',\B'] = -\P' \qquad\text{and}\qquad [H',\P'] = \B',
\end{equation}
which we recognise as $\hyperlink{S10l}{\mathsf{AdSG}}$.

The picture resulting from this discussion is illustrated in
Figure~\ref{fig:d=2-region}.  This shows that in $D=2$, spacetime
\twodgal$_{\gamma,\chi}$ interpolates between the one-dimensional
continua of torsional galilean de~Sitter and anti~de~Sitter
spacetimes.  Figure~\ref{fig:d=3-graph} shows how to insert this
figure into Figure~\ref{fig:generic-d-graph}.

\begin{figure}[h!]
  \centering
  \begin{tikzpicture}[>=latex, x=1.0cm,y=1.0cm]
    %
    % grid
    %
%    \draw [color=gray,step=1] (-4,-1) grid (4,8);
    % 
    % vertices
    % 
    \coordinate [label=below left:{\tiny $\hyperlink{S8l}{\mathsf{dSG}}$}] (dsg) at (-2,0);
    \coordinate [label=below right:{\tiny $\hyperlink{S9l}{\mathsf{dSG}}_1$}] (dsgone) at (2,0);
    \coordinate [label=above:{\tiny $\hyperlink{S10l}{\mathsf{AdSG}}$}] (adsgtoo) at  (-2,4);
    \coordinate [label=above:{\tiny $\hyperlink{S10l}{\mathsf{AdSG}}$}] (adsg) at (2,4);
    \coordinate [label=above:{\tiny $\hyperlink{S10l}{\mathsf{AdSG}}$}] (adsgthree) at  (0,4);    
    %
    % region
    %
    \fill [color=green!10!white] (adsgtoo) -- (dsg) -- (dsgone) -- (adsg) -- (adsgtoo);
    % 
    % labels
    %
    \coordinate [label=below:{\tiny  $\hyperlink{S9l}{\mathsf{dSG}_\gamma}$\quad $\gamma\in[-1,1]$}] (tdsg) at (0,0);
    \coordinate [label=right: {\tiny $\hyperlink{S11l}{\mathsf{AdSG}_{\frac2\chi}}$\quad $\chi>0$}] (tadsg) at (2,2);
    \coordinate [label={\small (\twodgal)$_{\gamma, \chi}$}] (tdg) at (0,1.5);
    % 
    % edges
    % 
    \draw [-,line width=2pt,color=green!70!black] (dsg) -- (dsgone);
    \draw [->,line width=2pt,color=green!15!white] (dsg) -- (adsgtoo);
    \draw [<->, line width=2pt,shorten >= 3pt,color=green!70!black] (adsg) -- (dsgone);
    % 
    % points
    % 
    \foreach \point in {dsg,dsgone}
    \filldraw [color=green!70!black,fill=green!70!black] (\point) circle (1.5pt);
  \end{tikzpicture}
  \caption{Parameter space for some $D=2$ spacetimes}
  \label{fig:d=2-region}
\end{figure}
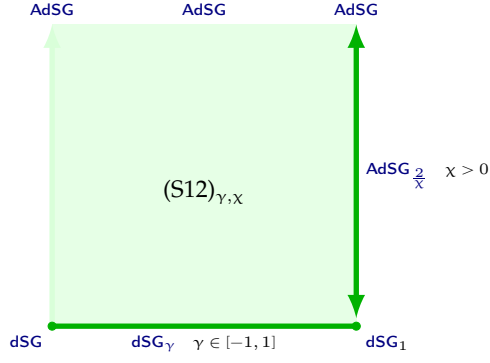

\subsection{$D=1$}
\label{sec:d=1}

Now the spacetimes are two-dimensional.  It is then mostly a matter of
convention what we call space and what we call time.  This manifests
itself in some accidental isomorphisms between spacetimes.  For
example, de~Sitter and anti~de~Sitter spacetimes are isomorphic as
homogeneous spaces of $\SO(2,1) \cong \SO(1,2)$.  At first one might
be surprised at this statement since after all de~Sitter space has
positive scalar curvature, whereas anti~de~Sitter space has negative
scalar curvature and surely they are geometrically distinguishable.
This is perhaps a good place to point out that scalar curvature is not
an invariant of a homogeneous space, but rather of the homogeneous
space together with the choice of an invariant metric.  In (A)dS there
is a one-parameter family of invariant metrics, labelled by the radius
of curvature, all sharing the same connection and curvature: after
all, the Levi--Civita connection and hence the Riemann curvature are
homothety invariant.  Even the Ricci tensor is homothety invariant and
it is only the Ricci scalar, which involves tracing with the metric,
that distinguishes between the (A)dS spacetimes with different radii
of curvature.  In two dimensions, we have the possibility of
exchanging space and time, which results in multiplying the metric by
$-1$, which is formally a homothety, so that the Riemann and Ricci
tensors are unchanged.  It is only when we calculate the Ricci scalar
that we see the effect of this homothety: namely, changing the sign.

There are other accidental isomorphisms coming from exchanging space
and time, e.g., galilean and carrollian spacetimes are isomorphic.
This is clear from Figure~\ref{fig:lightcones} since exchanging space
and time rotates the light cones by 90 degrees.  Similarly, carrollian
dS and galilean AdS spacetimes are isomorphic, as are carrollian AdS
and galilean dS spacetimes: the result of both changing the sign of
the curvature and rotating the light cone by 90 degrees.

In addition to these identifications, there are additional homogeneous
spacetimes which are unique to two dimensions: namely, \xone, \xtwo,
\xthree$_{\chi>0}$ and \xfour$_{\chi>0}$ and which admit none of the
low-rank invariant structures we have been focussing on.  We describe
two kinds of limits between these exotic spacetimes: limits which
manifest themselves infinitesimally as contractions of the relevant
Bianchi Lie algebra as well as limits which are not of this type.  The
contractions between the Bianchi Lie algebras have been determined in
\cite{MR1123589}.  All Lie algebras contract to the abelian Lie
algebra (here, Bianchi~I), and we will not mention these contractions
explicitly.  The resulting spacetimes and their limits are depicted in
Figure~\ref{fig:d=2-graph}.

\subsubsection{$\text{\xone}\to\text{\xtwo}$ and $\text{\xone}\to\text{\gal/\car}$}
\label{sec:s17-to-s18-g/c}

The Lie algebra associated to \xone\ is Bianchi~IV, which can contract
to Bianchi~II and Bianchi~V \cite{MR1123589}.  Indeed, both
contractions arise from limits of spacetime \xone: a limit to
galilean/carrollian spacetime and a limit to \xtwo.

The Lie algebra has brackets $[H,B] = -P$ and $[B,P]=-H-2P$.  This is
an extension by $B$ of the abelian Lie algebra spanned by $H$ and
$P$.  The action of $B$ is via a non-diagonalisable endomorphism with
one eigenvector with eigenvalue $-1$.  This means that we can change
basis in the span of $H$ and $P$ so that relative to the new basis $H',P'$,
\begin{equation}
  [B,H'] = - H' \qquad\text{and}\qquad [B,P'] = -P' + H'.
\end{equation}
If we now introduce a parameter $t \in (0,1]$ and a one-parameter
family $g_t$ of invertible endomorphisms defined by
\begin{equation}
  g_t \cdot H' = H', \qquad g_t \cdot B = B \qquad\text{and}\qquad
  g_t \cdot P' = t P',
\end{equation}
the brackets become
\begin{equation}
  [B,H']_t = - H' \qquad\text{and}\qquad [B,P'] = -P' + t H',
\end{equation}
so that the limit $t \to 0$ recovers the Lie algebra corresponding to
spacetime \xtwo.

If instead we define $g_t$ by
\begin{equation}
  g_t \cdot H' = tH', \qquad g_t \cdot B = t B \qquad\text{and}\qquad
  g_t \cdot P' = P',
\end{equation}
the brackets become
\begin{equation}
  [B,H']_t = - t H' \qquad\text{and}\qquad [B,P'] = -t P' + H',
\end{equation}
so that the limit $t \to 0$ recovers the Lie algebra corresponding to
spacetime \gal/\car\ after exchanging $H'$ and $P'$.

\subsubsection{$\text{\xthree}_\chi \to \text{\gal/\car}$}
\label{sec:s19-to-g/c}

The Lie algebra is Bianchi~VI$_\chi$ for $\chi>0$, which can contract
to Bianchi~II.  The Lie brackets are $[H,B]=(1-\chi) H$ and $[B,P] =
(1+\chi)P$.  Change basis to $H' = H+P$ and $P' = P-H$, so that the
Lie algebra in this new basis is
\begin{equation}
  [H',B] = -P' - \chi H' \qquad\text{and}\qquad [B,P'] = H' + \chi P'.
\end{equation}
Define $g_t$, for $t \in (0,1]$, by
\begin{equation}
  g_t \cdot B = t B, \qquad g_t \cdot H' = H' \qquad\text{and}\qquad
  g_t \cdot P' = t P'.
\end{equation}
The brackets become
\begin{equation}
  [H',B]_t = - P' - \chi t H' \qquad\text{and}\qquad [B,P'] = t^2 H' +
  \chi t P',
\end{equation}
so that the limit $t \to 0$ recovers the Lie algebra corresponding to
spacetime \gal/\car.

\subsubsection{$\text{\xfour}_\chi \to \text{\gal/\car}$}
\label{sec:s20-to-g/c}

The Lie algebra is Bianchi~VII$_\chi$ for $\chi>0$, which can contract
to Bianchi~II.  Under the same one-parameter family of
invertible transformations $g_t$ as in the previous case:
\begin{equation}
  g_t \cdot B = t B, \qquad g_t \cdot H = H \qquad\text{and}\qquad
  g_t \cdot P = t P,
\end{equation}
the brackets become
\begin{equation}
  [H,B]_t = P + \chi t H \qquad\text{and}\qquad [B,P] = t^2 H -
  \chi t P.
\end{equation}
Taking the limit $t\to 0$ and changing $B$ to $-B$, recovers the Lie
algebra of spacetime \gal/\car.

\subsubsection{Additional limits}
\label{sec:additional-limits}

The continua of spacetimes \xthree$_\chi$ and \xfour$_\chi$ depend on
a parameter $\chi >0$.  If we take $\chi \to 0$, then we obtain either
Minkowski spacetime or euclidean space:
$\text{\mink} = \text{\xthree}_{\chi=0}$ and $\text{\euc} =
\text{\xfour}_{\chi=0}$.  The limits $\chi \to \infty$ are different.
Defining $B' = \chi^{-1} B$ and letting $\chi \to \infty$, we obtain
the Lie algebra
\begin{equation}
  [B,H]=-H \qquad\text{and}\qquad [B,P] = -P,
\end{equation}
which corresponds to spacetime \xtwo.  The only limit of spacetime
\xtwo\ is to the aristotelian static spacetime \st.

\section{Some geometrical properties of homogeneous spacetimes }
\label{sec:some-geom-prop}

In this section we start to study some of the geometrical properties
of the homogeneous spacetimes in Tables~\ref{tab:spacetimes} and
\ref{tab:aristotelian}.  We will concentrate on those geometrical
properties which are easy to glean from the infinitesimal description
in terms of the Lie pair $(\k,\h)$ and which help distinguish between
the different homogeneous spacetimes.  In a follow-up paper
\cite{Figueroa-OFarrill:2019sex} we will study in more detail the local geometry
of these homogeneous spacetimes.

\subsection{Basic notions}
\label{sec:basic-notions}

We start by introducing some notions about Lie pairs which are the
algebraic analogues of geometric properties of their associated
homogeneous spaces.

\subsubsection{Reductive and symmetric Lie pairs}
\label{sec:reduct-symm-lie}

We say that a Lie pair $(\k,\h)$ is \textbf{reductive} if there is a
vector space decomposition $\k = \h \oplus \m$, where  $[\h,\h] \subset \h$ and $[\h,\m] \subset
\m$. A reductive Lie pair is said to be \textbf{symmetric} if
$[\m,\m] \subset \h$.  At the other extreme, if $[\m,\m] \subset \m$,
then $\m$ is a Lie algebra and then the homogeneous space is a
principal homogeneous space (i.e., with trivial stabilisers) of a Lie
group $\Mgr$ with Lie algebra $\m$, which is simply-connected if the
homogeneous space is.  The intersection between two cases is where
$[\m,\m] = 0$.  In that case $\Mgr$ is an abelian group and, if
simply-connected, then a vector space.  The homogeneous space is then
an affine space modelled on $\Mgr$.  We will say then that the
homogeneous space is \textbf{affine}.  (One should not confuse this
with the more general notion of an \emph{affine symmetric space} from
\cite{MR0059050,MR1393941}.)

In Table~\ref{tab:spacetimes} we have chosen a basis for
$\k$ in such a way that, in the reductive examples, $\h$ is spanned by
$\J$ and $\B$ and $\m$ is the span
of $\P$ and $\H$.  It is clear by inspection that all of the
spacetimes in Table~\ref{tab:spacetimes} are reductive with one
exception: the carrollian light cone $\zLC$.  We will see
below that $\zLC$ does not admit any invariant connections for $D\geq
2$, providing a separate proof that it is indeed non-reductive.  We
will also classify the invariant connections for $D=1$.

\subsubsection{The linear isotropy representation}
\label{sec:line-isotr-repr}

The natural geometric objects in a homogeneous space (e.g., metric,
connections, curvature, torsion,...) are those which are invariant
under the group action.  Invariance means, in particular, that their
value at any point is invariant under the stabiliser subgroup of that
point, which acts on the tangent space at that point via the linear
isotropy representation, which we now introduce at the Lie algebraic
level.

Even if $(\k,\h)$ is not reductive, we have a representation of $\h$
on $\k/\h$, called the \textbf{linear isotropy representation} and
denoted $\lambda: \h \to \gl(\k/\h)$ and sending $X \in \h$ to
$\lambda_X$. For every $Y \in \k$, let $\Ybar$ denote its image in
$\k/\h$. Then $X \in \h$ acts on $\Ybar$ as
\begin{equation}
  \lambda_X \Ybar := \overline{[X,Y]},
\end{equation}
which is well defined because $\h$ is a Lie subalgebra.  In the
reductive case, $\m$ is isomorphic to $\k/\h$ as a representation of
$\h$.  In the non-reductive case, we may choose a vector space
complement $\m$ and then define on it a representation of $\h$ by
transporting the linear isotropy representation on $\k/\h$ via the
vector space isomorphism $\m \cong \k/\h$.  In practice all this means
is that when calculating the action of $X \in \h$ on $Y \in \m$, we
can compute the Lie bracket $[X,Y]$ in $\k$ and set to zero anything
on the right-hand side that lies in $\h$.

If the Lie pair $(\k,\h)$ has a geometric realisation, the vector
space $\k/\h$ is an algebraic model for the tangent space at the
origin of any homogeneous space $M:=\Kgr/\Hgr$ with $\Kgr$ a Lie group
with Lie algebra $\k$ and $\Hgr$ a Lie subgroup with Lie algebra $\h$.
If $\Hgr$ is connected---which we can assume with no loss of
generality by passing to the universal cover of the homogeneous space,
if necessary---there is then a bijective correspondence between
$\h$-invariant tensors on $\k/\h$ and $\Kgr$-invariant tensor fields
on $M$.

\subsubsection{Invariant structures}
\label{sec:invariants}

For the purposes of this paper, we are particularly interested in
invariant tensors of low rank; that is, $\h$-invariant tensors in
$\k/\h$, $(\k/\h)^*$, $S^2(\k/\h)$ and $S^2(\k/\h)^*$.  A
non-degenerate invariant tensor in $S^2(\k/\h)^*$ gives rise to a
$\Kgr$-invariant metric on $M$, whereas $\h$-invariant tensors in
$\k/\h$ and $(\k/\h)^*$ give rise to an invariant vector field and
one-form on $M$, respectively.  We shall say that a homogeneous
kinematical spacetime $M$ is \textbf{lorentzian} or
\textbf{riemannian} if it admits a $\Kgr$-invariant metric of
lorentzian or riemannian signature.  This is the case if the
associated Lie pair admits an $\h$-invariant nondegenerate tensor in
$S^2(\k/\h)^*$ of the right signature.  In cases where rotations are
present in $\h$ (i.e., $D\geq 2$), the only possible invariant in
$\k/\h$ must be proportional to $\overline{H}$ and any invariants in
$S^2(\k/\h)$ must lie in the subspace spanned by $\overline{H}^2$ and
$\overline{P}^2 := \delta^{ab} \overline{P}_a \overline{P}_b$.

Let us introduce a basis $(\pi^a,\eta)$ for $(\k/\h)^*$ canonically
dual to the basis $(\overline{P}_a,\overline{H})$ for $\k/\h$; that
is,
\begin{equation}
  \pi^a (\overline{P}_b) = \delta^a_b \qquad \pi^a(\overline{H})=0
  \qquad
  \eta (\overline{P}_a) = 0 \qquad \eta(\overline{H})=1.
\end{equation}
Then similarly, for $D\geq 2$, any invariants in $(\k/\h)^*$ are
proportional to $\eta$, whereas in $S^2(\k/\h)^*$ any invariant lies
in the span of $\eta^2$ and $\pi^2$, where $\pi^2 = \delta_{ab}
\pi^a\pi^b$.

If $(\k,\h)$ is such that the one-form $\eta \in (\k/\h)^*$ and the
co-metric $\overline{P}^2 \in S^2(\k/\h)$ are $\h$-invariant, we say
that $M$ admits an invariant \textbf{galilean} structure. The one-form
$\eta$ is called the absolute clock reflecting that galilean
spacetimes are absolute in time. This means, two points of the
homogeneous space that are at same time stay that way, irrespective of
any galilean kinematical transformation. If $\overline{H} \in \k/\h$
and $\pi^2 \in S^2(\k/\h)^*$ are $\h$-invariant we say that $M$ admits
an invariant \textbf{carrollian} structure. Analogous to the galilean
case, the fundamental vector field $\overline{H}$ reflects the
absolute space character of carrollian spacetimes.  Notice that
aristotelian spacetimes in Table~\ref{tab:aristotelian} admit
simultaneously a galilean and carrollian structure, since
$\overline{H},\eta,\overline{P}^2,\pi^2$ are all rotationally
invariant.  It follows that they also admit many invariant lorentzian
and riemannian structures.

It is easy to determine the existence of these invariants from the
data in Table~\ref{tab:spacetimes}. After writing down down the
possible rotationally invariant tensors, we only need to check
invariance under $\B$. The action of $\B$ is induced by duality from
its action via the linear isotropy representation on $\g/\h$:
\begin{equation}
  \lambda_{B_a} \overline{H} = \overline{[B_a,H]} \qquad\text{and}\qquad
  \lambda_{B_a} \overline{P}_b = \overline{[B_a, P_b]}
\end{equation}
with the brackets being those of $\k$.  In practice, we can determine
this from the tables, by computing the brackets in $\k$ and simply
dropping any $\B$ or $\J$ from the right-hand side.

The only possible invariants in $\k/\h$ are proportional to
$\overline{H}$, which is invariant provided that $[\B,H] \in \h$.
Dually, the only possible invariants in $(\k/\h)^*$ are proportional
to $\eta$, which is invariant provided that there is no $X \in \k$
such that $H$ appears in $[\B,X]$.

\subsubsection{Parity and time reversal}
\label{sec:parity-time-reversal}

We define a \textbf{parity transformation} on a kinematical Lie pair
$(\k,\h)$ to be an automorphism of $\k$ which changes the sign of $\B$
and $\P$ and leaves $H$ and $\J$ inert.  Similarly, we define a
\textbf{time reversal transformation} to be an automorphism of $\k$
which changes the sign of $\B$ and $H$, but leaves $\P$ and $\J$
inert.  For aristotelian Lie pairs $(\a,\r)$, only $\P$ changes sign under
a parity transformation and only $H$ changes sign under a time
reversal transformation.  The combination of a parity and
time-reversal transformations is then an automorphism of $\k$ (or
$\a$) which changes simultaneously the signs of $\P$ and $H$ and leaves
other generators inert.

It follows from equivariance under rotations, that for $D\neq 1,3$,
every kinematical Lie algebra possesses a parity transformation.  It
is only in $D=1$ and $D=3$ where we can have Lie brackets which
violate parity: in $D=3$ because we have a vector product which is
only invariant under the orientation preserving orthogonal
transformations, and in $D=1$ because there are no rotations.  We saw
that there are no effective Lie pairs for $D=3$ whose kinematical Lie
algebra involves the vector product, so that it is only in $D=1$ where
we can expect to have effective Lie pairs without parity (or time
reversal, since in $D=1$ what is time and space is a matter of
convention) symmetry.

Any spacetime whose canonical connection (see
Section~\ref{sec:invar-conn} below) has torsion cannot be invariant
under $PT$, hence it cannot be invariant under both $P$ and $T$.
Since, as explained above, parity is guaranteed for $D\neq 1$, it is
time reversal invariance which fails for torsional geometries.

It is simply a matter of inspecting the brackets in
Tables~\ref{tab:spacetimes} and \ref{tab:aristotelian} to determine
whether the corresponding spacetimes possess parity and/or time
reversal invariance.  The results are summarised in
Table~\ref{tab:spacetimes-props}.

\subsection{Invariant connections, curvature and torsion}
\label{sec:invar-conn}

As a final geometric property, we discuss the existence of invariant
affine connections and their curvature and torsion. Let $(\k,\h)$ be
an effective Lie pair associated to a homogeneous space. By an
\textbf{invariant (affine) connection} on $(\k,\h)$ we shall mean a
linear map $\Lambda : \k \to \gl(\k/\h)$, denoted
$X \mapsto \Lambda_X$, satisfying the following two properties:
\begin{enumerate}
\item $\left.\Lambda\right|_\h = \lambda$, the linear isotropy
  representation, and
\item $\Lambda$ is $\h$-equivariant.
\end{enumerate}
Notice that $\Lambda$ is not generally a Lie algebra homomorphism;
although its restriction to $\h$ is.  In fact, as we will see, the
obstruction of $\Lambda$ being a Lie algebra homomorphism is the
curvature of the connection.

The equivariance condition for $\Lambda$ says that for all
$X,Y \in \k$ and $Z \in \h$,
\begin{equation}
  0 =
  \lambda_Z \Lambda_X \Ybar -
  \Lambda_X\lambda_Z \Ybar - \Lambda_{[Z,X]}\Ybar ,
\end{equation}
which is the $\h$-invariance of $\Lambda$ as an element in the space
of linear maps $\Hom(\k,\gl(\k/\h))$.  Notice that if either $X$ or
$Y$ are in $\h$, then this is automatically satisfied: it is clear if
$Y \in \h$, since then $\Ybar = 0$, but then also if $X \in \h$,
$\Lambda_X = \lambda_X$ and the invariance condition reads
\begin{equation}
  0 = \lambda_Z \lambda_X \Ybar - \lambda_X
  \lambda_Z \Ybar - \lambda_{[Z,X]} \Ybar ,
\end{equation}
which holds because $\lambda$ is a representation.  So the only
nontrivial condition comes from $X,Y \not\in\h$.

Now choose $\m$ such that $\k = \h \oplus \m$ as vector spaces.  In
the reductive case, we can choose $\m$ such that $[\h,\m] \subset \m$
and hence $\m$ is an $\h$-module.  But even in the non-reductive case,
$\m \cong \k/\h$ as a vector space and there is a unique way to give
$\m$ the structure of an $\h$-module so that this is also an
isomorphism of $\h$-modules.  Let us assume we have done that.

Now let $\Lambda$ be an invariant affine connection.  Since the
restriction of $\Lambda$ to $\h$ is fixed, $\Lambda$ is determined by 
its component mapping $\m \to \gl(\k/\h)$ and, as argued above, it is
only this component that is involved in the invariance condition.  In
the reductive case, this component is an $\h$-equivariant linear map
$\m \to \gl(\m)$ or, equivalently, an $\h$-equivariant bilinear map
$\m \times \m \to \m$, called the \textbf{Nomizu map}.  In a reductive
homogeneous space, we may always take the Nomizu map to be zero and in
this way arrive at a \textbf{canonical invariant connection} (termed ``of the
second kind'' in \cite{MR0059050}).  All invariant connections are
then classified by their Nomizu maps.

In the non-reductive case, it can very well be the case that there are
no invariant connections.  Turning this around, if a given homogeneous
space does not admit any invariant connections, it cannot be
reductive.

Given an invariant affine connection, its \textbf{torsion} is given for all
$X,Y\in\k$,
\begin{equation}
  \Theta(X,Y) = \Lambda_X \Ybar - \Lambda_Y \Xbar - \overline{[X,Y]}~.
\end{equation}
One checks that this only depends on the images $\Xbar,\Ybar$ of $X,Y$
in $\k/\h$, so it defines an $\h$-equivariant skewsymmetric bilinear
map $\Theta : \k/\h \times \k/\h \to \k/\h$.

The \textbf{curvature} $\Omega$ is given, for all $X,Y \in \k$ and $\Zbar\in \k/\h$, by
\begin{equation}
  \Omega(X,Y) \Zbar = [\Lambda_X,\Lambda_Y]\Zbar -
  \Lambda_{[X,Y]} \Zbar,
\end{equation}
from where we see that it measures the failure of
$\Lambda: \k \to \gl(\k/\h)$ to be a Lie algebra homomorphism. The
$\h$-equivariance of $\Lambda$ guarantees that this expression only
depends on $\Xbar,\Ybar$ and hence it defines an $\h$-equivariant
skewsymmetric bilinear map
$\Omega: \k/\h \times \k/\h \to \gl(\k/\h)$. When the curvature
vanishes we say that the canonical connection is \textbf{flat}.

In the reductive case, the torsion and curvature of an invariant
affine connection are given in terms of its Nomizu map $\alpha : \m
\times \m \to \m$ by the following expressions for all $X,Y,Z \in \m$,
\begin{equation}
  \begin{split}
    \Theta(X,Y) &= \alpha(X,Y) - \alpha(Y,X) - [X,Y]_\m,\\
    \Omega(X,Y) Z &= \alpha(X,\alpha(Y,Z)) - \alpha(Y,\alpha(X,Z)) -
    \alpha([X,Y]_\m,Z) - [[X,Y]_\h, Z],
  \end{split}
\end{equation}
where $[X,Y] = [X,Y]_\h + [X,Y]_\m$ is the decomposition of $[X,Y] \in
\k = \h \oplus \m$.  These expressions simplify for the canonical
connection with zero Nomizu map:
\begin{equation}
    \Theta(X,Y) = - [X,Y]_\m \qquad\text{and}\qquad
    \Omega(X,Y) Z = - [[X,Y]_\h, Z].
\end{equation}
If the space is symmetric, so that $[X,Y]\in \h$ for $X,Y \in \m$,
then we see that the canonical connection is torsion-free.  If, on the
contrary, the canonical connection is flat, then $[X,Y] \in \m$
for all $X,Y\in\m$ and hence $\m$ is a Lie subalgebra of $\k$.  In
this case we can identify the homogeneous space (assumed
simply-connected) with the group manifold of the simply-connected Lie
group with Lie algebra $\m$.  If the canonical connection is both
torsion-free and flat, then the Lie algebra $\m$ is abelian and hence
the homogeneous space (if simply-connected) is the group manifold of a
simply-connected abelian Lie group.  A simply-connected abelian Lie
group is a vector space and hence the homogeneous space in that case
is an affine space.

The holonomy of the canonical connection on a reductive homogeneous
space $M = \Kgr/\Hgr$ can be calculated in general. Indeed, by a
result of Nomizu's \cite[§12]{MR0059050}, the Lie algebra of the
holonomy group is isomorphic to the ideal $[\m,\m]_{\h}$ of $\h$
acting on $\m$ via (the restriction of) the linear isotropy
representation. As a corollary, the covariant derivative of a
$\Kgr$-invariant tensor field $T$, say, on $M$ with respect to the
canonical connection vanishes: $\nabla T = 0$. If $M$ is
simply-connected, then parallel-transporting $T$ along any closed loop
leaves it invariant. In particular, the curvature and torsion tensor
fields (of any invariant connection, but in particular of the
canonical connection), being themselves invariant, are parallel, and
so are the tensor fields corresponding to any invariant lorentzian,
riemannian, galilean, carrollian or aristotelian structure that $M$
might possess. In particular the connection is compatible with the
(co\nobreakdash-)metrics. It is thus clear that the galilean
($\hyperlink{S7l}{\zG}$) and carrollian ($\hyperlink{S13l}{\zC}$)
spacetimes reproduce the standard flat Newton--Cartan and Carroll
structures, respectively (cf.~\cite{Duval:2014uoa}).

All spacetimes, except for $\hyperlink{S16l}{\zLC}$ are reductive. It
is then easy to inspect Tables~\ref{tab:spacetimes}
and~\ref{tab:aristotelian} and determine the torsion and curvature of
the canonical invariant connection with zero Nomizu map. The results
are summarised in Section~\ref{sec:invar-homog-spac} and in
Table~\ref{tab:spacetimes-props}. The existence of invariant
connections for the non-reductive spacetime $\hyperlink{S16l}{\zLC}$
has to be studied separately and we do so below. We will see that for
$D\geq 2$ the spacetime admits no invariant connections, whereas for
$D=1$ it admits a three-parameter family of invariant connections and
a unique torsion-free, flat connection. The follow-up paper
\cite{Figueroa-OFarrill:2019sex} presents the classification of invariant
connections and the calculation of their torsion and curvature for the
reductive kinematical and aristotelian spacetimes.

\subsubsection{Invariant connections for spacetime $\zLC$}
\label{sec:invar-conn-flc}

We will show that this homogeneous spacetime does not admit any
invariant connections for $D\geq 2$.  Since $[H,B_a] = B_a$, $[H,P_a]
= - P_a$ and $[B_a,P_b] = \delta_{ab}H + J_{ab}$, we have that
$\lambda_{B_a} \barH = 0$ and $\lambda_{B_a} \barP_b = \delta_{ab}
\barH$.

\paragraph{($D\geq 3$)}

The most general rotationally equivariant map $\Lambda$ is given by
\begin{equation}
  \Lambda_{H} \barH = \alpha \barH \qquad \Lambda_{H} \barP_a = \beta \barP_a \qquad
  \Lambda_{P_a} \barH = \gamma \barP_a \qquad\text{and}\qquad
  \Lambda_{P_a} \barP_b = \mu \delta_{ab} \barH + \nu \epsilon_{abc} \barP_c,
\end{equation}
with the tacit understanding that the term proportional to $\nu$ is
only present if $D=3$.  Invariance demands, in particular, that
\begin{equation}
  0 = \lambda_{B_a} \Lambda_{P_b} \barP_c - \Lambda_{[B_a,P_b]} \barP_c - \Lambda_{P_b}
  \lambda_{B_a} \barP_c,
\end{equation}
which becomes
\begin{equation}
  \begin{split}
    0 &= \nu \epsilon_{abc} \barH - \delta_{ab} \Lambda_{H} \barP_c -
    \Lambda_{J_{ab}} \barP_c - \delta_{ac} \Lambda_{P_b} \barH\\
    &= \nu \epsilon_{abc} \barH - \beta \delta_{ab} \barP_c -
    (\delta_{bc} \barP_a - \delta_{ac} \barP_b) - \gamma \delta_{ac} \barP_b\\
    &= \nu \epsilon_{abc} \barH - \beta \delta_{ab} \barP_c - \delta_{bc} \barP_a
    +  (1-\gamma) \delta_{ac} \barP_b.
  \end{split}
\end{equation}
Taking any $b=c\neq a$, we arrive at a contradiction.  Therefore there are
no invariant connections for $D\geq 3$.

\paragraph{($D=2$)}

Now the most general rotationally equivariant $\Lambda$ is given by
\begin{equation}
  \begin{aligned}\relax
    \Lambda_{H} \barH &= \alpha \barH\\
    \Lambda_{H}\barP_a &= \beta \barP_a + \beta' \epsilon_{ab} \barP_b
  \end{aligned}
  \qquad\qquad
  \begin{aligned}\relax
    \Lambda_{P_a} \barH &= \gamma \barP_a + \gamma' \epsilon_{ab} \barP_b\\
    \Lambda_{P_a} \barP_b &= \mu \delta_{ab} \barH + \mu' \epsilon_{ab} \barH.
  \end{aligned}
\end{equation}
But as before, equivariance requires, in particular, that
\begin{equation}
  0 = \lambda_{B_a} \Lambda_{P_b} \barP_c - \Lambda_{[B_a,P_b]}\barP_c - \Lambda_{P_b}
  \lambda_{B_a} \barP_c,
\end{equation}
which again for any $b=c\neq a$ results in a contradiction.  Therefore there are
no invariant connections for $D=2$ either.

\paragraph{($D=1$)}

Now there are no rotations, and $\Lambda$ is a general linear map:
\begin{equation}
  \begin{aligned}\relax
    \Lambda_{H} \barH &= \alpha \barH + \alpha' \barP\\
    \Lambda_{H}\barP &= \beta \barH + \beta' \barP
  \end{aligned}
  \qquad\qquad
  \begin{aligned}\relax
    \Lambda_{P} \barH &= \gamma \barH + \gamma' \barP\\
    \Lambda_{P} \barP &= \delta \barH + \delta' \barP.
  \end{aligned}
\end{equation}
Invariance under $B$ says that for all $X,Y \in \{H,P\}$,
\begin{equation}
  \lambda_B \Lambda_X \Ybar - \Lambda_{[B,X]} \Ybar - \Lambda_X
  \lambda_B \Ybar = 0.
\end{equation}
Taking $(X,Y) = (H,H),(H,P),(P,H),(P,P)$ in turn we arrive at the
following conditions:
\begin{equation}
  \alpha' = 0 \qquad \beta' = \alpha -1 \qquad \gamma' = \alpha \qquad
  \delta' = \beta + \gamma \qquad\text{and}\qquad \beta' + \gamma' = 0.
\end{equation}
This results in the following three-parameter family of invariant
connections
\begin{equation}
  \begin{aligned}\relax
    \Lambda_{H} \barH &= \tfrac12 \barH\\
    \Lambda_{H}\barP &= \beta \barH - \tfrac12 \barP
  \end{aligned}
  \qquad\qquad
  \begin{aligned}\relax
    \Lambda_{P} \barH &= \gamma \barH + \tfrac12 \barP\\
    \Lambda_{P} \barP &= \delta \barH + (\beta + \gamma) \barP.
  \end{aligned}
\end{equation}
Calculating the torsion and curvature, we find
\begin{equation}
  \Theta(H,P) = (\beta-\gamma) \barH \qquad \Omega(H,P) \barH = (\gamma +
  \tfrac12 \beta) \barH \qquad\text{and}\qquad \Omega(H,P)\barP = (2\delta
  + \beta^2) \barH + (\tfrac12 \beta + \gamma) \barP.
\end{equation}
We see that there is a unique flat and torsion-free connection, given
by
\begin{equation}
  \begin{aligned}\relax
    \Lambda_{H} \barH &= \tfrac12 \barH\\
    \Lambda_{H}\barP &= - \tfrac12 \barP
  \end{aligned}
  \qquad\qquad
  \begin{aligned}\relax
    \Lambda_{P} \barH &= \tfrac12 \barP\\
    \Lambda_{P} \barP &= 0.
  \end{aligned}
\end{equation}

\subsection{Summary of properties of homogeneous spacetimes}
\label{sec:invar-homog-spac}

We now summarise the properties of the simply-connected homogeneous
spacetimes in Tables~\ref{tab:spacetimes} and \ref{tab:aristotelian}.

There is precisely one non-reductive spacetime: $\zLC$
(\hyperlink{S16l}{\flc}), which we identified with the future light
cone in Minkowski spacetime one dimension higher.  It is carrollian,
but for $D>1$ does not admit any invariant connections. We determined
above the invariant connections when $D=1$.

The remaining spacetimes in Tables~\ref{tab:spacetimes} and
\ref{tab:aristotelian} are reductive and we proceed to list them
according to the type of reductive structure they possess.

\subsubsection{Flat symmetric spacetimes}
\label{sec:affine-spacetimes}

These are symmetric spacetimes where the canonical connection is flat ($\Omega=0$).
This means that the homogeneous space is a principal homogeneous space
for the translations.  If simply-connected, then it is an affine space
modelled on the vector space of translations.

\begin{itemize}
\item[(\mink)] Minkowski spacetime
\item[(\euc)] euclidean space
\item[(\gal)] galilean spacetime
\item[(\car)] carrollian spacetime
\item[(\st)] static aristotelian spacetime
\end{itemize}
and the exotic two-dimensional spacetimes: (\xone), (\xtwo),
(\xthree)$_{\chi>0}$ and (\xfour)$_{\chi>0}$.

\subsubsection{Non-flat symmetric spacetimes}
\label{sec:symmetric-non-flat}

\begin{itemize}
\item[(\ds)] de~Sitter spacetime, with curvature
  \begin{equation}
    \label{eq:dS-curv}
    \Omega(H,P_a) = \lambda_{B_a} \qquad\text{and}\qquad \Omega(P_a,P_b) = \lambda_{J_{ab}}.
  \end{equation}
  The notation is such that we interpret $\Omega$ as a two-form with values
  in endomorphisms of the tangent space, which for a homogeneous space
  localises to a linear map $\Omega: \Lambda^2\m \to \gl(\m)$ and we write
  it in terms of the image in $\gl(\m)$ of the linear isotropy
  representation $\lambda: \h \to \gl(\m)$.  For example,
  \begin{equation}
    \Omega(H,P_a)H = \lambda_{B_a}H = [B_a,H] = P_a \qquad\text{and}\qquad
    \Omega(H,P_a)P_b = \lambda_{B_a} P_b = \delta_{ab} H,
  \end{equation}
  et cetera.  Notice that the curvature 2-form of the canonical
  connection of de~Sitter spacetime does not see the radius of
  curvature.  This is because the canonical connection (and hence its
  curvature) is an invariant of the reductive homogeneous space,
  whereas the radius of curvature is an additional structure: namely,
  an invariant lorentzian metric.  The same happens with
  anti~de~Sitter spacetime, the round sphere and hyperbolic space.
  
\item[(\ads)] anti~de~Sitter spacetime, with curvature
  \begin{equation}
    \label{eq:AdS-curv}
    \Omega(H,P_a) = -\lambda_{B_a} \qquad\text{and}\qquad \Omega(P_a,P_b) = -\lambda_{J_{ab}},
  \end{equation}
  which is formally like for de~Sitter spacetime except for an overall
  sign.
  
\item[(\sph)] round sphere, with curvature formally identical to that
  in equation \eqref{eq:dS-curv}, except that of course, the linear
  isotropy representation $\lambda_{B_a}$ differs.

\item[(\hyp)] hyperbolic space, with curvature formally identical to
  that in equation~\eqref{eq:AdS-curv}, except that again the linear
  isotropy representation $\lambda_{B_a}$ differs.

\item[(\dsg)] galilean de~Sitter spacetime, with curvature
  \begin{equation}
    \Omega(H,P_a) = \lambda_{B_a}.
  \end{equation}

\item[(\adsg)] galilean anti~de~Sitter spacetime, with curvature
  \begin{equation}
    \Omega(H,P_a) = -\lambda_{B_a},
  \end{equation}
  which is again a sign off the one for galilean de~Sitter spacetime.
  
\item[(\dsc)] carrollian de~Sitter spacetime, whose curvature is
  formally identical to that in equation~\eqref{eq:dS-curv}, with the
  different action of $\lambda_{B_a}$.

\item[(\adsc)] carrollian anti~de~Sitter spacetime, whose curvature is
  formally identical to that in equation~\eqref{eq:AdS-curv}, with the
  different action of $\lambda_{B_a}$.

\item[(\athree)] aristotelian: $\RR \times \SS^D$ for $\varepsilon=-1$
  and $\RR\times \HH^D$ for $\varepsilon=1$, with curvature given by
  \begin{equation}
    \Omega(P_a,P_b) = -\varepsilon \lambda_{J_{ab}}.
  \end{equation}
\end{itemize}

\subsubsection{Reductive torsional spacetimes}
\label{sec:reduct-non-symm}

The canonical connection of these reductive spacetimes has torsion and
hence they are not symmetric spaces.  If the connection is flat, then
the spacetime is actually a principally homogeneous space for a Lie
group whose Lie algebra is isomorphic to $\m$, which is a Lie algebra
in the flat case.

\begin{itemize}
\item[(\tdsg)] This is $\hyperlink{S9l}{\mathsf{dSG}}_\gamma$ for $\gamma \in (-1,1]$.  The
  torsion and curvature of the canonical connection are given by
  \begin{equation}
    \Theta(H,P_a) = -(1+\gamma) P_a \qquad\text{and}\qquad \Omega(H,P_a) =
    -\gamma \lambda_{B_a}.
  \end{equation}
  We see that it is torsion-free if and only if $\gamma = -1$, which
  corresponds to the symmetric space (\dsg).  It is flat if and only
  if $\gamma = 0$.  It is then a principally homogeneous space for the
  simply-connected solvable Lie group with Lie algebra $[H,P_a] =
  P_a$.
\item[(\tadsg)] This is $\hyperlink{S11l}{\mathsf{AdSG}}_\chi$ for $\chi>0$.  The torsion and
  curvature of the canonical connection are given by
  \begin{equation}
    \Theta(H,P_a) = -2 \chi P_a \qquad\text{and}\qquad \Omega(H,P_a) =
    -(1+\chi^2) \lambda_{B_a},
  \end{equation}
  so that it is never flat, but it is torsion-free if and only if
  $\chi = 0$, corresponding to the symmetric space (\adsg).
\item[(\twodgal)] This is a two-parameter family of three-dimensional galilean
  spacetimes \twodgal$_{\gamma,\chi}$ with $\gamma \in [-1,1)$ and
  $\chi>0$.  The torsion and curvature are given by
  \begin{equation}
    \Theta(H,P_a) = -(1+\gamma) P_a + \chi \epsilon_{ab} P_b
    \qquad\text{and}\qquad
    \Omega(H,P_a) = - \gamma \lambda_{B_a} + \chi \epsilon_{ab} \lambda_{B_b},
  \end{equation}
  which is torsion-free if and only if $\gamma=-1$ and $\chi=0$, which
  corresponds to galilean de~Sitter spacetime (\dsg).  The connection
  is flat if and only if $\gamma = \chi = 0$, which corresponds to
  \tdsg$_{\gamma=0}$.
\item[(\tst)] This is an aristotelian non-symmetric space with torsion
  \begin{equation}
    \Theta(H,P_a) = - P_a
  \end{equation}
  and zero curvature.  It is a principally homogeneous space for the
  simply-connected solvable Lie group with Lie algebra $[H,P_a] = P_a$.
\item[(\twoda)] This is a three-dimensional aristotelian spacetime with torsion
  \begin{equation}
    \Theta(P_a,P_b) = - \epsilon_{ab} H
  \end{equation}
  and zero curvature.  It is a principally homogeneous space for the
  simply-connected Heisenberg Lie group with Lie algebra $[P_a,P_b] =
  \epsilon_{ab} H$.
\end{itemize}

\subsubsection{Summary}
\label{sec:summary-2}

In Table~\ref{tab:spacetimes-props} we summarise the basic properties
of the homogeneous kinematical spacetimes in
Table~\ref{tab:spacetimes} and aristotelian spacetimes in
Table~\ref{tab:aristotelian}. The first column is simply our label in
this paper, the second column specifies the value of $D$, where the
dimension of the spacetime is $D+1$. For the columns labeled ``R'',
``S'' and ``A'' we indicate with a $\checkmark$ when a spacetime is
reductive, symmetric and/or affine, respectively. The columns labelled
``L'', ``E'', ``G'' and ``C'' indicate the kind of invariant
structures the spacetime possesses: lorentzian, riemannian
(``euclidean''), galilean and carrollian,
respectively. Again a $\checkmark$ indicates that the spacetime
possesses that structure. The columns ``P'', ``T'' and ``PT'' indicate
whether the spacetime is invariant under parity, time reversal or
their combination, respectively, with $\checkmark$ signalling when
they are. The columns ``$\Omega$'' and ``$\Theta$'' tell us,
respectively, about the curvature and torsion of the canonical
invariant connection for the reductive spacetimes (that is, all but
S16). A $\neq 0$ indicates the presence of curvature and/or torsion.
Otherwise the connection is flat and/or torsion-free, respectively.
The final column contains any relevant comments, including, when
known, the name of the spacetime.

The table is divided into six sections.  The first four correspond to
lorentzian, euclidean, galilean and carrollian spacetimes.  The
fifth section contains two-dimensional spacetimes with no invariant
structure of these kinds.  The sixth and last section contains the
aristotelian spacetimes.

\begin{table}[h!]
  \centering
  \caption{Properties of simply-connected homogeneous spacetimes}
  \label{tab:spacetimes-props}
  \rowcolors{2}{blue!10}{white}
    \begin{tabular}{l|*{1}{>{$}l<{$}}|*{3}{>{$}l<{$}}|*{4}{>{$}l<{$}}|*{3}{>{$}l<{$}}|*{2}{>{$}c<{$}}|l} \toprule
      \multicolumn{1}{c|}{Label}                 & \multicolumn{1}{c|}{$D$} & R   & S   & A   & L   & E   & G   & C   & P   & T   & PT  & \Omega & \Theta & \multicolumn{1}{c}{Comments}              \\\midrule
      % lorentzian                                                                                                                              
      \hyperlink{S1l}{\mink}                     & \geq 1                   & \cm & \cm & \cm & \cm & \tm & \tm & \tm & \cm & \cm & \cm & \zro   & \zro   & Minkowski                                 \\
      \hyperlink{S2l}{\ds}                       & \geq 1                   & \cm & \cm & \tm & \cm & \tm & \tm & \tm & \cm & \cm & \cm & \neq0  & \zro   & de~Sitter                                 \\
      \hyperlink{S3l}{\ads}                      & \geq 1                   & \cm & \cm & \tm & \cm & \tm & \tm & \tm & \cm & \cm & \cm & \neq0  & \zro   & anti~de~Sitter                            \\\midrule
      % riemannian                                                                                                                              
      \hyperlink{S4l}{\euc}                      & \geq 1                   & \cm & \cm & \cm & \tm & \cm & \tm & \tm & \cm & \cm & \cm & \zro   & \zro   & euclidean                                 \\
      \hyperlink{S5l}{\sph}                      & \geq 1                   & \cm & \cm & \tm & \tm & \cm & \tm & \tm & \cm & \cm & \cm & \neq0  & \zro   & sphere                                    \\
      \hyperlink{S6l}{\hyp}                      & \geq 1                   & \cm & \cm & \tm & \tm & \cm & \tm & \tm & \cm & \cm & \cm & \neq0  & \zro   & hyperbolic                                \\\midrule
      % galilean
      \hyperlink{S7l}{\gal}                      & \geq 1                   & \cm & \cm & \cm & \tm & \tm & \cm & \tm & \cm & \cm & \cm & \zro   & \zro   & galilean                                  \\
      \hyperlink{S8l}{\dsg}                      & \geq 1                   & \cm & \cm & \tm & \tm & \tm & \cm & \tm & \cm & \cm & \cm & \neq0  & \zro   & galilean dS = $\ztdSG_{-1}$                \\
      \hyperlink{S9l}{\tdsg$_{\gamma\neq 0}$}    & \geq 1                   & \cm & \tm & \mn & \tm & \tm & \cm & \tm & \cm & \tm & \tm & \neq0  & \neq0  & $\ztdSG_\gamma$, $0\neq \gamma \in (-1,1]$ \\
      \hyperlink{S9l}{\tdsg$_0$}                 & \geq 1                   & \cm & \tm & \mn & \tm & \tm & \cm & \tm & \cm & \tm & \tm & \zro   & \neq0  & $\ztdSG_0$                                 \\
      \hyperlink{S10l}{\adsg}                    & \geq 1                   & \cm & \cm & \tm & \tm & \tm & \cm & \tm & \cm & \cm & \cm & \neq0  & \zro   & galilean AdS = $\ztAdSG_0$                 \\
      \hyperlink{S11l}{\tadsg$_\chi$}            & \geq 1                   & \cm & \tm & \mn & \tm & \tm & \cm & \tm & \cm & \tm & \tm & \neq0  & \neq0  & $\ztAdSG_\chi$, $\chi>0$                   \\
      \hyperlink{S12l}{\twodgal$_{\gamma,\chi}$} & 2                        & \cm & \tm & \mn & \tm & \tm & \cm & \tm & \cm & \tm & \tm & \neq0  & \neq0  & $\gamma\in [-1,1)$, $\chi>0$              \\\midrule
      % carrollian                                                                                                                              
      \hyperlink{S13l}{\car}                     & \geq 1                   & \cm & \cm & \cm & \tm & \tm & \tm & \cm & \cm & \cm & \cm & \zro   & \zro   & carrollian                                \\
      \hyperlink{S14l}{\dsc}                     & \geq 1                   & \cm & \cm & \tm & \tm & \tm & \tm & \cm & \cm & \cm & \cm & \neq0  & \zro   & carrollian dS                             \\
      \hyperlink{S15l}{\adsc}                    & \geq 1                   & \cm & \cm & \tm & \tm & \tm & \tm & \cm & \cm & \cm & \cm & \neq0  & \zro   & carrollian AdS                            \\
      \hyperlink{S16l}{\flc}                     & \geq 1                   & \tm & \mn & \mn & \tm & \tm & \tm & \cm & \cm & \tm & \tm & \mn    & \mn    & carrollian light cone                  \\\midrule
      % no structure at all                                                                                                                     
      \hyperlink{S17l}{\xone}                    & 1                        & \cm & \cm & \cm & \tm & \tm & \tm & \tm & \tm & \tm & \cm & \zro   & \zro   &                                           \\
      \hyperlink{S18l}{\xtwo}                    & 1                        & \cm & \cm & \cm & \tm & \tm & \tm & \tm & \tm & \cm & \tm & \zro   & \zro   &                                           \\
      \hyperlink{S19l}{\xthree$_\chi$}           & 1                        & \cm & \cm & \cm & \tm & \tm & \tm & \tm & \tm & \cm & \tm & \zro   & \zro   & $\chi>0$                                  \\
      \hyperlink{S20l}{\xfour$_\chi$}            & 1                        & \cm & \cm & \cm & \tm & \tm & \tm & \tm & \tm & \tm & \cm & \zro   & \zro   & $\chi>0$                                  \\\midrule
      % aristotelian                                                                                                                            
      \hyperlink{A21}{\st}                       & \geq 0                   & \cm & \cm & \cm & \cm & \cm & \cm & \cm & \cm & \cm & \cm & \zro   & \zro   & static                                    \\
      \hyperlink{A22}{\tst}                      & \geq 1                   & \cm & \tm & \mn & \cm & \cm & \cm & \cm & \cm & \tm & \tm & \zro   & \neq0  & torsional static                          \\
      \hyperlink{A23p}{\athree$_\varepsilon$}    & \geq 2                   & \cm & \cm & \tm & \cm & \cm & \cm & \cm & \cm & \cm & \cm & \neq0  & \zro   & $\varepsilon= \pm 1$                      \\
      \hyperlink{A24}{\twoda}                    & 2                        & \cm & \tm & \mn & \cm & \cm & \cm & \cm & \cm & \tm & \tm & \zro   & \neq0  &                                           \\ \bottomrule
    \end{tabular}
                                                                                                                                                                                                      \\[10pt]
    \caption*{This table describes if a spacetime of dimension $D+1$ is
      reductive (R), symmetric (S) or affine (A). A spacetime might
      exhibit a lorentzian (L), riemannian (E), galilean (G) or
      carrollian (C) structure, and be invariant under parity (P),
      time reversal (T) or their combination (PT).  Furthermore the
      canonical connection might be have curvature ($\Omega$) and/or
      torsion ($\Theta$). For the precise definitions of these
      properties see Sections~\ref{sec:basic-notions} and
      \ref{sec:invar-conn}.}
\end{table}

In particular, we see how all the spacetimes in
Figure~\ref{fig:state-of-prior-art} are indeed symmetric: with \mink\
($\hyperlink{S1l}{\mathbb{M}}$), \ds\ ($\hyperlink{S2l}{\mathsf{dS}}$)
and \ads\ ($\hyperlink{S3l}{\mathsf{AdS}}$) lorentzian; \euc\
($\hyperlink{S4l}{\mathbb{E}}$), \sph\ ($\hyperlink{S5l}{\mathbb{S}}$)
and \hyp\ ($\hyperlink{S6l}{\mathbb{H}}$) riemannian; \gal\
($\hyperlink{S7l}{\mathsf{G}}$), \dsg\
($\hyperlink{S8l}{\mathsf{dSG}}$) and \adsg\
($\hyperlink{S10l}{\mathsf{AdSG}}$) galilean; and \car\
($\hyperlink{S13l}{\mathsf{C}}$), \dsc\
($\hyperlink{S14l}{\mathsf{dSC}}$) and \adsc\
($\hyperlink{S15l}{\mathsf{AdSC}}$) carrollian. It is clear from a
dimension count that there can be no other lorentzian or riemannian
kinematical spacetimes than the ones in
Figure~\ref{fig:state-of-prior-art}. The dimension of the kinematical
group associated to a ($D+1$)-dimensional homogeneous spacetime is
given by $(D+1)(D+2)/2$, which is also the maximal dimension of the
isometry group of a ($D+1$)-dimensional pseudo-riemannian manifold, so
the homogeneous lorentzian and riemannian homogeneous spaces are
necessarily maximally symmetric. The perhaps remarkable fact is that
for $D\geq 2$, every homogeneous (kinematical) spacetime which is not
lorentzian or riemannian is either galilean, carrollian or
aristotelian. The new spacetimes in Figures~\ref{fig:generic-d-graph}
and \ref{fig:d=3-graph}, not already present in
Figure~\ref{fig:state-of-prior-art}, are therefore necessarily
galilean, carrollian or aristotelian.  We see that the class of
galilean spacetimes is particularly rich: admitting two
one-dimensional continua of such spacetimes: one, denoted
$\ztAdSG_\chi$ (\hyperlink{S11l}{\tadsg$_\chi$}), which extends
the galilean anti~de~Sitter spacetime $\zAdSG$
(\hyperlink{S10l}{\adsg}) and a second, denoted $\ztdSG_\gamma$
(\hyperlink{S10l}{\tdsg$_\gamma$}), which extends the galilean de~Sitter spacetime $\zdSG$
(\hyperlink{S8l}{\dsg}).  The carrollian spacetimes are all
realisable as null hypersurfaces in either Minkowski or (anti)
de~Sitter spacetimes one dimension higher.  In particular, the
carrollian spacetime $\zLC$ (\hyperlink{S16l}{\flc})
can be realised as the future light cone in Minkowski space one
dimension higher.  There are also several aristotelian spacetimes.
The situation in $D=2$ is even richer, with a  two-dimensional
continuum interpolating between the one-dimensional continua present
in all $D\geq 1$. Finally, in $D=1$ there are exotic (i.e., without
any discernible invariant structures) homogeneous spacetimes,
including two one-dimensional continua.

\section{Conclusions}
\label{sec:conclusions}

In this paper we have classified isomorphism classes of
simply-connected homogeneous spacetimes of kinematical and
aristotelian Lie groups with $D$-dimensional space isotropy for all
$D\geq 0$.  We have done this by classifying the corresponding
infinitesimal algebraic objects, namely (geometrically realisable,
effective) Lie pairs.  A number of observations follow from the
classification.

It follows from our classification (see, e.g.,
Table~\ref{tab:LAs-to-spacetimes}) that inequivalent spacetimes
may have the same transitive kinematical Lie algebra, which might have
interesting consequences (e.g., as in AdS/CFT).  Conversely,
non-isomorphic kinematical Lie algebras may have isomorphic
spacetimes.  For example, the para-galilean and static
kinematical Lie algebras lead to the same homogeneous spacetime: the
static aristotelian spacetime (\hyperlink{A21}{\st}).

The classification yields novel (at least to us) spacetimes:
particularly, the family of torsional galilean (anti) de~Sitter
spacetimes (\hyperlink{S9l}{\tdsg$_\gamma$} and
\hyperlink{S11l}{\tadsg$_\chi$}) and the torsional static aristotelian
spacetime (\hyperlink{A22}{\tst}), as well as the new families of two-
and three-dimensional spacetimes:
\hyperlink{S12l}{\twodgal$_{\gamma,\chi}$},
\hyperlink{S19l}{\xthree$_{\chi}$} and
\hyperlink{S20l}{\xfour$_{\chi}$}.
These novel spacetimes can be distinguished from the known ones (see
Figure~\ref{fig:state-of-prior-art}) in one of several (equivalent)
ways:
\begin{itemize}
\item they are not symmetric homogeneous spaces;
\item they are not invariant under both parity and time-reversal (at
  least for $D\geq 2$); and
\item they do not arise as limits of the maximally symmetric
  riemannian and lorentzian spaces.
\end{itemize}
In particular, this last characterisation allows us to see their
existence as a purely non-relativistic prediction.  Conversely, one can
ask whether there is a relativistic set-up that leads to these
spacetimes via limits.  None of these characterisations is compelling
reason to ignore the novel spacetimes.

We observed that not all aristotelian spacetimes arise from
kinematical Lie algebras and this motivated us to present the separate
classification of aristotelian Lie algebras in
Appendix~\ref{sec:class-arist-lie}.

We also explored some of the geometrical properties of the
spacetimes, particularly those which can be easily read from the
infinitesimal description: namely the existence of invariant
(pseudo\nobreakdash-)riemannian, galilean, carrollian and aristotelian
structures.  In the reductive cases, which are the vast
majority, we have paid particular attention to the torsion and
curvature of the canonical connection, as this provides an
identifiable invariant for the spacetime in question.

The main results are contained in Tables~\ref{tab:spacetimes},
\ref{tab:aristotelian} and \ref{tab:spacetimes-props} and their
interrelations are conveniently summarised in
Figures~\ref{fig:generic-d-graph}, \ref{fig:d=3-graph} and
\ref{fig:d=2-graph}.  In this paper we have restricted ourselves to
the classification of the simply-connected homogeneous spacetimes,
without paying very close attention to each of the geometries.  This
will be remedied in a follow-up paper \cite{Figueroa-OFarrill:2019sex}, where we
will revisit the classification and investigate the local geometry of
the spacetimes.
% In particular, it is shown in that paper that the
% boosts act generically with non-compact orbits in all homogeneous
% kinematical spacetimes with the obvious exceptions of the riemannian
% maximally symmetric spaces, where the ``boosts'' act as rotations.

\section*{Acknowledgments}

This work started during a visit of SP to Edinburgh and continued
during a visit of JMF to Brussels. It is our pleasure to thank the
corresponding institutions for hospitality and support. Additional
work was done during our participation at the MITP Topical Workshop
``Applied Newton--Cartan Geometry'' (APPNC2018), held at the Mainz
Institute for Theoretical Physics, where we presented a preliminary
version of our results. We are grateful to the MITP for their support
and hospitality and for providing such a stimulating research
atmosphere. We are particularly grateful to Eric Bergshoeff and Niels
Obers for the invitation to participate. JMF would like to thank
Andrea Santi for the invitation to visit Bologna and the opportunity
to talk about this work. SP is grateful to Eric Bergshoeff, Daniel
Grumiller, Joaquim Gomis, Marc Henneaux, Axel Kleinschmidt, Victor
Lekeu, Javier Matulich, Arash Ranjbar, Jan Rosseel, Jakob Salzer,
Friedrich Schöller, and David Tong for useful discussions. The
research of JMF is partially supported by the grant ST/L000458/1
``Particle Theory at the Higgs Centre'' from the UK Science and
Technology Facilities Council. The research of SP is partially
supported by the ERC Advanced Grant ``High-Spin-Grav" and by
FNRS-Belgium (convention FRFC PDR T.1025.14 and convention IISN
4.4503.15).

JMF would like to dedicate this paper to the memory of his friend
Andrew Ranicki, the late Professor of Algebraic Surgery at the
University of Edinburgh.  Andrew was a great practitioner in the art
of turning topological/geometrical problems into algebra, which is the
philosophy we have tried to follow in arriving at the results
described in this paper.

\appendix

\section{Classification of aristotelian Lie algebras}
\label{sec:class-arist-lie}

In this appendix we present the classification of aristotelian Lie
algebras.  In complete analogy with the definition of a kinematical
Lie algebra, we have the following.

\begin{definition}
  A real Lie algebra $\a$ is said to be \textbf{aristotelian} (with
  $D$-dimensional space isotropy) if it satisfies two properties:
  \begin{enumerate}
  \item $\a$ contains a Lie subalgebra $\r \cong \so(D)$, and
  \item $\a$ decomposes as $\a = \r \oplus V \oplus S$ under $\r$,
  \end{enumerate}
  where now we only have one copy of the vector representation of $\so(D)$.
\end{definition}

We can choose a basis $(J_{ab}, P_a, H)$ for $\a$, relative to which the
Lie brackets include the following:
\begin{equation}
  \label{eq:aristotelian-LA}
  \begin{split}
    [J_{ab}, J_{cd}] &= \delta_{bc} J_{ad} - \delta_{ac} J_{bd} - \delta_{bd} J_{ac} + \delta_{ad} J_{bc}\\
    [J_{ab}, P_c] &= \delta_{bc} P_a - \delta_{ac} P_b\\
    [J_{ab}, H ] &= 0,
  \end{split}
\end{equation}
and any other Lie brackets are subject only to the Jacobi identity,
which implies, in particular, equivariance under $\r$.  Every
aristotelian Lie algebra $\a$ gives rise to a unique aristotelian
homogeneous spacetime with effective Lie pair $(\a,\r)$.  Therefore
classifying aristotelian Lie algebras up to isomorphism also
classifies the aristotelian spacetimes.  Aristotelian Lie algebras all
share the Lie brackets \eqref{eq:aristotelian-LA} and are thus
distinguished by the $[H,P_a]$ and $[P_a,P_b]$ brackets, which are
only constrained by the Jacobi identity.

Many aristotelian Lie algebras arise as quotients of kinematical Lie
algebras by a vectorial ideal.  Indeed, we have seen in
Section~\ref{sec:class-kinem-spac} that kinematical Lie algebras
giving rise to non-effective Lie pairs always reduce to an 
aristotelian Lie algebra after quotienting by the ideal generated by
the boosts.  However, as we shall see below, not all aristotelian Lie
algebras arise in this way.

We now proceed to classify aristotelian Lie algebras, starting with
those in $D>3$ and making our way down in dimension.

Let $D>3$.  Equivariance under $\r \cong \so(D)$ forces
\begin{equation}
  [H,P_a] = \alpha P_a \qquad\text{and}\qquad [P_a,P_b] = \beta J_{ab},
\end{equation}
for some $\alpha,\beta \in \RR$.  The Jacobi identity says that
$\alpha\beta = 0$, giving rise to four isomorphism classes of
aristotelian Lie algebras:
\begin{enumerate}[label=(A\arabic*),start=1]
\item the static aristotelian Lie algebra ($\alpha = \beta = 0$);
\item $[H,P_a] = P_a$ and $[P_a,P_b] = 0$ ($\alpha \neq 0$, $\beta = 0$); and
\item $[H,P_a] = 0$ and $[P_a,P_b] = \varepsilon J_{ab}$, with
  $\varepsilon = \pm 1$ ($\alpha = 0$, $\beta \neq 0$).
\end{enumerate}

In $D=3$ equivariance under $\r \cong \so(3)$ allows a further term
\begin{equation}
  [H,P_a] = \alpha P_a \qquad\text{and}\qquad [P_a,P_b] = \beta J_{ab}
  + \gamma \epsilon_{abc} P_c,
\end{equation}
for some $\alpha,\beta,\gamma \in \RR$.  The Jacobi identity now says
that $\alpha \beta = \alpha \gamma = 0$.  But in $D=3$ we can change
basis to $P_a \mapsto P'_a = P_a + \lambda \epsilon_{abc} J_{bc}$ for
some $\lambda \in \RR$, apart from an overall scale.  Choosing
$\lambda = \tfrac14 \gamma$, we can assume that $\gamma = 0$ without
loss of generality.  In terms of the new basis, we are back to the
case $D>3$ with the same results:
\begin{enumerate}[label=(A\arabic*),start=1]
\item the static aristotelian Lie algebra ($\alpha = 0$, $\beta =
  \tfrac14 \gamma^2$);
\item $[H,P_a] = P_a$ and $[P_a,P_b] = 0$ ($\alpha \neq 0$, $\beta =
  \gamma = 0$); and
\item $[H,P_a] = 0$ and $[P_a,P_b] = \varepsilon J_{ab}$, with
  $\varepsilon = \pm 1$ ($\alpha = 0$, $\beta \neq \tfrac14 \gamma^2$).
\end{enumerate}
It is only for $D=3$ that the aristotelian Lie algebra A3 arises by
reduction from a kinematical Lie algebra. Indeed, since $J_{ab}$ can
be written as a Lie bracket of translations and since boosts transform
nontrivially under rotations, boosts cannot commute with translations
in the kinematical Lie algebra.  Since the boosts define an ideal,
there would have to be a nonzero bracket $[\B,\P] = \B$, whose
existence requires a nontrivial vector product invariant under
rotations, and this only exists in $D=3$.

Let $D=2$.  Now $\r \cong \so(2)$ is abelian, so equivariance under
$\r$ implies
\begin{equation}
  [H,P_a] = \alpha P_a + \delta \epsilon_{ab} P_b
  \qquad\text{and}\qquad [P_a,P_b] = \epsilon_{ab}(\beta J + \gamma H),
\end{equation}
where we have defined $J$ via $J_{ab} = -\epsilon_{ab} J$.  But now we
can change basis to $H \mapsto H' = H - \lambda J$ for some $\lambda
\in \RR$, apart from an overall scale.  Choosing $\lambda = \delta$,
we can assume with no loss of generality that $\delta = 0$.  In
general, the Jacobi identity says that $\alpha\beta = \alpha \gamma =
0$.  There is now an additional aristotelian Lie algebra:
\begin{enumerate}[label=(A\arabic*),start=1]
\item the static aristotelian Lie algebra ($\alpha = \beta = \gamma =
  0$);
\item $[H,P_a] = P_a$ and $[P_a,P_b] = 0$ ($\alpha \neq 0$, $\beta =
  \gamma = 0$);
\item $[H,P_a] = 0$ and $[P_a,P_b] = \varepsilon J_{ab}$, with
  $\varepsilon = \pm 1$ ($\alpha = 0$, $\beta \neq \gamma\delta$); and
\item $[P_a, P_b] = \epsilon_{ab} H$ ($\alpha = 0$, $\beta =
  \gamma\delta$, $\gamma \neq 0$).
\end{enumerate}
Here, because of the possibility of redefining $H$ and $J$, the
aristotelian Lie algebra A3 can arise by reduction of a kinematical
Lie algebra.

Let $D=1$.  Here there are no rotations, so any
two-dimensional Lie algebra is aristotelian.  Up to isomorphism there
are precisely two such Lie algebras:
\begin{enumerate}[label=(A\arabic*),start=1]
\item the static aristotelian Lie algebra; and
\item $[H,P] = P$.
\end{enumerate}

Finally, in $D=0$, there is only the one-dimensional Lie algebra
spanned by $H$, which is the $D=0$ avatar of the static aristotelian
Lie algebra A1.

In summary, the isomorphism classes of aristotelian Lie algebras with
$D$-dimensional space isotropy are recorded in
Table~\ref{tab:ALAs}.

\begin{table}[h!]
  \centering
  \caption{Aristotelian Lie algebras}
  \label{tab:ALAs}
  \rowcolors{2}{blue!10}{white}
  \begin{tabular}{l|>{$}c<{$}|*{2}{>{$}l<{$}}|l}\toprule
    \multicolumn{1}{c|}{A\#} & D & \multicolumn{2}{c|}{Nonzero Lie brackets in addition to $[\J,\J] = \J $, $[\J,\P] = \P$}& \multicolumn{1}{c}{Comments}\\\midrule
    \hypertarget{ALA1}{1} & \geq 0 & & & static \\
    \hypertarget{ALA2}{2} & \geq 1 & [H,P_a] = P_a & &\\
    \hypertarget{ALA3}{3$_\varepsilon$} & \geq 2 & & [P_a,P_b] = \varepsilon J_{ab} & $\varepsilon = \pm 1$\\
    \hypertarget{ALA4}{4} & 2 & & [P_a,P_b] = \epsilon_{ab} H & \\\bottomrule
  \end{tabular}
\end{table}

\section{Infinitesimal description of homogeneous spaces}
\label{sec:infin-descr-homog}

For psychological reasons, in this appendix $\g$ will denote a
finite-dimensional real Lie algebra and not the galilean Lie algebra
as in the bulk of the paper.

In this appendix we will prove that the classification of
simply-connected homogeneous spaces (up to isomorphism) is equivalent
to the classification of isomorphism classes of geometrically
realisable, effective Lie pairs.  This statement is the analogue in
the homogeneous space setting of the well-known fact that associated
to every finite-dimensional real Lie algebra $\g$ there exists a
unique (up to isomorphism) simply-connected Lie group whose Lie
algebra is isomorphic to $\g$.  Indeed, a Lie group is a principally
homogeneous space over itself, so that the group/algebra statement is
a special case of the general statement about homogeneous spaces. The
crucial difference is that given a Lie pair the corresponding
simply-connected homogeneous space need not exist nor be unique,
unless we impose additional conditions on the Lie pair: effective (for
uniqueness) and geometrically realisable (for existence).  It is
surprisingly difficult to find this more general statement in the
literature, but it is certainly standard and one can piece it together
from results in \cite{MR1631937}.

\subsection{Transitive actions of Lie groups}
\label{sec:trans-acti-lie}

Let $M$ be a connected smooth manifold.  By an \textbf{action} of a
Lie group $\Ggr$ on $M$, we mean a smooth map
$\alpha : \Ggr \times M \to M$ satisfying axioms which are easier to
state after we introduce the following notation.  If $g \in \Ggr$ and
$m \in M$, we will write $\alpha(g,m)$ as $g \cdot m$.  Then $\alpha$ is
an action if, for $e \in \Ggr$ the identity element, $e \cdot m = m$
for all $m \in M$, and if $g_1\cdot (g_2 \cdot m) = (g_1 g_2) \cdot
m$ for all $g_1,g_2 \in \Ggr$ and $m\in M$.

Let $\Ggr$ act on $M$ and let $\g$ be the Lie algebra of $\Ggr$.  The
action induces a Lie algebra (anti)homomorphism $\xi: \g \to \eX(M)$
from $\g$ to the Lie algebra of vector fields on $M$ sending every $X
\in \g$ to the \textbf{fundamental vector field} $\xi_X$ on $M$.  If
$m \in M$, then $\xi_X(m)$ is the velocity of the curve $\exp(t X)
\cdot m$ at $t=0$.

Let $\Ggr$ act on $M$ and let
$\Ngr = \left\{ g \in \Ggr \, \middle| \,  g \cdot m = m,~\forall m \in
  M\right\}$ denote the kernel of the action.  The action of $\Ggr$ on
$M$ is said to be \textbf{effective} if $\Ngr = \{e\}$ and it is said
to be \textbf{locally effective} if $\Ngr$ is a discrete group.  This
is equivalent to the map $\xi : \g \to \eX(M)$ being injective.  In
any case, $\Ngr$ is a normal subgroup of $\Ggr$ and the $\Ggr$ action
on $M$ induces an effective action of $\Ggr/\Ngr$.  Nevertheless,
although assuming that the action is effective seems to represent no
loss of generality, we will allow for locally effective actions.

Let $\Ggr$ act on $M$.  If given any two points $m_1,m_2 \in M$, there is
some $g \in \Ggr$ such that $m_2 = g \cdot m_1$, we say that the
action is \textbf{transitive}.  If the action is both transitive and
locally effective, then $M$ is said to be a \textbf{homogeneous space}
of $\Ggr$.  We will assume that $\Ggr$ is connected.  This represents
no loss of generality because if $\Ggr$ acts transitively on a
connected manifold $M$, so does the connected component of the
identity of $\Ggr$.

\begin{definition}
  Let $M$ and $M'$ be homogeneous spaces of $\Ggr$ and $\Ggr'$,
  respectively.  We say that $M$ and $M'$ are \textbf{isomorphic} if
  there is a diffeomorphism $f: M \to M'$ and a Lie group isomorphism $\Phi:
  \Ggr \to \Ggr'$ such that $f( g \cdot m) = \Phi(g) \cdot f(m)$ for
  all $m \in M$ and $g \in \Ggr$.
\end{definition}

Let $\Ggr$ be a connected Lie group acting transitively on $M$ and let
$m \in M$.  The set $\Ggr_m = \left\{ g \in \Ggr \, \middle | \, g \cdot m =
  m \right\}$ is a closed Lie subgroup of $\Ggr$ called the
\textbf{stabiliser} of $m$.  It need not be a connected subgroup.
Pick an ``origin'' $o \in M$ and let $\Hgr = \Ggr_o$.  Then $M$ is
diffeomorphic to the space $\Ggr/\Hgr$ of right $\Hgr$-cosets.  The
diffeomorphism $\iota: M \to \Ggr/\Hgr$ is such that $\iota(o) =
e\Hgr$ and if $m = g \cdot o$ then $\iota(m) = g\Hgr$, which is
well-defined because $\Hgr$ is the stabiliser of $o$.  It follows that
$\iota$ is $\Ggr$-equivariant: $\iota(g \cdot m) = g \iota(m)$ for all
$g \in \Ggr$ and $m \in M$.  In the language of the previous
definition, the homogeneous spaces $M$ and $\Ggr/\Hgr$ of $\Ggr$ are
isomorphic.  We say that $\Ggr/\Hgr$ is a \textbf{coset model} for
$M$.

If we change the origin, we get a different (but isomorphic) coset
model.  Indeed, let $o' \in M$ have stabiliser $\Hgr'$.  Then if $o' =
g \cdot o$, $\Hgr' = g \Hgr g^{-1}$ and in the language of the above
definition, $\Phi : \Ggr \to \Ggr$ is the inner automorphism
corresponding to conjugation by $g$ and $f : \Ggr/\Hgr \to \Ggr/\Hgr'$
is such that $f(k \Hgr) = g k g^{-1} \Hgr'$.

\subsection{Lie pairs}
\label{sec:lie-pairs}

Every coset space $\Ggr/\Hgr$ has a corresponding \textbf{Lie pair}
$(\g,\h)$, where $\g$ is the Lie algebra of $\Ggr$
and $\h$ is the Lie algebra of $\Hgr$, and so there is a way to assign
a Lie pair to a homogeneous space $M$ of $\Ggr$ \emph{and} a choice of
origin.  A different choice of origin results in a different Lie pair,
but how are they related?  Let $o,o'\in M$ be two choices of origin
with stabilisers $\Hgr$ and $\Hgr' = g \Hgr g^{-1}$, where $g \cdot o
= o'$.  Then the resulting Lie pairs are $(\g,\h)$ and $(\g,\h')$,
where $\h'$ is the Lie algebra of $\Hgr'$.  Let $\Ad_g : \g \to \g$ be
the inner automorphism of $\g$ induced by conjugation by $g$ in
$\Ggr$.  Then  $\h' = \Ad_g \h$, so the Lie pairs are related by an
inner automorphism.  This motivates the following definition.

\begin{definition}\label{def:iso-lie-pair}
  Two Lie pairs $(\g_1,\h_1)$ and $(\g_2,\h_2)$ are said to be
  \textbf{isomorphic} if there is a Lie algebra isomorphism $\varphi:
  \g_1 \to \g_2$ with $\varphi(\h_1) = \h_2$.
\end{definition}

This notion of isomorphism is stronger than what is needed in order to
classify homogeneous spaces, but it is easier to implement
algebraically.  In this Appendix we will show that it corresponds to
classifying homogeneous spaces up to coverings.  Equivalently, we will
see that to each isomorphism class of (certain) Lie pairs there
corresponds a unique simply connected homogeneous space (up to
isomorphism).

\begin{lemma}
  Let $M$ and $M'$ be homogeneous spaces of $\Ggr$ and $\Ggr'$, respectively.  If
  $M$ and $M'$ are isomorphic, then so are any Lie pairs associated to
  $M$ and $M'$.
\end{lemma}

\begin{proof}
  It is enough to show that any Lie pair associated to $M$ is
  isomorphic to at least one Lie pair associated to $M'$, since as we
  have seen above all Lie pairs associated to a homogeneous space are
  isomorphic (by an inner automorphism).  Since $M$ and $M'$ are
  isomorphic homogeneous spaces, we have an isomorphism of Lie groups
  $\Phi: \Ggr \to \Ggr'$ and a diffeomorphism $f: M \to M'$ obeying
  the equivariance property $f(g \cdot m) = \Phi(g) \cdot f(m)$ for
  all $m \in M$ and $g \in \Ggr$.  We will show that the Lie algebra isomorphism
  $\varphi: \g \to \g'$ induced by $\Phi$ is the desired isomorphism
  between the Lie pairs.

  So choose an origin $o \in M$ with stabiliser $\Hgr \subset \Ggr$,
  leading to the Lie pair $(\g,\h)$ and let $o' = f(o) \in M'$ have
  stabiliser $\Hgr'$, leading to the Lie pair $(\g',\h')$.  It follows
  from the equivariance property that if $g \in \Hgr$, then $\Phi(g)
  \in \Hgr'$:
  \begin{equation}
    \Phi(g) \cdot o' = \Phi(g) \cdot f(o) = f(g \cdot o) = f(o) = o'.
  \end{equation}
  But if $g' \in \Hgr'$, then the unique  $g \in
  \Ggr$ such that $g' = \Phi(g)$ lies in $\Hgr$:
  \begin{equation}
    f(o) = o' = g' \cdot o' = \Phi(g) \cdot f(o) = f(g \cdot o),
  \end{equation}
  but since $f$ is one-to-one, $g \cdot o = o$.  Therefore $\Phi(\Hgr)
  = \Hgr'$ and the Lie algebra isomorphism $\varphi: \g \to \g'$
  induced by $\Phi$ sends $\h$ isomorphically to $\h'$.
\end{proof}

It turns out that not all Lie pairs come from homogeneous spaces.

\begin{definition}
  A Lie pair $(\g,\h)$ is said to be \textbf{effective} if $\h$ does
  not contain a nonzero ideal of $\g$.
\end{definition}

It follows from this definition that if two Lie pairs are isomorphic
and one is effective, so is the other.  The following lemma justifies
the definition.

\begin{lemma}
  Let $M = \Ggr/\Hgr$ be a coset space with Lie pair $(\g,\h)$.  Then
  $(\g,\h)$ is effective if and only if the action of $\Ggr$ on $M$ is
  locally effective.
\end{lemma}

\begin{proof}
  We start by proving that if $(\g,\h)$ is not effective, then $\Ggr$
  does not act locally effectively.  If $(\g,\h)$ is not effective,
  then there is a nonzero ideal $\n$ of $\g$ contained in $\h$.  Let
  $\Ngr$ be the unique connected subgroup of $\Ggr$ generated by $\n$.
  Since $\n$ is an ideal, $\Ngr$ is a normal subgroup.  We claim that
  $\Ngr$ stabilises every point on $M$.  Since $\Ngr \subset \Hgr$, it
  stabilises any point $o \in M$ with stabiliser $\Hgr$.  Let $m \in
  M$ be any other point and let $g \in \Ggr$ be such that $g \cdot o =
  m$.  Then the stabiliser of $m$ is $g \Hgr g^{-1}$, which contains
  $g \Ngr g^{-1} = \Ngr$.

  Conversely, suppose that $\Ggr$ does not act locally effectively, so
  that the Lie algebra (anti)homomorphism $\xi: \g \to \eX(M)$ has
  nonzero kernel $\n$, which is an ideal of $\g$.  Let $o \in M$ have
  stabiliser $\Hgr$.  Then $\h$ consists of those $X \in \g$ for which
  $\xi_X(o) = 0$.  But if $X \in \n$, $\xi_X(m) = 0$ for all
  $m \in M$, so in particular, $\xi_X(o) = 0$ and hence $X \in \h$.
  This means that there is an ideal of $\g$ contained in $\h$ and
  hence $(\g,\h)$ is not effective.
\end{proof}

In summary, to a homogeneous space of $\Ggr$ and a choice of
origin, we may assign an effective Lie pair and up to isomorphism the
choice of origin is immaterial.  We now wish to examine the inverse
problem: namely, does every effective Lie pair arise as the Lie pair of a
homogeneous space and a choice of origin?

\subsection{Geometric realisations}
\label{sec:geom-real-1}

It turns out that not every effective Lie pair arises from a
homogeneous space.  For example, consider $\g = \su(3)$, the simple
Lie algebra of $3\times 3$ traceless skewhermitian complex matrices,
and let $\h$ be the one-dimensional subalgebra spanned by the matrix
\begin{equation}
  X_\alpha =
  \begin{pmatrix}
    i & \zero & \zero\\
    \zero & \alpha i & \zero\\
    \zero & \zero & -(1+\alpha)i
  \end{pmatrix}
\end{equation}
for some irrational real number $\alpha$.  We claim that there is no
Lie group $\Ggr$ with Lie algebra isomorphic to $\g$ for which the
subgroup corresponding to $\h$ is closed.  Indeed, there are (up to
isomorphism) precisely two connected Lie groups with Lie algebra
isomorphic to $\su(3)$: $\SU(3)$ itself and the adjoint group
$\Ad\,SU(3) \cong \SU(3)/\ZZ_3$.  The one-parameter subgroup of either of
these groups generated by $X_\alpha$ is not closed.  It is enough to
see this for the simply-connected group $\SU(3)$, since if the
subgroup of $\SU(3)/\ZZ_3$ generated by $X_\alpha$ were closed, then
so would be its pre-image under the covering homomorphism
$\pi: \SU(3) \to \SU(3)/\ZZ_3$, which is the subgroup generated by
$X_\alpha$ in $\SU(3)$.  So let $H_\alpha$ denote the subgroup
generated by $X_\alpha$ in $\SU(3)$.  It is clearly contained in the
maximal torus of diagonal matrices in $\SU(3)$, which is a closed
subgroup.  So it defines a one-parameter subgroup of the torus with an
irrational slope and it's easy to see that the closure of this
subgroup is the whole torus.\footnote{Since the counterexample here is
  the irrational slope flow on a torus, one might have wondered why we
  didn't simply consider the abelian Lie algebra $\g = \u(1)\oplus\u(1)$
  and the subalgebra $\h$ spanned by $(i,\alpha i)$, with $\alpha$
  irrational.  Indeed, the subgroup of $U(1)\times U(1)$ generated
  by $\h$ is not closed, but the subgroup generated by $\h$ in
  the universal covering group $\RR^2$ is closed, so that the Lie pair
  $(\g,\h)$ is geometrically realisable.}

This suggests the following definition.

\begin{definition}
  A Lie pair $(\g,\h)$ is \textbf{geometrically realisable} if there
  is a connected Lie group $\Ggr$ with Lie algebra
  $\mathrm{Lie}(\Ggr)$ isomorphic to $\g$ and a \emph{closed} Lie
  subgroup $\Hgr$ with Lie algebra $\mathrm{Lie}(\Hgr)$ isomorphic to
  $\h$ (by restricting the isomorphism $\mathrm{Lie}(\Ggr) \cong \g$).
  The coset space $\Ggr/\Hgr$ is then a \textbf{geometric realisation}
  of $(\g,\h)$.
\end{definition}

It is clear from this definition that if two Lie pairs are isomorphic
and one pair admits a geometric realisation then so does the other
pair.

\subsection{Simply-connected homogeneous spaces}
\label{sec:simply-conn-homog}

Finally, we are ready to prove the main result of this Appendix.
Namely, we show that every geometrically realisable, effective Lie
pair admits a unique (up to isomorphism) simply-connected geometric
realisation.

Let $M := \Ggr/\Hgr$ be a geometric realisation of the Lie pair
$(\g,\h)$, where $\Ggr$ is connected.  Let
$\pi: \widetilde\Ggr \to \Ggr$ be the universal covering group of
$\Ggr$.  Since $\pi$ is surjective, $\widetilde\Ggr$ also acts
transitively on $M$ via $g \cdot m = \pi(g) \cdot m$, for
$g \in \widetilde\Ggr$ and $m \in M$.  If $o \in M$ denotes the
identity coset, then its stabiliser in $\widetilde\Ggr$ is
$\widetilde\Hgr = \pi^{-1}\Hgr = \left\{h \in \widetilde\Ggr \, \middle| \,
  \pi(g) \in \Hgr\right\}$.  Therefore
$M = \widetilde\Ggr/\widetilde\Hgr = \Ggr/\Hgr$.  Now let
$\widetilde\Hgr_1$ denote the connected component of the identity in
$\widetilde\Hgr$ and let
$\widetilde M := \widetilde\Ggr/\widetilde\Hgr_1$.

\begin{lemma}
  $\widetilde M$ is the universal cover of $M$ and the covering map
  $p: \widetilde M \to M$ is $\widetilde\Ggr$-equivariant.
  Furthermore the Lie pair associated to $\widetilde M$ is isomorphic
  to $(\g,\h)$.
\end{lemma}

\begin{proof}
  It is clear that $\widetilde M$ is a homogeneous space of
  $\widetilde\Ggr$ and hence it is the base of a principal
  $\widetilde\Hgr_1$-bundle
  \begin{equation}
    \begin{tikzcd}
      \widetilde\Hgr_1 \arrow[r] & \widetilde\Ggr \arrow[d] \\
      & \widetilde M
    \end{tikzcd}
  \end{equation}
  whose homotopy long exact sequence ends with
  \begin{equation}
    \begin{tikzcd}
     \pi_1(\widetilde\Ggr) \arrow[r] & \pi_1(\widetilde M) \arrow[r] & 
     \pi_0(\widetilde\Hgr_1) \arrow[r] & \pi_0(\widetilde \Ggr)
     \arrow[r] & 0,
    \end{tikzcd}
  \end{equation}
  where all maps are group homomorphisms.  Since $\widetilde\Ggr$ is
  connected and simply connected,
  $\pi_0(\widetilde\Ggr) = \pi_1(\widetilde\Ggr) = 0$ and since
  $\widetilde\Hgr_1$ is connected, $\pi_0(\widetilde\Hgr_1) = 0$,
  resulting in $\pi_1(\widetilde M) = 0$.  The map
  $p:\widetilde M = \widetilde\Ggr/\widetilde\Hgr_1\to M =
  \widetilde\Ggr/\widetilde\Hgr$, defined by $p(g \widetilde\Hgr_1)
  = g \widetilde\Hgr$, is manifestly
  $\widetilde\Ggr$-equivariant and moreover it is a covering since it
  is the projection of a principal bundle with base $M$ and discrete
  fibre $\pi_0(\widetilde\Hgr)$.  The Lie pair associated to
  $\widetilde M$ is
  $(\mathrm{Lie}(\widetilde\Ggr),\mathrm{Lie}(\widetilde\Hgr_1))$, but
  since $\pi: \widetilde\Ggr \to \Ggr$ is a covering homomorphism of
  Lie groups, the Lie map $\pi_*: \mathrm{Lie}(\widetilde\Ggr) \to \g$
  is a Lie algebra isomorphism which restricts to an isomorphism
  $\mathrm{Lie}(\widetilde\Hgr_1) = \mathrm{Lie}(\widetilde\Hgr) \to \h$.
  Hence the Lie pairs
  $(\mathrm{Lie}(\widetilde\Ggr),\mathrm{Lie}(\widetilde\Hgr_1))$ and
  $(\g,\h)$ are isomorphic.
\end{proof}

We have shown that every geometrically realisable Lie pair $(\g,\h)$
has a simply-connected geometric realisation $\widetilde M$ as
above. It turns out that this is unique up to isomorphism.

\begin{lemma}\label{lem:uniqueness}
  $\widetilde M$ is the unique (up to isomorphism) simply-connected
  geometric realisation of $(\g,\h)$.
\end{lemma}

\begin{proof}
  Suppose that $\widetilde M'$ is another simply-connected geometric
  realisation of $(\g,\h)$.  This means that there is a connected Lie
  group $\Ggr'$ and a closed subgroup $\Hgr'$ such that
  $\widetilde M' = \Ggr'/\Hgr'$ and such that the Lie pairs
  $(\g',\h')$ and $(\g,\h)$ are isomorphic, where $\g'$ and $\h'$ are
  the Lie algebras of $\Ggr'$ and $\Hgr'$, respectively.  This also
  means that $(\g',\h')$ is isomorphic to the Lie pair
  $(\mathrm{Lie}(\widetilde\Ggr),\mathrm{Lie}(\widetilde\Hgr_1))$ of
  $\widetilde M$.  Let
  $\varphi: (\g',\h') \to
  (\mathrm{Lie}(\widetilde\Ggr),\mathrm{Lie}(\widetilde\Hgr_1))$
  denote this isomorphism.  Passing to the universal covering group
  (if necessary), we may assume without loss of generality that
  $\Ggr'$ is simply-connected and, since $\widetilde M' = \Ggr'/\Hgr'$
  is simply connected, that $\Hgr'$ is connected.  The isomorphism
  $\varphi: \g' \to \mathrm{Lie}(\widetilde\Ggr)$ lifts to a unique Lie
  group isomorphism $\Phi: \Ggr' \to \widetilde\Ggr$ which restricts
  to an isomorphism $\Hgr' \to \widetilde\Hgr_1$ and hence induces a
  unique isomorphism of homogeneous spaces $\phi: \widetilde M'=
  \Ggr'/\Hgr' \to \widetilde M = \widetilde\Ggr/\widetilde\Hgr_1$,
  sending $g \Hgr' \mapsto \Phi(g)\widetilde\Hgr_1$.
\end{proof}

As a corollary of the above lemma, we see that two homogeneous spaces
with isomorphic Lie pairs have isomorphic universal covers.  Indeed,
let $M = \Ggr/\Hgr$ and $M' = \Ggr'/\Hgr'$ have isomorphic Lie pairs
$(\g,\h)$ and $(\g',\h')$, respectively.  Then the universal cover
$\widetilde M$ of $M$ has a Lie pair which is isomorphic to $(\g,\h)$
and the universal cover $\widetilde M'$ of $M'$ has a Lie pair which
is isomorphic to $(\g',\h')$ and hence also to $(\g,\h)$.  Therefore
$\widetilde M$ and $\widetilde M'$ are simply-connected geometric
realisations of $(\g,\h)$ and by Lemma~\ref{lem:uniqueness} they are
isomorphic as homogeneous spaces.

In summary, we have proved the following.

\begin{theorem}
  Isomorphism classes of geometrically realisable, effective Lie pairs
  are in one-to-one correspondence with isomorphism classes of
  simply-connected homogeneous spaces.
\end{theorem}

We may paraphrase this result as follows.  Introduce an equivalence
relation between homogeneous spaces by declaring two homogeneous
spaces to be equivalent if their universal covers are isomorphic as
homogeneous spaces.  The isomorphism classes of geometrically
realisable, effective Lie pairs are in one-to-one correspondence with
equivalence classes of homogeneous spaces.

If we wish to classify homogeneous spaces up to isomorphism and not
just up to covering, we would start from the classification of
simply-connected homogeneous spaces and then classify their
homogeneous quotients.  That, however, is beyond the scope of this
paper.

% \bibliographystyle{utphys}
% \bibliography{ST}

\providecommand{\href}[2]{#2}\begingroup\raggedright\endgroup

\end{document}